\documentclass[acmsmall,10pt]{acmart}

\usepackage{amsmath}
\usepackage{latexsym}
\usepackage{proof}

\usepackage{amsfonts}
\usepackage{graphicx}
\usepackage{booktabs}
\usepackage{tabularx}
\PassOptionsToPackage{colorlinks=true,citecolor=blue,linkcolor=blue,urlcolor=blue}{hyperref}

\usepackage{color}
\usepackage{multirow}
\usepackage{wrapfig,lipsum,booktabs}
\usepackage{semantic}
\usepackage{isabelle,isabellesym}
\usepackage{subcaption}

\usepackage{listings}
\usepackage{xcolor}
\usepackage{adjustbox}
\usepackage{lineno}

\newcommand{\specrg}[1]{{\color{blue}{#1}}}
\newcommand{\speccomment}[1]{{\color{gray}{#1}}}
\definecolor{proofcolor}{rgb}{0,0.45,0}

\newcommand{\sectprefix}{{Section}}

\newcommand{\figprefix}{{Figure}}

\newcommand{\tableprefix}{{Table}}

\newcommand{\eqprefix}{{Equation}}
\newcommand{\lemmaprefix}{Lemma}

\newcommand{\defprefix}{Definition}
\newcommand{\appendixprefix}{Appendix}

\newcommand{\tabincell}[2]{\begin{tabular}{@{}#1@{}}#2\end{tabular}}

\usepackage{array}
\newcommand{\PreserveBackslash}[1]{\let\temp=\\#1\let\\=\temp}
\newcolumntype{C}[1]{>{\PreserveBackslash\centering}p{#1}}
\newcolumntype{R}[1]{>{\PreserveBackslash\raggedleft}p{#1}}
\newcolumntype{L}[1]{>{\PreserveBackslash\raggedright}p{#1}}

\newcommand{\zipfigaft}{\vspace{-5mm}}

\isabellestyle{sl} 

\let\endisabellecode=\endisabellecode

\newcommand{\implang}{\emph{IMP}}
\newcommand{\csimpllang}{\emph{CSimpl}}

\newcommand{\slang}{\emph{PiCore}} 

\newcommand{\pcimp}{$\pi${\implang}}
\newcommand{\pccsimpl}{$\pi${\csimpllang}}

\newcommand{\cmdfont}[1]{\textbf{\textup{#1}}} 

\newcommand{\cmdfinal}{\perp}

\newcommand{\evtdef}[1]{\widehat{#1}}
\newcommand{\eventdef}{\evtdef{(l,g,P)}}
\newcommand{\event}[1]{\cmdfont{Event}\ {#1}}
\newcommand{\anonevt}[1]{\lfloor{#1}\rfloor}
\newcommand{\atomevt}[1]{\cmdfont{AtomE}\ {#1}}

\newcommand{\evtent}[1]{\cmdfont{EvtEnt}\ {#1}}
\newcommand{\aevtent}[1]{\cmdfont{AEvtEnt}\ {#1}}

\newcommand{\evtseq}[2]{{#1} \rhd {#2}}
\newcommand{\symbchoice}{\oplus}
\newcommand{\evtchoice}[2]{{#1} \symbchoice {#2}}
\newcommand{\evtjoin}[2]{{#1} \Join {#2}}
\newcommand{\evtrec}[2]{\cmdfont{Iter}\ {#1}\ {#2}}

\newcommand{\parsys}[2]{{#1} \Rightarrow {#2}}
\newcommand{\parsysc}{\parsys{\symbCore}{\symbevtsys}}

\newcommand{\lchoice}[1]{\textbf{CHOICE} \ {#1}}

\newcommand{\stmtirq}[2]{{#1} \blacktriangleright {#2}}

\newcommand{\stmtawait}[2]{\textbf{AWAIT}\ {#1}\ \textbf{THEN}\ {#2}\ \textbf{END}}

\newcommand{\stmtevthead}[3]{\textbf{EVENT}\ {#1}\ [ {#2} ]\ \textbf{OF}\ {#3}}

\newcommand{\symbprog}{P}
\newcommand{\symbenv}{\Sigma}

\newcommand{\symbEvt}{\mathcal{E}}

\newcommand{\symbevtsys}{\mathcal{S}}
\newcommand{\symbpes}{\mathcal{PS}}

\newcommand{\symbstate}{s}

\newcommand{\symbact}{\xi}
\newcommand{\symbCore}{\mathcal{K}}
\newcommand{\symbcore}{\kappa}

\newcommand{\actk}[2]{{#1}\sphericalangle{#2}}

\newcommand{\symbspec}{\sharp}

\newcommand{\rgenvtran}{{\rightsquigarrow}}
\newcommand{\rgptran}{{\rightarrow_c}}

\newcommand{\envtran}[1]{{\rightsquigarrow_{#1}}}
\newcommand{\envtrane}{\envtran{e}}
\newcommand{\envtranpe}{\envtran{pe}}

\newcommand{\fin}{\anonevt{\cmdfinal}}
\newcommand{\etranlabel}{\delta}
\newcommand{\etranlabelc}{\tau}
\newcommand{\ptran}[4]{\symbenv \vdash (#1,#2)\rightarrow_p(#3,#4)}

\newcommand{\ptrans}[4]{\symbenv \vdash (#1,#2)\rightarrow_p^{*}(#3,#4)}
\newcommand{\estranconf}[3]{\symbenv \vdash #1 \stackrel{#2}{\rightarrow_e} #3}
\newcommand{\estran}[5]{\estranconf{(#1,#2)}{#3}{(#4,#5)}}

\newcommand{\pestran}[5]{\symbenv \vdash(#1,#2) \stackrel{#3}{\rightarrow_{pe}} (#4,#5)}

\newcommand{\estrannx}[5]{\symbenv \vdash(#1,#2) \stackrel{#3}{\rightarrow_{\tilde{e}}} (#4,#5)}

\newcommand{\esenvtranconf}[2]{\symbenv \vdash #1 \envtrane #2}
\newcommand{\esenvtran}[4]{\esenvtranconf{(#1,#2)}{(#3,#4)}}
\newcommand{\pesenvtran}[4]{\symbenv \vdash (#1,#2) \envtranpe (#3,#4)}
\newcommand{\anyenvtran}[4]{\symbenv \vdash (#1,#2) \envtran{\square} (#3,#4)}

\newcommand{\semanticrule}[2]{
\begin{tabular}{l}
  \textsc{\color{ACMRed}[#1]} \\
  {#2}
\end{tabular}}

\newcommand{\comprule}[2]{
\begin{tabular}{l}
  \textsc{\color{ACMBlue}[#1]} \\
  {#2}
\end{tabular}}

\newcommand{\proofrule}[2]{
\begin{tabular}{l}
  \textsc{\color{ACMDarkBlue}[#1]} \\
  {#2}
\end{tabular}}

\newcommand{\bpelsemrule}[2]{
\begin{tabular}{l}
  \textsc{\color{ACMRed}[#1]} \\
  {#2}
\end{tabular}}

\newcommand{\semrulename}[1]{\textsc{\color{ACMRed}{#1}}}
\newcommand{\comprulename}[1]{\textsc{\color{ACMBlue}{#1}}}
\newcommand{\proofrulename}[1]{\textsc{\color{ACMDarkBlue}{#1}}}
\newcommand{\bpelsemrulename}[1]{\textsc{\color{ACMRed}{#1}}}

\newcommand{\compfun}{\Psi}
\newcommand{\compfunm}{\Psi_M}
\newcommand{\symbcomp}{\varpi}

\newcommand{{\compstps}}{\Psi_\symbpes}
\newcommand{{\compstes}}{\Psi_\symbevtsys}
\newcommand{{\compste}}{\Psi_\symbEvt}
\newcommand{{\compstp}}{\Psi_\symbprog}
\newcommand{\assumefun}{\mathit{Assume}}
\newcommand{\commitfun}{\mathit{Commit}}
\newcommand{\rgcond}[4]{\langle #1, #2, #3, #4 \rangle}
\newcommand{\rgcondb}[4]{\langle \textbf{#1}, \textbf{#2}, \textbf{#3}, \textbf{#4} \rangle}
\newcommand{\rgconddefault}{\rgcond{pre}{rely}{guar}{post}}
\newcommand{\rgconddefaultb}{\rgcondb{pre}{rely}{guar}{post}}
\newcommand{\RGSAT}[2]{\models {#1} \ \mathbf{sat} \ {#2}}
\newcommand{\RGSATp}[2]{\symbenv \models {#1} \ \mathbf{sat}_{\cmdfont{p}} \ {#2}}

\newcommand{\RGSATe}[2]{\symbenv \models {#1} \ \mathbf{sat}_{\cmdfont{e}} \ {#2}}
\newcommand{\RGSATpe}[2]{\symbenv \models {#1} \ \mathbf{sat}_{\cmdfont{pe}} \ {#2}}

\newcommand{\rgsat}[2]{\vdash {#1}\ \cmdfont{sat} \ {#2}}
\newcommand{\rgsatp}[2]{\symbenv \vdash {#1}\ \cmdfont{sat}_{\cmdfont{p}}\ {#2}}
\newcommand{\rgsatpI}[2]{\symbenv \vdash_I {#1}\ \cmdfont{sat}_{\cmdfont{p}}\ {#2}}
\newcommand{\rgsate}[2]{\symbenv \vdash {#1}\ \cmdfont{sat}_{\cmdfont{e}}\ {#2}}
\newcommand{\rgsatpe}[2]{\symbenv \vdash {#1}\ \cmdfont{sat}_{\cmdfont{pe}}\ {#2}}

\newcommand{\idset}{\textbf{Id}}
\newcommand{\univset}{\textbf{UNIV}}
\newcommand{\stable}{\textbf{stable}}
\newcommand{\stablefun}[2]{\stable (#1,#2)}

\newcommand{\csimplprocenv}{\Gamma}
\newcommand{\csimplprocrg}{\Theta}
\newcommand{\csimpltran}[4]{\csimplprocenv \vdash_c (#1,#2) \rightarrow (#3,#4)}
\newcommand{\csimplenvtran}[4]{\csimplprocenv \vdash_c (#1,#2) \rightsquigarrow (#3,#4)}

\newcommand{\bpelfin}{\textbf{ActFin}}
\newcommand{\bpeltran}[4]{(#1,#2)\rightarrow_{bpel}(#3,#4)}
\newcommand{\bpelehtran}[4]{(#1,#2)\rightarrow_{eh}(#3,#4)}

\newcommand{\bisim}[4]{\symbenv \vdash ({#1},{#2}) \simeq ({#3},{#4})}

\newcommand{\isactrlenum}{$\blacktriangleright$}

\makeatletter
\newcommand{\superimpose}[2]{%
  {\ooalign{$#1\@firstoftwo#2$\cr\hfil$#1\@secondoftwo#2$\hfil\cr}}}
\makeatother

\def\BibTeX{{\rm B\kern-.05em{\sc i\kern-.025em b}\kern-.08emT\kern-.1667em\lower.7ex\hbox{E}\kern-.125emX}}

\copyrightyear{2018}
\acmYear{2018}
\setcopyright{acmlicensed}
\acmConference[Woodstock '18]{Woodstock '18: ACM Symposium on Neural Gaze Detection}{June 03--05, 2018}{Woodstock, NY}
\acmBooktitle{Woodstock '18: ACM Symposium on Neural Gaze Detection, June 03--05, 2018, Woodstock, NY}
\acmPrice{15.00}
\acmDOI{10.1145/1122445.1122456}
\acmISBN{978-1-4503-9999-9/18/06}

\setcopyright{acmcopyright}
\acmJournal{TOPLAS}
\acmYear{2020}\acmVolume{1}\acmNumber{1}\acmArticle{1}\acmMonth{1}
\acmDOI{10.1145/1122445.1122456}

\begin{document}

%
\title[Rely-guarantee Reasoning about Concurrent Reactive Systems]{Rely-guarantee Reasoning about Concurrent Reactive Systems}
\subtitle{The {\slang} Framework, Languages Integration and Applications}

\titlenote{This article is an extended version of \cite{Zhao19FM,Zhao19CAV}. \\
This work has been supported by the National Natural Science Foundation of China (NSFC) under the Grant No. 61872016, and the National Satellite of Excellence in Trustworthy Software Systems, funded by NRF Singapore under National Cyber-security R\&D (NCR) programme. 
}. 

%
\author{Yongwang Zhao}
\email{zhaoyw@zju.edu.cn}
\affiliation{%
  \institution{Zhejiang University}
  \department{School of Cyber Science and Technology, College of Computer Science and Technology}
  \streetaddress{38 Zheda Road}
  \city{Hangzhou}
  \state{Zhejiang}
  \country{China}
  \postcode{310007}
}


\author{David San\'{a}n}
\affiliation{%
  \institution{Nanyang Technological University}
  \department{School of Computer Science and Engineering}
  \streetaddress{50 Nanyang Avenue}
  \city{Singapore}
  \country{Singapore}
  \postcode{639798}
}


%
\renewcommand{\shortauthors}{Zhao Y.W., Sanan D.}

%
\begin{abstract}
The rely-guarantee approach is a promising way for compositional verification of concurrent reactive systems (CRSs), e.g. concurrent operating systems, interrupt-driven control systems and business process systems. However, specifications using heterogeneous reaction patterns, different abstraction levels, and the complexity of real-world CRSs are still challenging the rely-guarantee approach. 
This article proposes {\slang}, a rely-guarantee reasoning framework for formal specification and verification of CRSs. We design an event specification language supporting complex reaction structures and its rely-guarantee proof system to detach the specification and logic of reactive aspects of CRSs from event behaviours. {\slang} parametrizes the language and its rely-guarantee system for event behaviour using a rely-guarantee interface and allows to easily integrate 3rd-party languages via rely-guarantee adapters. By this design, we have successfully integrated two existing languages and their rely-guarantee proof systems without any change of their specification and proofs. {\slang} has been applied to two real-world case studies, i.e. formal verification of concurrent memory management in Zephyr RTOS and a verified translation for a standardized Business Process Execution Language (BPEL) to {\slang}. 

\end{abstract}

%
%


%

%

%
\maketitle

\section{Introduction}

Nowadays high-assurance systems are often designed as \textit{concurrent reactive systems} (CRSs) \cite{Aceto07}, which react to their computing environment by executing a sequence of commands under an input event. 
Examples of CRSs are concurrent operating systems (OSs), interrupt-driven control systems and business process systems, whose implementation follows a shared-variable event-driven paradigm. However, formal verification of shared-variable CRSs is still challenging \cite{Andronick17}. 

The rely-guarantee technique \cite{Jones83} represents a fundamental approach to compositional reasoning of \textit{concurrent programs} with shared variables, where imperative languages with extensions for concurrency describes the programs. 
The interference that occurs when concurrent programs access and update shared variables is explicitly specified by rely and guarantee conditions in the rely-guarantee method. This article concentrates on compositional reasoning about \emph{shared-variable CRSs} using the rely-guarantee technique. 

\subsection{Motivation and Challenges}

Reactive systems respond to continuous stimulus from their computing environment \cite{Harel85} by changing their state and, in turn, affecting their environment by sending back signals to it or initiating other operations. The inherent features of state-of-the-art CRSs are as follows. 
\begin{itemize}
\item CRSs may involve many different competitive agents executed concurrently with shared resources. Such concurrency might be due to multicore setting, embedded interrupts or task preemption, like in concurrent OS kernels \cite{Chen16,Xu16} and interrupt driven control systems, where the execution of handlers is not atomic. 
\item CRSs have various \emph{reaction structures}, such as \emph{peer reaction}, \emph{flow reaction} and \emph{hierarchical reaction}. Concurrent OS kernels behave as peer reactions, where kernel services behaves as reactors providing a fair scheduler to user calls. Modern business processes are regarded as flow reaction structures, which has a set of complex flow activities. An example of CRSs with a hierarchical reaction structure is Real-time Control Systems (RCSs) \cite{rcs97}, which is a layered reference model architecture for unmanned vehicles. 
\item The configuration and context of the underlying hardware of systems are not usually encoded in programs, which represent only a portion of the whole system behaviour. For instance, when and how interrupts are triggered, and which handlers are invoked to react with an interrupt, are out of the handler code when considering the models for interrupt handlers in OS kernels (e.g. kernel services and scheduling).
\end{itemize}

In such a setting of CRSs, we observe that the rely-guarantee approach faces the following challenges. 
\paragraph{Challenge 1: Tackling complex reaction structures}
Whilst rely-guarantee provides a general framework and can certainly be applied for CRSs, the languages in existing mechanizations of rely-guarantee (e.g. in \cite{Xu97,Nieto03,LiangFF12,Moreira13,Sanan17}) are imperative and designed only for pure programs, i.e, programs following a flow of procedure calls from an entry point. 
Examples of the reactive systems mentioned above are far more complex than pure programs because they involve many different agents as well as heavy interactions with their environment. 
Without statements for such system behaviour, we often use imperative programs to simulate them. For instance, in the setting of imperative languages, concurrent OS kernels in the \emph{peer reaction} structure are usually modelled as the parallel composition of reactive systems, each of which is simulated by a non-terminating loop program. The relevant handlers are invoked in the loop body adding the input arguments from user calls (e.g. in \cite{Andronick15}). 
First, the environment non-deterministically decides the event handler being triggered and its arguments. 
Second, some critical properties, such as noninterference of OS kernels \cite{Murray12}, concern traces of reactions rather than program states only. 
Without native support in the language semantics, \emph{while} loop programs have to use auxiliary logical/program variables to simulate the non-deterministic choice of handlers and parameters, and to capture the event context of each reactive system. 
This adds an additional level of complexity in programs and rely-guarantee conditions, in particular for realistic CRSs with many event handlers. 
The reason of these problems is the lack of a rely-guarantee approach for system reactions and, as a result, the mixture of system and program behaviours. 
\paragraph{Challenge 2: Reusing 3rd-party languages and their rely-guarantee proof systems}
Deep specification and verification is a trend in formal methods for both academy and industry \cite{deepspec,Klein14}. 
High assurance of critical systems demands that formal verification of functional correctness and safety/security properties is applied not only at the specification level \cite{Zhao16}, but also at the code level in the highest assurance levels \cite{Klein09,Dam13}. Verification of the specification can be conducted using abstract specification languages, but the verification of the implementation requires that modelling languages capture the features in programming languages such as exceptions and procedure calls. To formally verify CRSs at different levels, the rely-guarantee approach should be flexible enough to incorporate various 3rd-party languages and their rely-guarantee proof systems, and thus to be able to integrate and to reuse them. The challenge is to design a rely-guarantee framework to deal with interoperability of languages and bridge the gap between their rely-guarantee components. 
\paragraph{Challenge 3: Reasoning about real-world concurrent reactive systems.} 
Formal verification of real-world CRSs at low-level design and source code levels is challenging because of their complexity, being necessary to apply scalable methods. 
As a typical class of CRSs, formal verification of concurrent OSs is a challenge \cite{Klein14}, in particular the formal verification of the memory management. Data structures and algorithms of memory management in OSs are usually laid out in a complex manner to achieve a high performance. This complexity implies as well complex invariants that the formal model must preserve. In order to prevent memory leaks and block overlapping, these invariants have to guarantee the well-formedness and consistency of the structures, meanwhile keeping a precise reasoning to track both numerical and shape properties. Moreover, thread preemption and fine-grained locking often make the kernel execution of memory services to be concurrent. 

\subsection{Approach Overview}

In this article, {\slang} takes the level of abstraction and reusability of the rely-guarantee method a step further by proposing an expressive event language for complex reaction structures at the system level, as well as decoupling the system and program levels. The result is a flexible rely-guarantee framework for CRSs, which can integrate existing rely-guarantee implementations at program level without any change of them. {\slang} is organized in two levels:  
\begin{itemize}
\item The system level introduces the notion of ``events'' \cite{Back91,Abrial07} into the rely-guarantee method to modelling system reactions. This level defines the events composing a system, and how and when they are triggered. It specifies the language, semantics, and mechanisms to reason on sequences of events and their execution conditions. {\slang} provides a \emph{rely-guarantee interface} which is an abstraction for common rely-guarantee components for the program level. 
\item The program level focuses on the specification and reasoning of the behaviour of the events in the first level. {\slang} can easily integrate concrete languages used to model the behaviour of events by using a \emph{rely-guarantee adapter} which implements the rely-guarantee interface. This design allows {\slang} to be independent of program languages and thus to easily reuse existing rely-guarantee frameworks.
\end{itemize}

\begin{figure}
\begin{flushleft}
\begin{isabellec}
\footnotesize

$ \left .
\begin{aligned}
& \stmtevthead{mem\_pool\_alloc}{Ref\ pool, Nat\ size, Int\ timeout}{t}
\\
& \isacommand{WHEN} 
\\
& \quad pool\ {\isasymin}\ {\isasymacute}mem{\isacharunderscore}pools\ 
{\isasymand}\ timeout\ {\isasymge}\ {\isacharminus}{\isadigit{1}} 
\end{aligned}
\quad \right \}\  \textbf{Event declaration}
$

$ \left .
\begin{aligned}
& \isacommand{THEN} 
\\
& \quad ...... \text{\color{ACMGreen} // we omit some statements here. }
\\
& \quad \isacommand{IF}\ timeout\ {\isachargreater}\ {\isadigit{0}}\ \isacommand{THEN}
\\
& \quad \quad  {\isasymacute}endt\ {\isacharcolon}{\isacharequal}\ {\isasymacute}endt{\isacharparenleft}t\ {\isacharcolon}{\isacharequal}\ {\isasymacute}tick\ {\isacharplus}\ timeout{\isacharparenright}
\\
& \quad \isacommand{FI};; 
\\
& \quad ...... 
\\
& \quad \isacommand{FOR}\ {\isasymacute}from\_l := {\isasymacute}from\_l(t := {\isasymacute}free\_l\ t);
{\isasymacute}from\_l \ t < {\isasymacute}alloc\_l\ t; 
\\
& \quad \quad \quad {\isasymacute}from\_l := {\isasymacute}from\_l(t:={\isasymacute}from\_l\ t + 1)\ \isacommand{DO}
\\
& \quad \quad \isacommand{ATOMIC}
\\
& \quad \quad \quad {\isasymacute}mem\_pool\_info := set\_bit\_divide\ {\isasymacute}mem\_pool\_info\ p \ ({\isasymacute}from\_l\ t)\ ({\isasymacute}bn\ t);;
\\
& \quad \quad \quad ......
\\
& \quad \quad \isacommand{END}
\\
& \quad \isacommand{ROF}
\\
& \quad  ......
\\
& \isacommand{END}
\end{aligned}
\quad \right \}\  \textbf{Event body}
$

\end{isabellec}
\end{flushleft}
\caption{An Example of Event}
\label{fig:event_examp}
\end{figure}

At the system reaction level, we consider a reactive system as a set of event handlers called \emph{event systems} responding to stimulus from the environment. 
{\figprefix} \ref{fig:event_examp} illustrates an \emph{event} for memory allocation in OSs. The \emph{event declaration} has an event name, a list of input parameters, and a guard condition to determine the conditions triggering the event. In addition to the input parameters, an event has a special parameter $\textbf{OF}\ t$, which indicates the execution context, e.g. the thread invoking the service and the external devices triggering the interrupt. The behaviour of an event is specified as the \emph{event body} by using an imperative program described in a 3rd-party language. 

To deal with complex reaction structures, we design an event specification language in {\slang} supporting a structural compositions of  set events, e.g. \emph{join} and \emph{nondeterministic choice} that defines an event system. For instance, an event system composed of a nondeterministic choice of events models a peer reaction structure. 

The execution of an event system concerns the continuous evaluation of the guards of the events with their arguments. From the set of events for which their associated guard condition holds in an execution state, one event is non-deterministically selected to be triggered, and its body executed. After the event finishes, the guard evaluation starts again looking for the next event to be executed. We call the semantics of event systems \emph{reactive semantics}. 
A CRS is modelled as the \emph{parallel composition} of event systems interleaving their execution. 
Different from other modelling approaches (e.g., Event-B \cite{Abrial07}) where the execution of events are atomic, we consider the interleaved semantics of events and relax the atomicity of them in {\slang}.


CRSs representation in PiCore presents the following advantages: 
\begin{itemize}
\item \textbf{Expressivity}: the event language in {\slang} is able to model complex reaction structures of CRSs. We have built two different cases studies showing how does {\slang} specify a peer and a flow reaction systems.
\item \textbf{Reusability}: the {\slang} framework can integrate and thus reuse 3rd-party languages and their rely-guarantee frameworks. In this work we describe events on two different imperative languages by using the PiCore integration interface. 
\item \textbf{Simplicity}: compared to using imperative languages, the event language in {\slang} can considerably simplify the formal specification of CRSs. We formalize the Zhephyr RTOS memory management to show how complex concurrent reactive systems can be elegantly and simply described with {\slang}. 
\end{itemize}

At the time of writing this article, {\slang} supports the verification of total correctness of events, whose execution is usually assumed to be terminating, as well as the properties of event systems, whose execution is often non-terminating. 
In this article, we consider the verification of two different classes of properties for CRSs: pre and post conditions and invariants in the fine-grained execution. Invariants are used on the verification of safety properties concerning the small steps inside events, which must preserve such invariants. 
For instance, in the case of Zephyr RTOS, one safety property is that memory blocks do not overlap each other even during the internal steps of the \emph{alloc} and \emph{free} services. 
The execution trace of events can define other critical properties like noninterference \cite{Murray12,Mantel11,Murray16}, which is part of our ongoing work. 

Following the above principles, the {\slang} architecture is shown in {\figprefix} \ref{fig:approach}, where the lines represent the \emph{uses} relation of blocks. We use Isabelle/HOL theorem prover as the specification and verification system. The {\slang} framework use Isabelle's \emph{locales} to define the rely-guarantee interface, which is a set of common rely-guarantee components that existing languages may respect. 
For the integration of external languages, we adopt the \emph{adapter} design pattern \cite{Gamma95}. By implementing a \emph{rely-guarantee adapter} that respects the rely-guarantee interface, 3rd-party languages and their rely-guarantee proof systems are integrated into {\slang} by instantiating the {\slang} framework. 

We have integrated two existing languages ({\implang}\cite{Nieto03} and {\csimpllang}\cite{Sanan17}) and their rely-guarantee proof systems into the {\slang} framework. As a result we create two instances of {\slang}: {\pcimp} and {\pccsimpl}. {\implang} is a simple imperative language in Isabelle/HOL, which follows the specification in \cite{Xu97}. 
{\csimpllang} is a generic and expressive imperative language designed for modelling real world concurrent languages. {\csimpllang} extends \emph{Simpl}~\cite{Sch06} with concurrency, and provides a rely-guarantee proof system in Isabelle/HOL. \emph{Simpl} is composed of the necessary constructs to capture most of the features present in common sequential languages, such as conditional branching, loops, abrupt termination and exceptions, assertions, mutually recursive functions, expressions with side effects, and nondeterminism. Additionally, \emph{Simpl} can express memory-related features like the memory heap, pointers, and pointers to functions. 
\emph{Simpl} is able to represent a large subset of C99 code and has been applied to the formal verification of seL4 OS kernel \cite{Klein09} at C code level.

\begin{figure}
\centering
\includegraphics[width=4.0in]{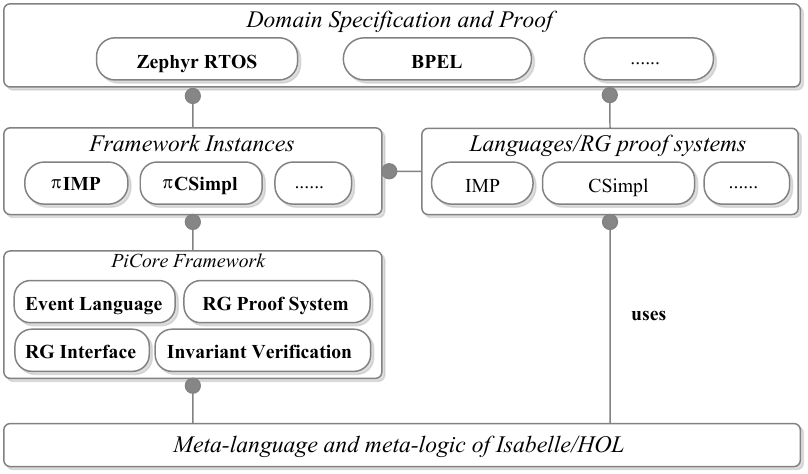}
\caption{The {\slang} Approach}
\label{fig:approach}
\end{figure}

Finally, we apply the {\slang} framework in two case studies. As a peer reaction case study, we model in {\slang} the memory management of Zephyr OS~\cite{zephyr}, a real-world concurrent RTOS. We use the {\slang} proof system to prove complex properties about the Zephyr concurrent memory management correctness and partial termination. As a case of flow reactions, we model the BPEL language~\cite{bpel} in {\slang}, a standard of business process execution language, showing soundness and completeness of the model. 

\paragraph{Zephyr RTOS}
Considering an OS Kernel as a reactive system with a \emph{peer reaction} structure, the OS is in an \emph{idle} loop until it receives an interruption which is handled by an interruption handler. 
Whilst interrupt handlers execution is atomic in sequential kernels, they can be interrupted in concurrent kernels \cite{Chen16,Xu16} allowing services invoked by threads to be interrupted and resumed later. We have applied {\slang} to the formal specification and mechanized proof of the concurrent buddy memory allocation of Zephyr RTOS. The formal specification is fine-grained, providing a high level of detail. It closely follows the Zephyr C code, covering all the data structures and imperative statements present in the implementation. 
We use the rely-guarantee proof system of {\slang} for the formal verification of functional correctness and invariant preservation in the model, revealing three bugs in the C code. 

\paragraph{BPEL}
Compared to OSs, BPEL programs are CRSs with a \emph{flow reaction} structure. BPEL aims to model the behaviour of processes via a language for the specification of both \emph{executable} and \emph{abstract} business processes. It extends the Web Services interaction model and enables it to support business transactions. BPEL has complex basic and structural activities, such as \emph{receive}, \emph{reply}, \emph{invoke}, \emph{sequence}, \emph{while}, \emph{repeatuntil}, \emph{flow} (parallel). We have applied {\slang} to interpret the semantics of the BPEL language by translating BPEL into {\slang}. To show the correctness of this translation, we prove a strong bisimulation between the source BPEL program and the translated {\slang} specification. In this way, formal verification of BPEL programs can be conducted in the {\slang} framework. The strong bisimulation implies the soundness and completeness of the formal verification of BPEL programs in {\slang}.

\subsection{Contributions}
This work is entirely mechanized and proved in Isabelle/HOL; the sources are available at {\color{ACMBlue}\url{https://lvpgroup.github.io/picorebook/}}. 
In summary, this article makes the following new contributions. 
\begin{enumerate}
\item We propose a novel event specification language and its rely-guarantee proof system for CRSs. To the best of our knowledge, {\slang} is the first rely-guarantee approach supporting complex reaction structures of CRSs. 
\item {\slang} is a parametric rely-guarantee framework for CRSs. By using rely-guarantee interfaces and adapters, it allows reusing existing languages and their rely-guarantee proof systems. 
To the best of our knowledge, {\slang} is the first flexible and open rely-guarantee framework for 3rd-party languages. 
\item We have integrated two existing languages and their rely-guarantee proof systems into the {\slang} framework, and thus created two instances of {\slang}. These instances can formally specify and verify CRSs at different abstraction levels. 
\item As a case study of concurrent OSs, we have applied {\slang} to the formal specification and mechanized proof of the concurrent buddy memory allocation of Zephyr RTOS. During the formal verification of functional correctness and invariant preservation, we reveal three bugs in the C code. To the best of our knowledge, this work is the first formal verification of a concurrent buddy memory allocation from a real-world operating system. 
\item As a case study of concurrent business processes, we have applied {\slang} to the formalization of the BPEL standard by an automated translator. We use a bisimulation relation to prove the correctness of this translation. To the best of our knowledge, this work presents the first verified translation of BPEL to a formal reasoning system. 
\end{enumerate}

We partially present contributions (1) - (3) in~\cite{Zhao19FM}, where we only supported modelling CRSs with a \emph{peer reaction} structure. In this article, we have completely redesigned the event language and its proof system, supporting complex reaction structures. The semantics of the event language has also been extended by the modular definition of computations to considerably simplify the soundness proof of rely-guarantee rules. Also, the current design of {\slang} has a better parameter support in the creation of the events. The rely-guarantee interface of {\slang} has been also simplified and made clearer. Contribution (4) was partially presented in \cite{Zhao19CAV}. In this article, we use the new event language to revise the formal specification of Zephyr RTOS. 

\subsection{Organization}

The rest of this article is organized as follows. We first introduce the rely-guarantee method and relevant Isabelle/HOL notation for this article in {\sectprefix} \ref{sect:preliminary}. {\sectprefix} \ref{sect:picore} presents the {\slang} rely-guarantee framework for CRSs. {\sectprefix} \ref{sect:langintegrate} presents the rely-guarantee interface and adapter. {\sectprefix} \ref{sect:integrate} shows how we integrate the two 3rd-party languages, {\implang} and {\csimpllang}, and their rely-guarantee proof systems. 
Next, we present the two case studies in {\sectprefix}s \ref{sect:mem} and \ref{sect:bpel} respectively. 
We evaluate our work and compare the related work in {\sectprefix} \ref{sect:eval}. Finally, we conclude this article in {\sectprefix} \ref{sect:conclusion}. 

\section{Preliminaries}
\label{sect:preliminary}

In this section, we briefly introduce the rely-guarantee method as well as Isabelle/HOL notations used in languages integration and case studies. 

\subsection{Rely-guarantee Reasoning}
In the rely-guarantee method \cite{Xu97,Nieto03}, executions of concurrent programs are considered as \emph{computations}, i.e. traces of configurations which are a pair of a program and a state. The rely-guarantee method distinguishes \emph{environment} steps $\rgenvtran$ from \emph{component} steps $\rgptran$ in computations. 
A rely-guarantee specification for a program is a quadruple $\rgconddefault$, where $\mathit{pre}$ is the precondition, $\mathit{rely}$ is the rely condition, $\mathit{guar}$ is the guarantee condition, and $\mathit{post}$ is the postcondition.  The designer of a concurrent program can assume that if
\begin{itemize}
\item the initial state satisfies the precondition \textbf{pre}, and
\item any environment step satisfies the rely condition \textbf{rely},
\end{itemize}
then the program makes the following commitments:
\begin{itemize}
\item any program step satisfies the guarantee condition \textbf{guar}, and
\item if the program terminates, the final state satisfies the postcondition \textbf{post}.
\end{itemize}
In the case of total correctness, the program ensures the termination as well.
The validity of a rely-guarantee condition is illustrated as follows.
\begin{equation}
\label{eq:rg_validity}
\overbrace{(\symbprog_0,s_0)}^{s_0 \in \textbf{pre}} \cdots \overbrace{(\symbprog_i,s_i) \rgenvtran (\symbprog_{i+1},s_{i+1})}^{(s_i,s_{i+1}) \in \textbf{rely}}  \cdots \underbrace{(\symbprog_j,s_j) \rgptran (\symbprog_{j+1},s_{j+1})}_{(s_j,s_{j+1}) \in \textbf{guar}} \cdots \underbrace{(\symbprog_n,s_n)}_{s_n \in \textbf{post}}
\end{equation}

The assumption, commitment and the validity of a specification are defined as it follows, where $\symbcomp_i$ is the $i$th configuration in a computation $\symbcomp$, $\symbstate_{\symbcomp_i}$ is the state part of the configuration, and $\symbprog_{\symbcomp_i}$ is the program of the configuration. $\compfun(\symbprog, \symbstate)$ is the set of all computations of the $\symbprog$ starting from an initial state $\symbstate$.
\begin{equation*}
\centering
\begin{aligned}
  & \assumefun(\mathit{pre}, \mathit{rely}) \equiv \{\symbcomp \mid \symbstate_{\symbcomp_0} \in \mathit{pre} \wedge (\forall i < \mathit{len}(\symbcomp) - 1. \ (\symbcomp_i \rgenvtran \symbcomp_{i+1}) \longrightarrow (\symbstate_{\symbcomp_i},\symbstate_{\symbcomp_{i+1}}) \in \mathit{rely})\} \\
  & \commitfun(\mathit{guar}, \mathit{post}) \equiv \{\symbcomp \mid \ (\forall i < \mathit{len}(\symbcomp) - 1. \ (\symbcomp_i \rgptran \symbcomp_{i+1}) \longrightarrow (\symbstate_{\symbcomp_i},\symbstate_{\symbcomp_{i+1}}) \in \mathit{guar}) \\
& \quad \quad \quad \quad \quad \quad \quad \quad \quad \quad \quad
  \wedge (\symbprog_{\mathit{last}(\symbcomp)} = \cmdfinal \longrightarrow \symbstate_{\symbcomp_n} \in \mathit{post})\} \\
  & \RGSAT{\symbprog}{\rgconddefault} \equiv \forall \symbstate. \ \compfun(\symbprog, \symbstate) \cap \assumefun(\mathit{pre}, \mathit{rely}) \subseteq \commitfun(\mathit{guar}, \mathit{post})
\end{aligned}
\end{equation*}

A set of rules inductively define the derivable correctness formulas $\rgsat{\symbprog}{\rgconddefault}$.  The verification of a rely-guarantee specification on a concurrent program is thus carried out by inductively applying the rules for each language construct and discharging the proof obligations the rules generate. The proof rule for parallel composition is typically defined as
\begin{equation*}
\infer{\rgsat{\symbprog_1 \parallel \symbprog_2}{\rgconddefault}}
{
\begin{aligned}
\rgsat{\symbprog_1}{\rgcond{pre_1}{rely_1}{guar_1}{post_1}} \quad \rgsat{\symbprog_2}{\rgcond{pre_2}{rely_2}{guar_2}{post_2}} \quad pre \subseteq pre_1 \cap pre_2\\
rely \cup guar_2 \subseteq rely_1 \quad rely \cup guar_1 \subseteq rely_2 \quad guar_1 \cup guar_2 \subseteq guar \quad post_1 \cap post_2 \subseteq post
\end{aligned}
}
\end{equation*}

Since valid proofs depend on a sound proof system, the soundness of the inference rules with regards to the validity of a rely-guarantee specification is proved by
\[\rgsat{\symbprog}{\rgconddefault} \Longrightarrow \RGSAT{\symbprog}{\rgconddefault}\]

\subsection{Isabelle/HOL Notations}
This subsection covers the Isabelle/HOL notations used later and their semantics. 
Isabelle/HOL is a generic and tactic-based theorem prover that uses higher order logic. 

The keyword \textbf{datatype} is used to define an inductive data type. 
For instance, the list in Isabelle is one of the essential data type in computing and is defined as 
$$\textbf{datatype} \ 'a \ list = Nil \ (``[]") \ | \ Cons \ 'a \ ``\ 'a \ list \ " (\textbf{infixr} \ ``\#")$$
where $'a$ is the type parameter for the type of elements in the list, $[ ]$ is the empty list and $\#$ is the \emph{append} operation. 
The polymorphic option type is defined as $$\textbf{datatype} \ 'a \ option = None \ | \ Some \ (the: \ 'a)$$
$Some\ (e :: 'a)$ represents defined elements, and $None$ those that are not defined. The value of a defined element is obtained using operator $the$, with $the\ (Some\ e) = e$. 

Isabelle/HOL allows the definition of types as records. An example of a record type is $$\textbf{record} \ point = xcoord::int \ \ ycoord::int$$ where an instance $p$ of type \emph{point} is defined as $p = {\isasymlparr} xcoord = 10, ycoord = 10 {\isasymrparr}$.
Isabelle/HOL provides operations for field projection, update, and extension. For instance, $xcoord \ p$ and $ycoord \ p$ are the projections of the fields $xcoord$ and $ycoord$ in $p$; $p{\isasymlparr} xcoord := 20, ycoord := 20 {\isasymrparr}$ updates $p$; and $\textbf{record} \ cpoint = point + col :: color$ extends $p$ with an additional field $col$ of type $color$. 
Tuples in Isabelle/HOL are defined as types of ($'a \times 'b$). They provide two functions to select the elements of the tuple: $fst :: ('a \times 'b) \Rightarrow 'a$ and $snd :: ('a \times 'b) \Rightarrow 'b$.

A total functionf $fun1$ is are defined as $fun1 :: 'a \Rightarrow 'b$, i.e. a mapping from type $'a$ to type $'b$. 
Given a function $f$ of type $'a \Rightarrow 'b$, the update of an element $a::'a$ in the domain of $f$ to the value $b::'b$ is represented by $f(a := b)$. This returns a new function $f'$ in which $f'\ x = f\ x$ if $x \neq a$ and $f'\ a = b$. 
A function of type $'a \rightharpoonup \ 'b$ models a partial function, which is equal to $'a \Rightarrow ('b \ option)$. 
Nonrecursive definitions can be made with the \textbf{definition} command and the \textbf{primrec} function definition is used for primitive recursions. 


Isabelle's deals with parametric theories by using \emph{locales} \cite{locale}.  A locale can be seen as a formula schema
\[
\bigwedge x_1...x_n. \ \isasymlbrakk A_1; ...; A_m \isasymrbrakk \Longrightarrow C
\]
where the variables $x_1,...,x_n$ are called \emph{parameters} and the premises $A_1, ..., A_m$ \emph{assumptions}. A formula $C$ is a \emph{theorem} in the locale if it is a conclusion. In its simplest form, a locale declaration consists of a sequence of context elements declaring parameters (keyword \textbf{fixes}) and assumptions (keyword \textbf{assumes}). Moreover, a locale can be imported and extended by new parameters and assumptions in a new locale (using ``+'' operator). 

Locales can be interpreted in the contexts of theories and structured proofs. The interpretation is made by assigning concrete data to parameters, and the resulting instantiated declarations are propagated to the current theory or the current proof context. This is called \emph{locale interpretation}. 

To enhance readability, we will use standard mathematical notation where possible. We use Isabelle/HOL syntax to represent the integration of the languages as well as the case studies.


\section{{\slang}: The Rely-guarantee Reasoning Framework}
\label{sect:picore}

This section presents the event language of {\slang} as well as its notion of computations, rely-guarantee proof system, the soundness of proof rules and invariant verification. 

\subsection{An Event Specification Language}
We design a new specification language for CRSs by introducing the notion of ``events''. The language decouples the \emph{triggering mechanism} in system reactions from the \emph{execution behaviour} of triggered events. The triggering mechanism specifies how and when the system reaction is triggered, and the execution behaviour defines how the system reaction is executed. The event language focuses on the triggering mechanism, and parametrizes the execution behaviour with the input parameters to the event that are provided when the event is triggered. The elemental component to specify a reaction triggering is the \emph{event}. To incorporate complex reaction structures, we design a set of structural compositions of events in the language.

\begin{figure}
\centering
\[
\begin{aligned}
\textbf{Event System}: & & \\
  \symbevtsys \ ::= &\ \event{\eventdef} 						& \text{(Basic Event)}\\
                  | &\ \atomevt{\eventdef} 						& \text{(Atomic Event)}\\
                  | &\ \anonevt{\symbprog}     					& \text{(Triggered Event)}\\
                  | &\ \evtseq{\symbevtsys_1}{\symbevtsys_2} 	& \text{(Sequence)}\\
                  | &\ \evtchoice{\symbevtsys_1}{\symbevtsys_2} & \text{(Choice)}\\
                  | &\ \evtjoin{\symbevtsys_1}{\symbevtsys_2} 	& \text{(Join)}\\
                  | &\ \evtrec{b}{\symbevtsys} 					& \text{(Iteration)} \\
\textbf{Parallel Event System}: & & \\
  \symbpes\ ::= &\ {\parsysc} 			& \text{(Parallel Composition of Event Systems)}
\end{aligned}
\]
\caption{The Abstract Syntax of Event Language}
\label{fig:lang}
\end{figure}

The abstract syntax of the event language is defined as shown in {\figprefix} \ref{fig:lang}. 
An event system $\symbevtsys$ is an event (\emph{basic}, \emph{atomic} and \emph{triggered}), or a structural compositions of event systems. 
A CRS is modelled as a parallel composition of event systems with shared states, denoted as $\symbpes$. It is a function from $\symbCore$ to event systems, where $\symbCore$ indicates the identifiers of event systems and is considered as the execution context of event systems. 
This design is more general and can be applied to track executing events. For instance, we use $\symbCore$ to represent the core identifier in multicore systems. 

\subsubsection{Events}\label{sec:events}
The syntax for events distinguishes basic, atomic and triggered events. Basic and atomic events are pending to be triggered in particular states.  
A basic event is defined as $\event{\eventdef}$, where $l$ is the event label, $g$ is the event guard condition, and $\symbprog$ is the body of the event. ${\eventdef}$ parametrizes triplets $(l,g,P)$ with the input parameters of the event, representing the set of all possible instances of $(l,g,P)$ from the domain of the parameters. This set is automatically constructed from the domain of the event parameters in the concrete syntax of {\slang}.
$\event{\eventdef}$ can be triggered when there is at least one element in the set of event instances for which the current state satisfies its guard. Then, the related body is executed, and it becomes a triggered event $\anonevt{\symbprog}$. The execution of $\anonevt{\symbprog}$ just simulates the execution of the program $\symbprog$. The execution of $\anonevt{\symbprog}$ interleaves with its parallel event systems. However, $\atomevt{\eventdef}$ is atomically executed in one single step without interleaving with the other parallel event systems. 

\subsubsection{Event Composition}
We design four structural constructs for event composition: \emph{Sequence}, \emph{Choice}, \emph{Join} and \emph{Iteration}. 
The sequence of two event systems models the sequential execution of them. The \emph{Choice} of two event systems is represented by $\evtchoice{\symbevtsys_1}{\symbevtsys_2}$. It models non-determinism of events by choosing one of them to be executed, i.e. if $\symbevtsys_1$ (or $\symbevtsys_2$) is able to move one step forward and becomes $\symbevtsys'_1$ ($\symbevtsys'_2$), $\evtchoice{\symbevtsys_1}{\symbevtsys_2}$ evolves to $\symbevtsys'_1$ ($\symbevtsys'_2$). The \emph{Choice} construct can represent the conditional structure (e.g. \textbf{IF} cond \textbf{THEN} P \textbf{ELSE} Q) using complementary guards for two basic events. 
Different from \emph{Choice}, the execution of the \emph{Join} of two event systems $\evtjoin{\symbevtsys_1}{\symbevtsys_2}$ is concurrent and finishes when the two event systems finish. Thus, it is also called \emph{concurrent join}. 
In addition to the top-level parallel composition of event systems ($\symbpes$), we retain the parallel composition at the event level (\emph{Join}). With the two levels of parallel composition, it becomes easy to specify CRSs with two levels of concurrency, such as concurrent OS kernels on multicore processors, where the kernel instance on each core is also concurrent due to its interruptable and preemptive execution. 
The \emph{Iteration} construct carries out interaction of events which can be repeated unboundedly. This construct is similar to \emph{while-loops} in imperative languages.


\subsection{Operational Semantics}

{\figprefix} \ref{fig:semantics} shows the small-step operational semantics of the event language. Due to the multi-level structure of the event language, we classify the transition rules into three categories: program transitions, event systems and parallel event systems. 

Since {\slang} is a parametric framework for underlying programs, the operational semantics of {\slang} is defined in terms of the semantics of the underlying program language and they are an input parameter in the locale modelling {\slang} in the Isabelle HOL mechanization. Hence, this section does not provide details on the semantics of programs.

\subsubsection{Transition Rules}

\begin{figure}
\small 

\begin{tabular}{cc}
\semanticrule{BasicEvt}{
\infer
{\estran{\event{\evtdef{ev}}}{\symbstate}{\actk{\evtent{l}}{\symbcore}}{\anonevt{\symbprog}}{\symbstate}}
{(l,g,P) \in \evtdef{ev} & \symbstate \in g}
}
&
\semanticrule{AtomEvt}{
\infer
{\estran{\atomevt{\evtdef{ev}}}{\symbstate}{\actk{\aevtent{l}}{\symbcore}}{\fin}{t}}
{(l,g,P) \in \evtdef{ev} & s \in g & \ptrans{\symbprog}{\symbstate}{\cmdfinal}{t}}
}
\end{tabular}
\vspace{0.1cm}

\begin{tabular}{cc}
\semanticrule{TriggeredEvt}{
\infer
{\estran{\anonevt{\symbprog}}{\symbstate}{\actk{\etranlabelc}{\symbcore}}{\anonevt{Q}}{t}}
{\ptran{\symbprog}{\symbstate}{Q}{t} & Q \neq \cmdfinal}
}
&
\semanticrule{TriggeredEvtFin}{
\infer
{\estran{\anonevt{\symbprog}}{\symbstate}{\actk{\etranlabelc}{\symbcore}}{\anonevt{Q}}{t}}
{\ptran{\symbprog}{\symbstate}{Q}{t} & Q = \cmdfinal}
}
\end{tabular}
\vspace{0.1cm}

\begin{tabular}{cc}
\semanticrule{EvtSeq}{
\infer
{\estran{\evtseq{\symbevtsys_1}{\symbevtsys_2}}{\symbstate}{\etranlabel}{\evtseq{\symbevtsys'_1}{\symbevtsys_2}}{t}}
{\estran{\symbevtsys_1}{\symbstate}{\etranlabel}{\symbevtsys'_1}{t}  & \symbevtsys'_1 \neq \fin }
}
&
\semanticrule{EvtSeqFin\vspace{0.2cm}}{
\infer{
\estran{\evtseq{\symbevtsys_1}{\symbevtsys_2}}{\symbstate}{\etranlabel}{\symbevtsys_2}{t}
}
{
\estran{\symbevtsys_1}{\symbstate}{\etranlabel}{\symbevtsys'_1}{t}  & \symbevtsys'_1 = \fin
}}
\end{tabular}
\vspace{0.1cm}

\begin{tabular}{cc}
\semanticrule{EvtChoice1}{
\infer{
\estran{\evtchoice{\symbevtsys_1}{\symbevtsys_2}}{\symbstate}{\etranlabel}{\symbevtsys'_1}{t}
}{
\estran{\symbevtsys_1}{\symbstate}{\etranlabel}{\symbevtsys'_1}{t}
}
}
&
\semanticrule{EvtChoice2}{
\infer{
\estran{\evtchoice{\symbevtsys_1}{\symbevtsys_2}}{\symbstate}{\etranlabel}{\symbevtsys'_2}{t}
}{
\estran{\symbevtsys_2}{\symbstate}{\etranlabel}{\symbevtsys'_2}{t}
}}
\end{tabular}
\vspace{0.1cm}

\begin{tabular}{cc}
\semanticrule{EvtJoin1}{
\infer{
\estran{\evtjoin{\symbevtsys_1}{\symbevtsys_2}}{\symbstate}{\etranlabel}{\evtjoin{\symbevtsys'_1}{\symbevtsys_2}}{t}
}{
\estran{\symbevtsys_1}{\symbstate}{\etranlabel}{\symbevtsys'_1}{t}
}}
&
\semanticrule{EvtJoin2}{
\infer{
\estran{\evtjoin{\symbevtsys_1}{\symbevtsys_2}}{\symbstate}{\etranlabel}{\evtjoin{\symbevtsys_1}{\symbevtsys'_2}}{t}
}{
\estran{\symbevtsys_2}{\symbstate}{\etranlabel}{\symbevtsys'_2}{t}
}}
\end{tabular}
\vspace{0.1cm}

\begin{tabular}{c}
\semanticrule{EvtJoinFin\vspace{0.2cm}}{
\infer{
\estran{\evtjoin{\fin}{\fin}}{\symbstate}{\actk{\etranlabelc}{\symbcore}}{\fin}{\symbstate}
}{-}}
\end{tabular}
\vspace{0.1cm}

\begin{tabular}{cc}
\semanticrule{EvtIterT}{
\infer{
\estran{\evtrec{b}{\symbevtsys}}{\symbstate}{\actk{\etranlabelc}{\symbcore}}{\evtseq{\symbevtsys}{(\evtrec{b}{\symbevtsys})}}{\symbstate}
}{\symbstate \in b & \symbevtsys \neq \fin}}
&
\semanticrule{EvtIterF}{
\infer{
\estran{\evtrec{b}{\symbevtsys}}{\symbstate}{\actk{\etranlabelc}{\symbcore}}{\fin}{\symbstate}
}{\symbstate \notin b}}
\end{tabular}
\vspace{0.1cm}

\begin{tabular}{ccc}
\semanticrule{Par}{
\infer{
\pestran{\symbpes}{\symbstate}{\actk{\symbact}{\symbcore}}{\symbpes(\symbcore \mapsto \symbevtsys)}{t}
}{
\estran{\symbpes \ \symbcore}{\symbstate}{\actk{\symbact}{\symbcore}}{\symbevtsys}{t}
}}
&
\semanticrule{EnvTran\_es\vspace{0.2cm}}{
\infer{\esenvtran{\symbevtsys}{\symbstate}{\symbevtsys}{t}}{-}
}
&
\semanticrule{EnvTran\_pes\vspace{0.2cm}}{
\infer{\pesenvtran{\symbpes}{\symbstate}{\symbpes}{t}}{-}
}
\end{tabular}

\caption{Small-step Operational Semantics of the Event Language}
\label{fig:semantics}
\end{figure}

The small-step transition relation of the underlying program language is denoted as $\ptran{P}{s}{Q}{t}$, where $P,Q$ are programs, $s,t$ are states and $\symbenv$ is used for some static configuration for programs (e.g. an environment for procedure declarations). The notation $\ptrans{P}{s}{Q}{t}$ is the reflexive transitive closure of $\ptran{P}{s}{Q}{t}$. 
In {\slang}, $\cmdfinal$ denotes program termination. 
The transition rules for event systems have the form of $\estran{\symbevtsys}{s}{\etranlabel}{\symbevtsys'}{t}$, where $\etranlabel =  \actk{\symbact}{\symbcore}$ is the combination of a transition label $\symbact$ of event systems and the execution context $\symbcore$. It means that transition $\symbact$ is done by the event system $\symbcore$. $\symbact$ indicates the transition type of event systems and is defined as 
\[\symbact = \etranlabelc \mid \evtent{l} \mid \aevtent{l} \]
The notation $\etranlabelc$ is a label constant representing anonymous transitions of triggered events (\semrulename{TriggeredEvt} and \semrulename{TriggeredEvtFin}), of expression evaluations (\semrulename{EvtIterT} and \semrulename{EvtIterF}) and of eliminating final triggered events $\fin$ (\semrulename{EvtJoinFin}). $\cmdfinal$ is the notation to represent the termination of programs. 
The other two types of labels indicate actions respectively triggering events and atomic events. $\actk{ }{\symbcore}$ means that the action occurs in event system $\symbcore$. 
The transition rule for parallel event systems has a similar form of $\pestran{\symbpes}{s}{\actk{\symbact}{\symbcore}}{\symbpes'}{t}$. 

Next, we discuss each transition rule of the semantics. The transition rule for triggered events is straightforward and defined directly by the transition rules of programs. 
The \semrulename{BasicEvt} rule specifies that a basic event occurs when one condition is satisfied by the current state. When a basic event is triggered in event system $\symbcore$, the execution context is changed to the event, and the body begins to be executed. 
As shown by the \semrulename{AtomEvt} rule, an atomic event behaves like a basic event, but it is atomically executed, i.e. it finishes its execution in a single step. 
The execution of a triggered event $\anonevt{\symbprog}$ just simulates the program $\anonevt{\symbprog}$ (\semrulename{TriggeredEvt}) and the execution context of the event system is set to {\textbf{None}} when its execution finishes (\semrulename{TriggeredEvtFin}). 
The \semrulename{EvtSeq} and \semrulename{EvtSeqFin} rules specify the behaviour of the sequential composition of two event systems. 
The \semrulename{EvtChoice1} and \semrulename{EvtChoice2} rules show a nondeterministic choice between two event systems. 
The \semrulename{EvtJoin1} and \semrulename{EvtJoin2} rules show the interleaving of the concurrent execution of two event systems in the interleaving manner. 
The execution of $\evtjoin{\symbevtsys_1}{\symbevtsys_2}$ finishes as specified by the \semrulename{EvtParFin} rule. 
The \semrulename{EvtIterT} and \semrulename{EvtIterF} rules specify the execution of an event system in a loop manner. The body of an iteration cannot be specified as the termination since it is a semantical notation. 
Finally, the \semrulename{Par} rule describes behaviour of parallel event systems. The execution of $\symbpes$ is the interleaving of its event systems. 
This design relaxes the atomicity of events in other approaches (e.g., Event-B \cite{Abrial07}).

Environment transition rules (\semrulename{EnvTran\_es} and \semrulename{EnvTran\_pes}) have the form $\anyenvtran{\symbspec}{s}{\symbspec}{t}$ where $\symbspec$ is a specification and the subscript $_{\square}$ ($_e$ and $_{pe}$) indicates the transition objects. Intuitively, a transition made by the environment may change the state but not the specification and the execution context. 

\subsection{Computation}
\label{subsect:comp}
The concept of \emph{computation} is intended to define the notion of validity using of the operational semantics in the {\slang} rely-guarantee proof system. 
We define a computation as a sequence of component and environment transitions. 
The most intuitively way of describing a computation is by directly using the operational semantics to construct sequence of computations as shown in Figure~\ref{fig:linear_comp}. That is, a computation is inductively constructed from an existing computation $(\symbevtsys_2, t)\#cs$ and a environment or component transition from $(\symbevtsys_1,s)$ to $(\symbevtsys_2,t)$. Note that this definition loses the structure of a computation. An alternative equivalent definition for computations takes its structure into account as we show in Figure~\ref{fig:modular_comp}. Although it is more elaborated, it can considerably simplify proofs, especially those involving iterative event systems. 

In previous versions of {\slang} in \cite{Zhao19FM}, we use the linear definition of computation and thus the proof of its soundness takes $\approx$ 8,000 lines. By the modular definition, the proof is reduced to $\approx$ 4,200 lines, even though we extend more general constructs for {\slang}.

\subsubsection{Linear Computation}
We define the set of computations of all event systems with static information $\symbenv$ as $\compfun(\symbenv)$, which is a set of lists of configurations inductively defined as shown in {\figprefix} \ref{fig:linear_comp}. The singleton list is always a computation (\comprulename{CptsOne}). Two consecutive configurations are part of a computation if they are the initial and final configurations of an environment (\comprulename{CptsEnv}) or action transition (\comprulename{CptsAct}). The operator $\#$ in $e\# l$ represents the insertion of element $e$ in list $l$. 

\begin{figure}
\small
\begin{tabular}{ccc}
\comprule{CptsOne\vspace{0.3cm}}{\infer{[(\symbevtsys, \symbstate)] \in \compfun(\symbenv)}{ - }}
&
\comprule{CptsEnv\vspace{0.2cm}}{
\infer{
(\symbevtsys, \symbstate)\#(\symbevtsys, t)\#cs \in \compfun(\symbenv) 
}{
 (\symbevtsys, t)\#cs \in \compfun(\symbenv) 
}}
&
\comprule{CptsAct}{
\infer{ 
(\symbevtsys_1, \symbstate)\#(\symbevtsys_2, t)\#cs \in \compfun(\symbenv) 
}{ 
\estran{\symbevtsys_1}{s}{\etranlabel}{\symbevtsys_2}{t} & (\symbevtsys_2, t)\#cs \in  \compfun(\symbenv) 
}}
\end{tabular}
\caption{Linear Definition of Computation}
\label{fig:linear_comp}
\end{figure}


We use $\compfun(\symbenv,\symbevtsys)$ to denote the set of computations of an event system $\symbevtsys$. The function $\compfun(\symbenv, \symbevtsys, \symbstate)$ denotes the computations of $\symbevtsys$ starting up from an initial state $\symbstate$ and execution context $x$. Computations for parallel event systems are defined in a similar way. 

\subsubsection{Modular Computation}
We define the set of computations of event systems as the lists of configurations in an inductive manner formed by the rules shown in {\figprefix} \ref{fig:modular_comp}. For the modular computation, we use the notation $\compfunm$ instead of $\compfun$. 

\begin{figure}
\small 

\begin{tabular}{cc}
\comprule{CptsMOne\vspace{0.3cm}}{
\infer{[(\symbevtsys, \symbstate)] \in \compfunm(\symbenv)}{ - }}
&
\comprule{CptsMEnv\vspace{0.2cm}}{
\infer
{ (\symbevtsys, \symbstate)\#(\symbevtsys, t)\#cs \in \compfunm(\symbenv) }
{ (\symbevtsys, t)\#cs \in \compfunm(\symbenv) }}
\end{tabular}
\vspace{0.2cm}

\begin{tabular}{cc}
\comprule{CptsMTrgEvt}{
  \infer
	{ (\anonevt{\symbprog}, \symbstate)\#(\anonevt{Q}, t)\#cs \in \compfunm(\symbenv) }
	{ \begin{tabular}{l}
		$\ptran{\symbprog}{\symbstate}{Q}{t} \quad Q \neq \cmdfinal$ \\
		$(\anonevt{Q}, t)\#cs \in \compfunm(\symbenv)$
	  \end{tabular}
	}
}
&
\comprule{CptsMTrgEvtFin}{
  \infer
	{ (\anonevt{\symbprog}, \symbstate)\#(\anonevt{Q}, t)\#cs \in \compfunm(\symbenv) }
	{ \begin{tabular}{l}
		$\ptran{\symbprog}{\symbstate}{Q}{t} \quad Q = \cmdfinal$\\
		$(\anonevt{Q}, t)\#cs \in \compfunm(\symbenv)$
	   \end{tabular}
	}
}
\end{tabular}
\vspace{0.2cm}

\begin{tabular}{cc}
\comprule{CptsMBasicEvt}{
  \infer
	{ (\event{\evtdef{ev}},\symbstate)\#(\anonevt{\symbprog},\symbstate)\#cs \in \compfunm(\symbenv) }
	{ \begin{tabular}{l}
		$(l,g,P) \in \evtdef{ev} $\\
		$\symbstate \in g \quad (\anonevt{\symbprog},\symbstate)\#cs \in \compfunm(\symbenv)$
	  \end{tabular}
	}
}
&
\comprule{CptsMAtomEvt}{
  \infer
	{ (\atomevt{\evtdef{ev}},\symbstate)\#(\fin,t)\#cs \in \compfunm(\symbenv) }
	{ \begin{tabular}{l}
		$(l,g,P) \in \evtdef{ev} \quad \ptrans{\symbprog}{\symbstate}{\cmdfinal}{t} $\\
		$s \in g \quad (\fin,t,x)\#cs \in \compfunm(\symbenv)$
	  \end{tabular}
	}
}
\end{tabular}
\vspace{0.2cm}

\begin{tabular}{cc}
\comprule{CptsMSeq}{
  \infer
	{ (\evtseq{\symbevtsys_1}{\symbevtsys_2},\symbstate) \# (\evtseq{\symbevtsys'_1}{\symbevtsys_2},t) \# cs \in \compfunm(\symbenv) }
	{ \begin{tabular}{l}
		$\estran{\symbevtsys_1}{\symbstate}{\etranlabel}{\symbevtsys'_1}{t}$\\
		$(\evtseq{\symbevtsys'_1}{\symbevtsys_2},t) \# cs \in \compfunm(\symbenv)$
	  \end{tabular}
	}
}
&
\comprule{CptsMSeqFin}{
  \infer
	{ (\evtseq{\symbevtsys_1}{\symbevtsys_2},\symbstate) \# (\symbevtsys_2,t) \# cs \in \compfunm(\symbenv) }
	{ \begin{tabular}{l}
		$\estran{\symbevtsys_1}{\symbstate}{\etranlabel}{\fin}{t}$ \\
		$(\symbevtsys_2,t) \# cs \in \compfunm(\symbenv) $
	  \end{tabular}
	}
}
\end{tabular}
\vspace{0.2cm}

\begin{tabular}{cc}
\comprule{CptsMChc1}{
  \infer
	{ (\evtchoice{\symbevtsys_1}{\symbevtsys_2},\symbstate) \# (\symbevtsys'_1,t) \# cs \in \compfunm(\symbenv) }
	{ \begin{tabular}{l}
		$\estran{\symbevtsys_1}{\symbstate}{\etranlabel}{\symbevtsys'_1}{t}$\\
		$(\symbevtsys'_1,t) \# cs \in \compfunm(\symbenv)$
	  \end{tabular}
	}
}
&
\comprule{CptsMChc2}{
  \infer
	{ (\evtchoice{\symbevtsys_1}{\symbevtsys_2},\symbstate) \# (\symbevtsys'_2,t) \# cs \in \compfunm(\symbenv) }
	{ \begin{tabular}{l}
		$\estran{\symbevtsys_2}{\symbstate}{\etranlabel}{\symbevtsys'_2}{t}$ \\
		$(\symbevtsys'_2,t,y) \# cs \in \compfunm(\symbenv)$
	  \end{tabular}
	}
}
\end{tabular}
\vspace{0.2cm}

\begin{tabular}{cc}
\comprule{CptsMJoin1}{
  \infer
	{ (\evtjoin{\symbevtsys_1}{\symbevtsys_2},\symbstate) \# (\evtjoin{\symbevtsys'_1}{\symbevtsys_2},t) \# cs \in \compfunm(\symbenv) }
	{ \begin{tabular}{l}
		$\estran{\symbevtsys_1}{\symbstate}{\etranlabel}{\symbevtsys'_1}{t}$ \\
		$(\evtjoin{\symbevtsys'_1}{\symbevtsys_2},t) \# cs \in \compfunm(\symbenv)$
	  \end{tabular}
	}
}
&
\comprule{CptsMJoin2}{
  \infer
	{ (\evtjoin{\symbevtsys_1}{\symbevtsys_2},\symbstate) \# (\evtjoin{\symbevtsys_1}{\symbevtsys'_2},t) \# cs \in \compfunm(\symbenv) }
	{ \begin{tabular}{l}
		$\estran{\symbevtsys_2}{\symbstate}{\etranlabel}{\symbevtsys'_2}{t}$ \\
		$(\evtjoin{\symbevtsys_1}{\symbevtsys'_2},t) \# cs \in \compfunm(\symbenv)$
	  \end{tabular}
	}
}
\end{tabular}
\vspace{0.2cm}

\begin{tabular}{cc}
\comprule{CptsMJoinFin}{
  \infer
	{ (\evtjoin{\fin}{\fin},\symbstate) \# ({\fin},\symbstate) \# cs \in \compfunm(\symbenv) }
	{ ({\fin},\symbstate) \# cs \in \compfunm(\symbenv) }
}
&
\comprule{CptsMIterF}{
  \infer
	{ (\evtrec{b}{\symbevtsys},\symbstate) \# (\fin,\symbstate) \# cs \in \compfunm(\symbenv)  }
	{ \symbstate \notin b
	& (\fin,\symbstate) \# cs \in \compfunm(\symbenv) }
}
\end{tabular}
\vspace{0.2cm}


\begin{tabular}{c}
\comprule{CptsMIterTOne}{
  \infer
	{ (\evtrec{b}{\symbevtsys},\symbstate) \# (\ \textbf{lift-seq-cpt}\ (\evtrec{b}{\symbevtsys}) \ ((\symbevtsys,\symbstate) \# cs)\ ) \in \compfunm(\symbenv) }
	{ \symbstate \in b
	& (\symbevtsys,\symbstate) \# cs \in \compfunm(\symbenv)
	}
}
\end{tabular}
\vspace{0.2cm}

\begin{tabular}{c}
\comprule{CptsMIterTMore}{
  \infer
	{ (\evtrec{b}{\symbevtsys},\symbstate) \# (\ \textbf{lift-seq-cpt}\ (\evtrec{b}{\symbevtsys}) \ ((\symbevtsys,\symbstate) \# cs)\ ) @ (\evtrec{b}{\symbevtsys},t) \# cs' \in \compfunm(\symbenv) }
	{ \symbstate \in b
	& (\symbevtsys,\symbstate) \# cs \in \compfunm(\symbenv)
	& \estranconf{last ((\symbevtsys,\symbstate) \# cs)}{\etranlabel}{(\fin,t)}
	& (\evtrec{b}{\symbevtsys},t) \# cs' \in \compfunm(\symbenv) }
}
\end{tabular}

\caption{Modular Definition of Computation}
\label{fig:modular_comp}
\end{figure}

The first two rules are the same as those in the linear definition. The {\comprulename{CptsAct}} rule, however, is now replaced by fourteen rules which correspond to different kinds of component transitions. 
The rules in the modular semantics for events and for \emph{sequence}, \emph{choice}, and \emph{join} composition, directly follow the system event operational semantics for each construct as defined in {\figprefix} \ref{fig:semantics}.

We explain the detail of the last two rules for iteration. They unfold the iteration into the sequential execution of the event system being iterated. 
The {\comprulename{CptsMIterTOne}} rule represents the computations with one iteration, and {\comprulename{CptsMIterTMore}} those where the body has been iteratively executed for more than once. The $last$ function indicates the last element in a list and the ``@'' function concatenates two lists. 

Here, we use an auxiliary function \textbf{lift-seq-cpt} as follows. Given a computation $c$ and an event system $Q$, the function returns the same computation where the specification of each configuration has been sequentially composed with $Q$. The $map$ function applies the transformation (a lambda function) to each element of a list.
\[
\textbf{lift-seq-cpt} \ c \ Q \equiv map \ (\lambda\ (\symbevtsys,\symbstate).\ (\evtseq{\symbevtsys}{Q}, \symbstate))\ c
\]

\subsubsection{Equivalence of the Linear and Modular Computations}
The following theorem shows the equivalence in {\slang} between the linear and modular definitions for computation.

\begin{theorem}[Equivalence of Linear and Modular Computations]\label{thm:equiv_comp}
The sets of linear and modular computations are equivalent, i.e. $\compfun(\symbenv) = \compfunm(\symbenv)$.
\end{theorem}

This theorem is intensively used when proving the soundness of the rely-guarantee proof rules in {\slang}, especially the proof rule of iteration. 
From theorem~\ref{thm:equiv_comp}, it can be directly inferred that $\compfun(\symbenv,\symbevtsys) = \compfunm(\symbenv,\symbevtsys)$. 

\subsection{Rely-guarantee Proof System}
For compositional reasoning, we propose a rely-guarantee proof system for {\slang}. 
We first introduce the rely-guarantee specification and its validity. Then, we present a set of proof rules and its soundness with regards to the notion of validity. 

We consider the verification of two classes of properties in the rely-guarantee proof system of {\slang}: functional correctness by pre and post conditions and safety by invariants. 

\subsubsection{Validity}

A rely-guarantee specification in {\slang} is a quadruple $\rgconddefault$, where $pre$ is the precondition, $rely$ is the rely condition, $guar$ is the guarantee condition, and $post$ is the post condition. The assumption and commitment functions for event systems are denoted by $\assumefun$ and $\commitfun$ respectively. For each computation of an event system $\symbevtsys$, $\symbcomp \in \compfun(\symbenv,\symbevtsys)$, we use $\symbcomp_i$ to denote the configuration at index $i$. $\symbspec_{\symbcomp_i}$ and $\symbstate_{\symbcomp_i}$ represent the projection of each component in the tuple $\symbcomp_i =(\symbspec, \symbstate)$. The $last$ function indicates the last element in a list. 
\begin{equation} 
\label{eq:assume_commit}
\begin{aligned}
& \assumefun(\textbf{pre}, \textbf{rely}) \equiv \{\symbcomp \mid \symbstate_{\symbcomp_0} \in \textbf{pre} \wedge (\forall i < len(\symbcomp) - 1. \ (\symbspec_{\symbcomp_i} = \symbspec_{\symbcomp_{i+1}}) \longrightarrow (\symbstate_{\symbcomp_i},\symbstate_{\symbcomp_{i+1}}) \in \textbf{rely})\} \\
& \commitfun(\symbenv, \textbf{guar}, \textbf{post}) \equiv \{\symbcomp \mid \ (\forall i < len(\symbcomp) - 1. \ (\exists \etranlabel.\ \estranconf{\symbcomp_i}{\etranlabel}{\symbcomp_{i+1}}) \longrightarrow (\symbstate_{\symbcomp_i},\symbstate_{\symbcomp_{i+1}}) \in \textbf{guar}) \\
& \quad \quad \quad \quad \quad \quad \quad \quad \quad \quad \quad \quad 
\wedge (\symbspec_{last(\symbcomp)} = \fin \longrightarrow \symbstate_{last(\symbcomp)} \in \textbf{post})\} 
\end{aligned}
\end{equation}

The validity of rely-guarantee specification for event systems is defined as 
\[
\RGSATe{\symbevtsys}{\rgconddefaultb} \equiv \forall \symbstate. \ \compfun(\symbenv, \symbevtsys, \symbstate) \cap \assumefun(\textbf{pre}, \textbf{rely}) \subseteq \commitfun(\symbenv, \textbf{guar}, \textbf{post})
\]

Intuitively, validity represents that the set of computations $cpts$ starting at the configuration $(\symbevtsys, \symbstate)$, with $\symbstate \in \textbf{pre}$ and environment transitions in a computation $cpt \in cpts$ belonging to the rely relation $\textbf{rely}$, is a subset of the set of computations where component transitions belong to the guarantee relation $\textbf{guar}$ and if the event system terminates, then the final states belongs to $\textbf{post}$. 
Validity for parallel event systems is defined similarly.

\subsubsection{Proof Rules}
Now we present the proof rules of {\slang} as shown in {\figprefix} \ref{fig:proof_rules}, which give us a relational proof method for CRSs. We first define $\stablefun{f}{g} \equiv \forall x,y. \ x \in f \wedge (x, y) \in g \longrightarrow y \in f$. Thus, $\stablefun{pre}{rely}$ means that the precondition is stable when the rely condition holds. Rules may assume stability of the precondition with regards to the rely relation to ensure that the precondition holds during the environment transitions. 
The identity set is denoted as {\idset}, which is defined as $\{{p \mid \exists x .\ p = (x, x)}\}$. The universal set is denoted as {\univset}. 

We design three forms of proof rules according to the different specification levels: program, event systems, and parallel event systems. They are denoted by:

\begin{equation*} 
\begin{aligned}
(1)\ & \rgsatp{\symbprog}{\rgconddefault} \\
(2)\ & \rgsate{\symbevtsys}{\rgconddefault} \\
(3)\ & \rgsatpe{\symbpes}{\rgconddefault}
\end{aligned}
\end{equation*}

The {\slang} framework reuses the proof rules for programs in 3rd-party languages rather than defining them. 
We discuss the proof rules for event systems and parallel event systems in detail.

\begin{figure}
\small 

\begin{tabular}{cc}
\proofrule{RG-BasicEvt}
{
  \infer{ \rgsate{\event{\evtdef{ev}}}{\rgconddefault} }
		{ 
		\begin{tabular}{l}
			$\forall (l,g,P) \in \evtdef{ev}.\ \rgsatp{\symbprog}{\rgcond{pre \cap g}{rely}{guar}{post}}$ \\
			$\stablefun{pre}{rely} \quad \idset \subseteq guar$
		\end{tabular} 
		}
}
&
\proofrule{RG-TrgEvt\vspace{0.32cm}}
{
  \infer{ \rgsate{\anonevt{\symbprog}}{\rgconddefault} }
		{ \rgsatp{\symbprog}{\rgconddefault} }
}
\end{tabular}
\vspace{0.2cm}

\begin{tabular}{c}
\proofrule{RG-AtomEvt}
{
  \infer{ \rgsate{\atomevt{\evtdef{ev}}}{\rgconddefault} }
		{ 
		\begin{tabular}{l}
			$\forall (l,g,P) \in \evtdef{ev}. \forall V.\ \rgsatp{\symbprog}{\rgcond{pre \cap g \cap \{V\}}{\idset}{\univset}{\{s \mid (V,s) \in guar\} \cap post }}$ \\
			$\stablefun{pre}{rely} \quad \stablefun{post}{rely}$
		\end{tabular} 
    	}
}
\end{tabular}
\vspace{0.2cm}

\begin{tabular}{cc}
\proofrule{RG-Seq}
{
  \infer{ \rgsate{\evtseq{\symbevtsys_1}{\symbevtsys_2}}{\rgconddefault} }
		{ 
		\begin{tabular}{l}
			$\rgsate{\symbevtsys_1}{\rgcond{pre}{rely}{guar}{mid}} $ \\
			$\rgsate{\symbevtsys_2}{\rgcond{mid}{rely}{guar}{post}}$ 
		\end{tabular}
		}
}
&
\proofrule{RG-Choice}
{
  \infer{ \rgsate{\evtchoice{\symbevtsys_1}{\symbevtsys_2}}{\rgconddefault} }
		{ 
		\begin{tabular}{l}
			$\rgsate{\symbevtsys_1}{\rgconddefault}$ \\
			$\rgsate{\symbevtsys_2}{\rgconddefault}$
		\end{tabular}
		}
}
\end{tabular}
\vspace{0.2cm}

\begin{tabular}{c}
\proofrule{RG-Join}
{
  \infer{ \rgsate{\evtjoin{\symbevtsys_1}{\symbevtsys_2}}{\rgcond{pre_1 \cap pre_2}{rely}{guar}{post_1 \cap post_2}} }
		{ 
		\begin{tabular}{l}
			(1) $\rgsate{\symbevtsys_1}{\rgcond{pre_1}{rely_1}{guar_1}{post_1}}$ \quad (2) $\rgsate{\symbevtsys_2}{\rgcond{pre_2}{rely_2}{guar_2}{post_2}}$
			\\
			(3) $guar_1 \cup guar_2 \subseteq guar$ \quad (4) $rely \cup guar_2 \subseteq rely_1$ \quad (5) $rely \cup guar_1 \subseteq rely_2$ \quad (6) $\idset \subseteq guar$ 
		\end{tabular}
	    }
}
\end{tabular}
\vspace{0.2cm}

\begin{tabular}{c}
\proofrule{RG-Iter}
{
  \infer{ \rgsate{\evtrec{b}{\symbevtsys}}{\rgcond{inv}{rely}{guar}{post}} }
		{ \rgsate{\symbevtsys}{\rgcond{inv \cap b}{rely}{guar}{inv}}
		& inv \cap ~b \subseteq post & \idset \subseteq guar & \stablefun{inv}{rely} & \stablefun{post}{rely}
		}
}
\end{tabular}
\vspace{0.2cm}

\begin{tabular}{c}
\proofrule{RG-Conseq}
{
  \infer{ \rgsate{\symbevtsys}{\rgconddefault} }
		{ \rgsate{\symbevtsys}{\rgcond{pre'}{rely'}{guar'}{post'}} 
		& pre \subseteq pre' & rely \subseteq rely' & guar' \subseteq guar & post' \subseteq post}
}
\end{tabular}
\vspace{0.2cm}

\begin{tabular}{c}
\proofrule{RG-ParEvtSys}
{
  \infer{ \rgsatpe{\symbpes}{\rgconddefault} }
		{ 
		\begin{tabular}{l}
			$\forall \symbcore. \ \rgsate{\symbpes(\symbcore)}{\rgcond{Pre(\symbcore)}{Rely(\symbcore)}{Guar(\symbcore)}{Post(\symbcore)}}$ 
			\quad
			$\forall \symbcore \ \symbcore'. \ \symbcore \neq \symbcore' \longrightarrow Guar(\symbcore) \subseteq Rely(\symbcore')$\\
			$\forall \symbcore. \ pre \subseteq Pre(\symbcore)$ \quad $\forall \symbcore. \ rely \subseteq Rely(\symbcore)$ \quad 
			$\forall \symbcore. \ Guar(\symbcore) \subseteq guar$ \quad $(\bigcap \symbcore. \ Post(\symbcore)) \subseteq post$
		\end{tabular}
		}
}
\end{tabular}
%
%
%

\caption{Rely-guarantee Proof Rules of {\slang}}
\label{fig:proof_rules}
\end{figure}

\paragraph{Basic event}
As we explained in section~\ref{sec:events}, an event  $\event{\eventdef}$ is composed of the set of all possible instances of $(l,g,P)$ from the parameters domain. 
A basic event satisfies its rely-guarantee specification (rule {\proofrulename{RG-BasicEvt}}) if all parameter instantiation of the body $P\in \eventdef$ satisfies the specification, strengthening the precondition with the guard of the event. Since the occurrence of a basic event (see the {\semrulename{BasicEvt}} transition rule in {\figprefix} \ref{fig:semantics}) does not change the state, the guarantee relation must include the identity relation to accept stuttering transitions. 

\paragraph{Atomic event}
By the semantics of atomic events, a positive evaluation of its guard and the execution of its body is done atomically. Thus, the state transition caused by the complete execution of $\symbprog$ must satisfy the guarantee condition. This is presented in the pre and post conditions of $\symbprog$ in the assumptions of the {\proofrulename{RG-AtomEvt}} rule. We use a universally quantified variable $V$ to relate the state before and after the transformation. 
Changes during the execution of $\symbprog$ may arbitrarily modify intermediate states, hence $\symbprog$ must guarantee any possible transition. Thus the guarantee condition is $\univset$.
Since $\symbprog$ is executed atomically, the environment cannot change the state, i.e. the rely condition is the identity relation $\idset$. To ensure that both the pre and post conditions are preserved after environment transitions, we request stability of $pre$ and $post$. 

\paragraph{Triggered event}
The {\proofrulename{RG-TrgEvt}} rule says that a triggered event $\anonevt{\symbprog}$ satisfies the rely-guarantee specification if the program $\symbprog$ satisfies the specification. This rule is directly derived from the semantics for triggered events in {\figprefix} \ref{fig:semantics}, where triggered events modify the state according to how the program modifies the state. Note that this rule is not necessary in the proof system, since valid formulas on an event are tackled by  {\proofrulename{RG-BasicEvt}}, however we keep it for the sake of completeness.

\paragraph{Sequence of event systems}
The rules for the sequential composition (rule {\proofrulename{RG-Seq}}) is standard. 

\paragraph{Choice of event systems}
Regarding the semantics of the Choice of event systems, $\evtchoice{\symbevtsys_1}{\symbevtsys_2}$ behaves as $\symbevtsys_1$ or $\symbevtsys_2$. Therefore, $\evtchoice{\symbevtsys_1}{\symbevtsys_2}$ satisfies the rely-guarantee specification when both $\symbevtsys_1$ and $\symbevtsys_2$ satisfy the specification.  

\paragraph{Join of event systems}
The concurrent join of two event systems satisfies a rely-guarantee specification if 

\begin{itemize}
\item The two systems satisfy their specifications (Premise 1 and 2).
\item since a component transition of the join is performed by one of its systems, the guarantee condition \emph{guar} must be at least the union of the guarantee conditions of the two systems (Premise 3);
\item An environment transition for one system corresponds to a component transition of another system, or of a transition from the overall environment. Hence, the strongest rely condition for $\symbevtsys_1$ is $rely \cup guar_2$ (Premise 4), and it is the same as for $\symbevtsys_2$ (Premise 5);
\item Since the \semrulename{EvtJoinFin} rule in {\figprefix} \ref{fig:semantics} does not change the state, the guarantee relation must include the identity relation (Premise 6).
\item Finally, the pre and post conditions for the join are the intersection of pre and post conditions of the two systems, respectively. 
\end{itemize}

\paragraph{Iteration of event systems}
In the {\proofrulename{RG-Iter}} rule the precondition plays the role of the invariant; it must hold before and after the execution of the body at every iteration. 

\paragraph{Consequence rule}
The rule of consequence allows us to strengthen the assumptions and weaken the commitments. 

\paragraph{Parallel event systems}
The {\proofrulename{RG-ParEvtSys}} rule presents the parallel composition of a set of event systems, which are described as a function $\symbpes: \symbCore \rightarrow \symbevtsys$. The premises of this rule are similar to {\proofrulename{RG-Join}, extended to an arbitrary number of parallel event systems.


\subsubsection{Soundness}
Finally, the soundness theorem for a specification relates rely-guarantee specifications proven on the proof system with its validity.

\begin{theorem}[Soundness of Proof Rules for Event Systems]
if \ $\rgsate{\symbevtsys}{\rgconddefault}$, then $\RGSATe{\symbevtsys}{\rgconddefault}$. 
\end{theorem}

\begin{theorem}[Soundness of rule \proofrulename{RG-ParEvtSys}]
if \ $\rgsatpe{\symbpes}{\rgconddefault}$, then $\RGSATpe{\symbpes}{\rgconddefault}$. 
\end{theorem}

\subsection{Invariant Verification}
In many cases, we would like to show that CRSs preserve certain data invariants. Since CRSs may not be closed systems, i.e. their environment may change the system state that is represented by rely conditions of CRSs, the reachable states of CRSs are dependent on both the initial states and the environment. The preservation of an invariant on a CRS is defined as follows.

\begin{definition}[Preservation of Invariants]
A CRS $\symbpes$ with static information $\symbenv$, starting up from a set of initial states $init$ under an environment $rely$, preserves an invariant $inv$ when its reachable states satisfy the predicate: 
\[
\forall \symbcomp.\ \symbcomp \in \compfun(\symbenv, \symbpes) \cap \assumefun(init, rely) \longrightarrow (\forall i<len(\symbcomp).\ inv(s_{\symbcomp_i}))
\]
\end{definition}

In this definition, $\symbcomp$ denotes an arbitrary computation of $\symbpes$ from a set of initial states $init$ and under an environment $rely$. It requires that all states in $\symbcomp$ satisfy the invariant $inv$.  

To show that $inv$ is preserved by a system $\symbpes$, it suffices to show the invariant verification theorem as follows. This theorem indicates that (1) the system satisfies its rely-guarantee specification $\rgcond{init}{rely}{guar}{post}$, (2) $inv$ initially holds in the set of initial states, and (3) each component transition as well as each environment transition preserves $inv$. Later, by the proof system of {\slang}, invariant verification is decomposed to the verification of individual events. 

\begin{theorem}[Invariant Verification]
\label{thm:invariant}
For formal specification $\symbpes$ and $\symbenv$, a state set $init$, a rely condition $rely$, and an invariant $inv$, if 
\begin{itemize}
\item $\rgsatpe{\symbpes}{\rgcond{init}{rely}{guar}{post}}$.
\item $init \subseteq \{s \mid inv(s)\}$.
\item $\stablefun{\{s \mid inv(s)\}}{rely}$ and $\stablefun{\{s \mid inv(s)\}}{guar}$ are satisfied.
\end{itemize}
then $inv$ is preserved by $\symbpes$ w.r.t. $\symbenv$, $init$ and $rely$. 
\end{theorem}

\section{Language Integration by Adapter Pattern}
\label{sect:langintegrate}
{\slang} parametrizes the language and its rely-guarantee proof system at the program level using a rely-guarantee interface and thus allows to easily reuse existing rely-guarantee frameworks. This reusability is achieved by using an \emph{adapter} design pattern~\cite{Gamma95}. This section first overviews the idea of language integration. We then present the rely-guarantee interface and its adapters for languages. Finally, we show how the {\slang} framework is implemented in Isabelle/HOL.

\subsection{Integration Overview}
The adapter pattern is a software design pattern that allows the interface of an existing class to be used as another interface \cite{Gamma95} by defining a separate adapter class that converts the (incompatible) interface of a class (\emph{adaptee}) into another interface (\emph{target}) that clients require. It is often used to make existing classes work with others without modifying their source code. 

To implement a flexible integration of languages used for the definition of the event bodies, {\slang} provides a rely-guarantee interface that program languages respect. The interface is an abstraction for common rely-guarantee components required by {\slang}. These components are represented as a set of \emph{parameters} and their \emph{assumptions}. The language, semantics, proof rules and soundness proof of {\slang} in {\sectprefix} \ref{sect:picore} are developed using this interface. The architecture of language integration into {\slang} is shown in {\figprefix} \ref{fig:integrate}. 

\begin{figure}
\begin{center}
\includegraphics[width=5.0in]{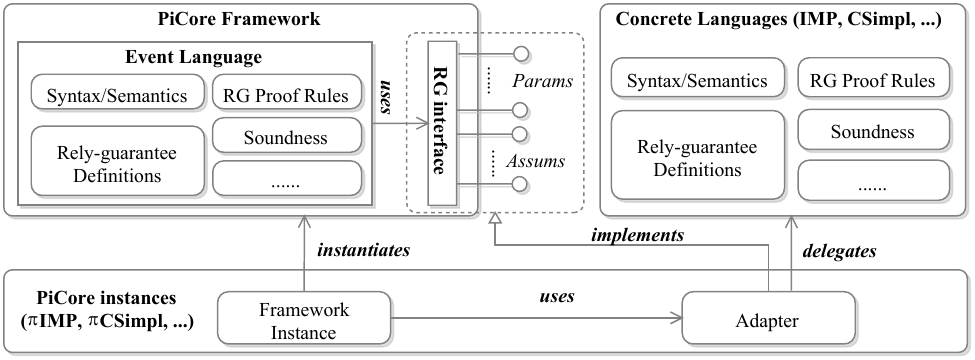}
\end{center}
\caption{Languages Integration in {\slang}}
\label{fig:integrate}
\end{figure}

When integrating a 3rd-party language into {\slang}, the language is the adaptee, and the rely-guarantee interface is the target. It is necessary to provide a \emph{rely-guarantee adapter} to bridge the differences of rely-guarantee components between {\slang} and the language. The adapter implements the rely-guarantee interface by delegating functionality of the event language to the integrated language. 
Following this adapter, a 3rd-party language and its rely-guarantee proof systems are integrated into {\slang} as a new instance. Since these languages may have been developed independently from {\slang}, they are not necessarily completely consistent with the {\slang} interface. 
This architecture makes it possible to integrate existing languages without modifying their specification, semantics, and rely-guarantee proof system.

\subsection{Rely-guarantee Interface and Adapter}
The rely-guarantee interface requires specifications and assumptions for four differentiated elements: language definition (syntax and semantics), rely-guarantee definitions (validity), rely-guarantee proof rules, and their soundness. 
The detailed information is discussed as follows.

\begin{enumerate}

\item \emph{interface types}: they are data types used by the parameters. These are the types of the program, the program state and the static configuration for programs. 

\item \emph{interface parameters and types}: 
as a parametric framework, {\slang} does not define the syntax for the languages used in event bodies. However, it requires the representation for program termination, which is denoted as $\cmdfinal$ in {\slang} (Parameter 1 in {\tableprefix} \ref{tbl:interface}).
{\slang} also needs the transition relation representing the program language semantics (Parameters 2). Moreover, the interface requires the components related to the validity of a rely-guarantee specification (Parameter 3) and the proof rules (Parameter 4). 

\item \emph{interface assumptions}: 
to reason about event behaviours, {\slang} assumes that 
\begin{enumerate}
\item The program $\cmdfinal$ cannot take a step to another state (Assumption 1 in {\tableprefix} \ref{tbl:interface}).
\item if a program $\symbprog$ takes a component transition, $\symbprog$ is changed in the next configuration (Assumption 2).
\item the definition of the validity is similar to those in {\slang}, and are relaxed to be not necessarily equivalent (Assumption 3). 
\item the validity is usually defined by three functions, i.e. computation ($\compfun(\symbenv,\symbprog,\symbstate)$), assumption ($\assumefun_p(pre,rely)$) and commitment ($\commitfun_p(\symbenv, guar,post)$). The functions have the same form as of those in {\eqprefix} \ref{eq:assume_commit} (Page \pageref{eq:assume_commit}) by using the program transition relation (Parameter 2). {\slang} requires that the rely-guarantee proof rules in languages are sound (Assumption 4). 
\end{enumerate}

\end{enumerate}

\begin{table}[t]
  \centering
  \small 
  \caption{Parameters and Their Assumptions of \slang} 
  \begin{tabular} {|c|c|c|c|}
    \hline
    \textbf{No.} & \textbf{Parameter} & \textbf{Notation of Parameter} & \textbf{Assumption of Parameter} \\
    \hline
    (1) & Terminal statement & $\cmdfinal$ & \multirow{2}{*}{\tabincell{l}{(1) $\neg(\ptran{\cmdfinal}{s}{P}{t})$ \\ (2) $\neg(\ptran{P}{s}{P}{t})$}} \\
    \cline{1-3}
    (2) & Program transition & $\ptran{P}{s}{Q}{t}$ &  \\
    \hline
    (3) & Validity & $\RGSATp{\symbprog}{\rgconddefault}$ &
    \tabincell{l}{(3) $\RGSATp{\symbprog}{\rgconddefault}$ \\ $\Longrightarrow \forall \symbstate.\ \compfun(\symbenv,\symbprog,\symbstate) \cap \assumefun_p(pre,rely)$ \\ $\quad \quad \quad \quad \quad \subseteq \commitfun_p(\symbenv, guar,post)$}
    \\
    \hline
    (4) & Proof rule & $\rgsatp{P}{\rgconddefault}$ & \tabincell{l}{(4) $\rgsatp{P}{\rgconddefault}$ \\ $\Longrightarrow \RGSATp{\symbprog}{\rgconddefault}$} \\
    \hline
  \end{tabular}
  \label{tbl:interface}
\end{table}

The rely-guarantee adapter for a concrete language implements the rely-guarantee interface by delegating the functionality of the event language to the integrated language. 
In the adapter, we first need to instantiate the interface types to be \emph{adapter types}. 
For each interface parameter, there is a corresponding definition (or function) in the adapter to instantiate the parameter. 
Moreover, the adapter provides the necessary set of \emph{adapter conclusions} (lemmas and theorems) to show that it satisfies the interface assumptions. 

\subsection{Isabelle/HOL Implementation}

In the mechanized implementation of {\slang} in Isabelle/HOL, we design a set of hierarchical locales, where the final locale is shown in {\figprefix} \ref{fig:rginterface_impl}. 
The interface types are represented as ${\isacharprime}prog$ for the type of programs, ${\isacharprime}s$ for the type of program states, and ${\isacharprime}Env$ for the type of static configuration for programs represented by $\symbenv$. 
The four interface parameters are represented as $fin{\isacharunderscore}com$, $ptran$, $prog{\isacharunderscore}validity$ and  $rghoare{\isacharunderscore}p$ respectively. The four interface assumptions are represented as $none{\isacharunderscore}no{\isacharunderscore}tran{\isacharprime}$, $ptran{\isacharunderscore}neq$, $prog{\isacharunderscore}validity{\isacharunderscore}def$ and $rgsound{\isacharunderscore}p$ respectively. 

\begin{figure}[t] 
\begin{flushleft}
\isabellestyle{sl} 
\begin{isabellec} 

\isacommand{type{\isacharunderscore}synonym}\isamarkupfalse%
\ {\isacharparenleft}{\isacharprime}s{\isacharcomma}{\isacharprime}prog{\isacharparenright}\ pconf\ {\isacharequal}\ {\isachardoublequoteopen}{\isacharprime}prog\ {\isasymtimes}\ {\isacharprime}s{\isachardoublequoteclose}


\isacommand{locale}\isamarkupfalse \ PiCore\ {\isacharequal}

\ \ \isakeyword{fixes}\ fin{\isacharunderscore}com\ {\isacharcolon}{\isacharcolon}\ {\isachardoublequoteopen}{\isacharprime}prog{\isachardoublequoteclose}

\ \ \isakeyword{fixes}\ ptran\ {\isacharcolon}{\isacharcolon}\ {\isachardoublequoteopen}{\isacharprime}Env\ {\isasymRightarrow}\ {\isacharparenleft}{\isacharparenleft}{\isacharprime}s{\isacharcomma}{\isacharprime}prog{\isacharparenright}\ pconf\ {\isasymtimes}\ {\isacharparenleft}{\isacharprime}s{\isacharcomma}{\isacharprime}prog{\isacharparenright}\ pconf{\isacharparenright}\ set{\isachardoublequoteclose}

\ \ \isakeyword{fixes}\ prog{\isacharunderscore}validity\ {\isacharcolon}{\isacharcolon}\ {\isachardoublequoteopen}{\isacharprime}Env\ {\isasymRightarrow}\ {\isacharprime}prog\ {\isasymRightarrow}\ {\isacharprime}s\ set\ {\isasymRightarrow}\ {\isacharparenleft}{\isacharprime}s\ {\isasymtimes}\ {\isacharprime}s{\isacharparenright}\ set\ {\isasymRightarrow}\ {\isacharparenleft}{\isacharprime}s\ {\isasymtimes}\ {\isacharprime}s{\isacharparenright}\ set\ {\isasymRightarrow}\ {\isacharprime}s\ set\ {\isasymRightarrow}\ bool{\isachardoublequoteclose}

\quad \quad \quad \quad \quad \quad \quad \quad \quad 
{\isacharparenleft}{\isachardoublequoteopen}{\isacharunderscore}\ {\isasymTurnstile}\ {\isacharunderscore}\ sat\isactrlsub p\ {\isacharbrackleft}{\isacharunderscore}{\isacharcomma}\ {\isacharunderscore}{\isacharcomma}\ {\isacharunderscore}{\isacharcomma}\ {\isacharunderscore}{\isacharbrackright}{\isachardoublequoteclose} {\isacharparenright}

\ \ \isakeyword{fixes}\ rghoare{\isacharunderscore}p\ {\isacharcolon}{\isacharcolon}\ {\isachardoublequoteopen}{\isacharprime}Env\ {\isasymRightarrow}\ {\isacharbrackleft}{\isacharprime}prog{\isacharcomma}\ {\isacharprime}s\ set{\isacharcomma}\ {\isacharparenleft}{\isacharprime}s\ {\isasymtimes}\ {\isacharprime}s{\isacharparenright}\ set{\isacharcomma}\ {\isacharparenleft}{\isacharprime}s\ {\isasymtimes}\ {\isacharprime}s{\isacharparenright}\ set{\isacharcomma}\ {\isacharprime}s\ set{\isacharbrackright}\ {\isasymRightarrow}\ bool{\isachardoublequoteclose}
{\isacharparenleft}{\isachardoublequoteopen}{\isacharunderscore}\ {\isasymturnstile}\ {\isacharunderscore}\ sat\isactrlsub p\ {\isacharbrackleft}{\isacharunderscore}{\isacharcomma}\ {\isacharunderscore}{\isacharcomma}\ {\isacharunderscore}{\isacharcomma}\ {\isacharunderscore}{\isacharbrackright}{\isachardoublequoteclose} {\isacharparenright}

\ \ \isakeyword{assumes}\ none{\isacharunderscore}no{\isacharunderscore}tran{\isacharprime}{\isacharcolon}\ {\isachardoublequoteopen}{\isacharparenleft}{\isacharparenleft}fin{\isacharunderscore}com{\isacharcomma}\ s{\isacharparenright}{\isacharcomma}{\isacharparenleft}P{\isacharcomma}t{\isacharparenright}{\isacharparenright}\ {\isasymnotin}\ ptran\ {\isasymGamma}{\isachardoublequoteclose}

\ \ \isakeyword{assumes}\ ptran{\isacharunderscore}neq{\isacharcolon}\ {\isachardoublequoteopen}{\isacharparenleft}{\isacharparenleft}P{\isacharcomma}\ s{\isacharparenright}{\isacharcomma}{\isacharparenleft}P{\isacharcomma}t{\isacharparenright}{\isacharparenright}\ {\isasymnotin}\ ptran\ {\isasymGamma}{\isachardoublequoteclose}

\ \ \isakeyword{assumes}\ prog{\isacharunderscore}validity{\isacharunderscore}def{\isacharcolon}\ {\isachardoublequoteopen}{\isasymGamma}\ {\isasymTurnstile}\ P\ sat\isactrlsub p\ {\isacharbrackleft}pre{\isacharcomma}\ rely{\isacharcomma}\ guar{\isacharcomma}\ post{\isacharbrackright}

\quad \quad \quad \quad \quad \quad
{\isasymLongrightarrow}\ 
{\isasymforall}s{\isadigit{0}}{\isachardot}\ cpts{\isacharunderscore}from\ {\isacharparenleft}ptran\ {\isasymGamma}{\isacharparenright}\ {\isacharparenleft}P{\isacharcomma}s{\isadigit{0}}{\isacharparenright}\ {\isasyminter}\ assume\ pre\ rely\ {\isasymsubseteq}\ commit\ {\isacharparenleft}ptran\ {\isasymGamma}{\isacharparenright}\ {\isacharbraceleft}fin{\isacharunderscore}com{\isacharbraceright}\ guar\ post

\ \ \isakeyword{assumes}\ rgsound{\isacharunderscore}p{\isacharcolon}\ {\isachardoublequoteopen}{\isasymGamma}\ {\isasymturnstile}\ P\ sat\isactrlsub p\ {\isacharbrackleft}pre{\isacharcomma}\ rely{\isacharcomma}\ guar{\isacharcomma}\ post{\isacharbrackright}\ {\isasymLongrightarrow}\ {\isasymGamma}\ {\isasymTurnstile}\ P\ sat\isactrlsub p\ {\isacharbrackleft}pre{\isacharcomma}\ rely{\isacharcomma}\ guar{\isacharcomma}\ post{\isacharbrackright}{\isachardoublequoteclose}

\isakeyword{begin}%

...... {\color{ACMGreen} // we omit here the {\slang} implementation using parameters and assumptions. }

\isakeyword{end}%
\end{isabellec} 
	\caption{Isabelle/HOL Implementation of {\slang} Rely-guarantee Interface} 
	\label{fig:rginterface_impl}
\end{flushleft}
\end{figure}


Finally, the interpretation of the {\slang} locale is shown as follows, where \emph{finI}, \emph{ptranI}, \emph{validityI} and \emph{rghoareI} are adapter parameters of a concrete language $I$ passed to the locale. The \textbf{interpretation} command in Isabelle will generate proof goals showing the satisfaction of assumptions. 
$$
\isacommand{interpretation} \ {\slang} \ finI \ ptranI \ validityI \ rghoareI
$$

\section{Integrating Concrete Languages}
\label{sect:integrate}

Using the {\slang} framework and its Isabelle/HOL implementation, the remaining work to integrate a concrete language is to design its adapter.  We have created two {\slang} instances by interpreting the {\slang} locale using adapters for {\implang} and {\csimpllang}. This section presents the two adapters, together with an overview of the integrated language for each adapter. 

\subsection{{\pccsimpl}: Integrating {\csimpllang} Language}

\subsubsection{{\csimpllang} Overview} 
The {\csimpllang} language and its rely-guarantee proof system are presented in detail in \cite{Sanan17}. 
The abstract syntax of {\csimpllang} is defined as in {\figprefix} \ref{fig:csimpl_syntax} in terms of states, of type ${\isacharprime}s$; a set of fault types, of type ${\isacharprime}f$; a set of procedure names of type ${\isacharprime}p$, and a set of simulation events ${\isacharprime}e$. 

In order to capture all aspects of abrupt termination, assertions, and function calls, the program state in {\csimpllang} is modelled as a datatype $xstate$, which is composed of four different constructors to represent a regular execution, a failed assertion, an exceptional state and a state where a call to a non-defined function is made. 
Type $({\isacharprime}s,{\isacharprime}p,{\isacharprime}f,{\isacharprime}e)\ config$ defines the configuration used in its transition semantics and $({\isacharprime}s,{\isacharprime}p,{\isacharprime}f,{\isacharprime}e)\ body$ denoted as an environment $\csimplprocenv$ defines the procedure declarations as a partial function from the set ${\isacharprime}p$ of procedure names to the body of the procedures. $({\isacharprime}s,{\isacharprime}p,{\isacharprime}f,{\isacharprime}e)\ confs$ defines the type of computations. To support reasoning about procedure invocations, {\csimpllang} uses the notation $\csimplprocrg$ to maintain the rely-guarantee specification for procedures. The validity in {\csimpllang} requires that each procedure in $\csimplprocrg$ satisfies its specification.

\begin{figure}[t] 
\begin{flushleft}
\isabellestyle{sl}
	\begin{isabellebody} 
		\small 
		\isacommand{datatype}\isamarkupfalse%
		\ {\isacharparenleft}{\isacharprime}s{\isacharcomma}\ {\isacharprime}p{\isacharcomma}\ {\isacharprime}f{\isacharcomma}\ {\isacharprime}e{\isacharparenright}\ com\ {\isacharequal} 
		\ \ Skip {\isacharbar}\ Throw {\isacharbar}\ Basic\ \ {\isachardoublequoteopen}{\isacharprime}s\ {\isasymRightarrow}\ {\isacharprime}s{\isachardoublequoteclose}\ \ {\isacharprime}e\ option \ \ {\isacharbar}\ Spec\ \ {\isachardoublequoteopen}{\isacharparenleft}{\isacharprime}s\ {\isasymtimes}\ {\isacharprime}s{\isacharparenright}\ set {\isachardoublequoteclose}\ \ {\isachardoublequoteopen}{\isacharprime}e\ option{\isachardoublequoteclose}
		
		\ \ {\isacharbar}\ Seq\ \ {\isachardoublequoteopen}{\isacharparenleft}{\isacharprime}s\ {\isacharcomma}{\isacharprime}p{\isacharcomma}\ {\isacharprime}f{\isacharcomma}\ {\isacharprime}e{\isacharparenright}\ com{\isachardoublequoteclose}\ \ {\isachardoublequoteopen}{\isacharparenleft}{\isacharprime}s{\isacharcomma}\ {\isacharprime}p{\isacharcomma}\ {\isacharprime}f{\isacharcomma}\ {\isacharprime}e{\isacharparenright}\ com{\isachardoublequoteclose}\ {\isacharbar}\ Await\ \ {\isachardoublequoteopen}{\isacharprime}s\ bexp{\isachardoublequoteclose}\ \ {\isacharprime}e\ option \ \ {\isachardoublequoteopen}{\isacharparenleft}{\isacharprime}s{\isacharcomma}\ {\isacharprime}p{\isacharcomma}\ {\isacharprime}f{\isacharcomma}\ {\isacharprime}e{\isacharparenright}\ com{\isachardoublequoteclose} 
		
		\ \ {\isacharbar}\ Cond\ \ {\isachardoublequoteopen}{\isacharprime}s\ bexp{\isachardoublequoteclose}\ \ {\isachardoublequoteopen}{\isacharparenleft}{\isacharprime}s{\isacharcomma}\ {\isacharprime}p{\isacharcomma}\ {\isacharprime}f{\isacharcomma}\ {\isacharprime}e{\isacharparenright}\ com{\isachardoublequoteclose}\ \ \ {\isachardoublequoteopen}{\isacharparenleft}{\isacharprime}s{\isacharcomma}\ {\isacharprime}p{\isacharcomma}\ {\isacharprime}f{\isacharcomma}\ {\isacharprime}e{\isacharparenright}\ com{\isachardoublequoteclose}
		
		\ \ {\isacharbar}\ While\ \ {\isachardoublequoteopen}{\isacharprime}s\ bexp{\isachardoublequoteclose}\ \ {\isachardoublequoteopen}{\isacharparenleft}{\isacharprime}s{\isacharcomma}\ {\isacharprime}p{\isacharcomma}\ {\isacharprime}f{\isacharcomma}\ {\isacharprime}e{\isacharparenright}\ com{\isachardoublequoteclose} {\isacharbar}\ Call\ \ {\isachardoublequoteopen}{\isacharprime}p{\isachardoublequoteclose}\ 
		\ \ {\isacharbar}\ DynCom\ \ {\isachardoublequoteopen}{\isacharprime}s\ {\isasymRightarrow}\ {\isacharparenleft}{\isacharprime}s{\isacharcomma}\ {\isacharprime}p{\isacharcomma}\ {\isacharprime}f{\isacharcomma}\ {\isacharprime}e{\isacharparenright}\ com{\isachardoublequoteclose}
		
		\ \ {\isacharbar}\ Guard\ \ {\isachardoublequoteopen}{\isacharprime}f{\isachardoublequoteclose}\ \ {\isachardoublequoteopen}{\isacharprime}s\ bexp{\isachardoublequoteclose}\ \ {\isachardoublequoteopen}{\isacharparenleft}{\isacharprime}s{\isacharcomma}\ {\isacharprime}p{\isacharcomma}\ {\isacharprime}f{\isacharcomma}\ {\isacharprime}e{\isacharparenright}\ com{\isachardoublequoteclose}\ 
		\ \ {\isacharbar}\ Catch\ {\isachardoublequoteopen}{\isacharparenleft}{\isacharprime}s{\isacharcomma}\ {\isacharprime}p{\isacharcomma}\ {\isacharprime}f{\isacharcomma}\ {\isacharprime}e{\isacharparenright}\ com{\isachardoublequoteclose}\ \ {\isachardoublequoteopen}{\isacharparenleft}{\isacharprime}s{\isacharcomma}\ {\isacharprime}p{\isacharcomma}\ {\isacharprime}f{\isacharcomma}\ {\isacharprime}e{\isacharparenright}\ com{\isachardoublequoteclose}

		\isacommand{datatype}\isamarkupfalse%
		\ {\isacharparenleft}{\isacharprime}s{\isacharcomma}{\isacharprime}f{\isacharparenright}\ xstate\ {\isacharequal}\ Normal\ {\isacharprime}s\ {\isacharbar}\ Abrupt\ {\isacharprime}s\ {\isacharbar}\ Fault\ {\isacharprime}f\ {\isacharbar}\ Stuck
		
		\isacommand{type{\isacharunderscore}synonym}\isamarkupfalse%
		{\isacharparenleft}{\isacharprime}s{\isacharcomma}\ {\isacharprime}p{\isacharcomma}\ {\isacharprime}f{\isacharcomma}\ {\isacharprime}e{\isacharparenright}\ config\ {\isacharequal}\ {\isachardoublequoteopen}{\isacharparenleft}{\isacharprime}s{\isacharcomma}\ {\isacharprime}p{\isacharcomma}\ {\isacharprime}f{\isacharcomma}\ {\isacharprime}e{\isacharparenright}com\ \ {\isasymtimes}\ {\isacharparenleft}{\isacharprime}s{\isacharcomma}{\isacharprime}f{\isacharparenright}\ xstate{\isachardoublequoteclose}
		
		\isacommand{type{\isacharunderscore}synonym}\isamarkupfalse%
		\ {\isacharparenleft}{\isacharprime}s{\isacharcomma}\ {\isacharprime}p{\isacharcomma}\ {\isacharprime}f{\isacharcomma}\ {\isacharprime}e{\isacharparenright}\ body\ {\isacharequal}\ {\isachardoublequoteopen}{\isacharprime}p\ {\isasymRightarrow}\ {\isacharparenleft}{\isacharprime}s{\isacharcomma}\ {\isacharprime}p{\isacharcomma}\ {\isacharprime}f{\isacharcomma}\ {\isacharprime}e{\isacharparenright}\ com\ option{\isachardoublequoteclose}	
		
		\isacommand{type{\isacharunderscore}synonym}\isamarkupfalse%
\ {\isacharparenleft}{\isacharprime}s{\isacharcomma}{\isacharprime}p{\isacharcomma}{\isacharprime}f{\isacharcomma}{\isacharprime}e{\isacharparenright}\ confs\ {\isacharequal}\ {\isachardoublequoteopen}{\isacharparenleft}{\isacharprime}s{\isacharcomma}{\isacharprime}p{\isacharcomma}{\isacharprime}f{\isacharcomma}{\isacharprime}e{\isacharparenright}\ body\ {\isasymtimes}{\isacharparenleft}{\isacharparenleft}{\isacharprime}s{\isacharcomma}{\isacharprime}p{\isacharcomma}{\isacharprime}f{\isacharcomma}{\isacharprime}e{\isacharparenright}\ config{\isacharparenright}\ list{\isachardoublequoteclose}	
	\end{isabellebody}%
	\caption{Syntax and State Definition of the {\csimpllang} Language in Isabelle/HOL \cite{Sanan17}} 
	\label{fig:csimpl_syntax}
\end{flushleft}
\end{figure}

%
%
%

\subsubsection{Rely-guarantee Adapter for {\csimpllang}}

\paragraph{Types}
Instantiating the interface types in {\csimpllang} is straightforward. The interface type of programs ${\isacharprime}prog$ is implemented as type ${\isacharparenleft}{\isacharprime}s{\isacharcomma}\ {\isacharprime}p{\isacharcomma}\ {\isacharprime}f{\isacharcomma}\ {\isacharprime}e{\isacharparenright}\ com$, and interface type of program state ${\isacharprime}s$ as type ${\isacharparenleft}{\isacharprime}s{\isacharcomma}{\isacharprime}f{\isacharparenright}\ xstate$. 
The procedure declarations $\csimplprocenv$ and their rely-guarantee specification $\csimplprocrg$ are the static configuration of {\csimpllang}. Thus, we use type $\csimplprocenv \times \csimplprocrg$ to instantiate the environment ${\isacharprime}Env$ in {\slang}. 

\paragraph{Parameters}
The interface parameter of the terminal statement is instantiated as the $Skip$ command in {\csimpllang}. 
The program transition in {\csimpllang} is defined as $\csimpltran{P}{s}{Q}{t}$. 
The interface parameter of program transition is instantiated as $(\Gamma,\Theta) \vdash_{cI} (P,s) \rightarrow (Q,t) \equiv \csimpltran{P}{s}{Q}{t}$ by delegating it to the {\csimpllang} semantics. 

Instantiating the interface parameter of the validity has to consider the \emph{computation}, \emph{assumption} and \emph{commitment} functions. The validity function in {\csimpllang} is shown in {\figprefix} \ref{fig:csimpl_validity}. The rely-guarantee specification in {\csimpllang} is in the form $\rgcond{pre}{rely}{guar}{(q,a)}$ with a set of \emph{Fault} states $F$, where the postcondition $(q,a)$ is a pair of state sets. The set $q$ constrains the final state if the program terminates as \emph{Skip} representing a normal state, whilst $a$ constrains abrupt terminations in an exception with the command \emph{Throw}. The set $F$ constrains that the final state of the program should not be in \emph{Fault} states in $F$. 
For procedure invocations, {\csimpllang} defines another validity function using the general one, which also requires that each procedure satisfies its rely-guarantee specification.

First, the computation function in {\csimpllang} is defined by using the program transition and the environment transition. As discussed above, the program transition of {\csimpllang} is adapted straightforwardly. 
{\csimpllang} semantics for programs can transit from a \emph{Normal} state to a different type. However, it does not allow transitions from a non \emph{Normal} state to any other state. 
Therefore, the environment transition in {\csimpllang} is defined as follows. 
To adapt the restricted environment transition, we first define the environment transition in the adapter as $(\Gamma,\Theta) \vdash_{cI} (P, s) \rightsquigarrow (Q,t) \equiv P = Q$, which allows any state transition and is compatible with that in the rely-guarantee interface. Then, we restrict the rely condition in the definition of proof rules in the adapter to bridge this difference, which will be discussed later. Based on the transition functions, the computation function $\compfun$ of the adapter is defined in the same form as in {\csimpllang}. 
\[
\left\{
\begin{aligned}
& \csimplenvtran{P}{Normal\ s}{P}{t} \\
& (\forall t'.\ t \neq Normal\ t') \Longrightarrow \csimplenvtran{P}{t}{P}{t}
\end{aligned}
\right.
\]

Second, the validity of {\csimpllang} only concerns the preconditions over \emph{Normal} states. Thus, the \emph{assumption} function of {\csimpllang} as shown in {\figprefix} \ref{fig:csimpl_validity} restricts the starting state of a computation to \emph{Normal}. 
For type consistency, {\slang} does not impose that restriction, but rather it is enforced by the adapter to bridge the difference. So in the adapter, we define the \emph{assumption} function as the same as in {\slang}. Then, we restrict the precondition in the definition of proof rules in the adapter, which will be discussed later.

\begin{figure}[t] 
\begin{flushleft}
\isabellestyle{sl}
\begin{isabellec}

assum {\isacharparenleft}pre, \ rely{\isacharparenright}\ {\isasymequiv}\ {\isacharbraceleft}c{\isachardot}\ snd{\isacharparenleft}{\isacharparenleft}snd\ c{\isacharparenright}{\isacharbang}{\isadigit{0}}{\isacharparenright}\ {\isasymin}\ Normal\ {\isacharbackquote}\ pre\ {\isasymand}\ {\isacharparenleft}{\isasymforall}i{\isachardot}\ Suc\ i{\isacharless}length\ {\isacharparenleft}snd\ c{\isacharparenright}\ {\isasymlongrightarrow}

\ \ \ \ \ \ \ \ \ \ \ \ \ \ \ \ \ {\isacharparenleft}fst\ c{\isacharparenright}{\isasymturnstile}\isactrlsub c{\isacharparenleft}{\isacharparenleft}snd\ c{\isacharparenright}{\isacharbang}i{\isacharparenright}\ \ $\rightsquigarrow$\ {\isacharparenleft}{\isacharparenleft}snd\ c{\isacharparenright}{\isacharbang}{\isacharparenleft}Suc\ i{\isacharparenright}{\isacharparenright}\ {\isasymlongrightarrow}\
{\isacharparenleft}snd{\isacharparenleft}{\isacharparenleft}snd\ c{\isacharparenright}{\isacharbang}i{\isacharparenright}{\isacharcomma}\ snd{\isacharparenleft}{\isacharparenleft}snd\ c{\isacharparenright}{\isacharbang}{\isacharparenleft}Suc\ i{\isacharparenright}{\isacharparenright}{\isacharparenright}\ {\isasymin}\ \ rely{\isacharparenright}{\isacharbraceright}

comm\ {\isacharparenleft}guar{\isacharcomma}\ {\isacharparenleft}q{\isacharcomma}a{\isacharparenright}{\isacharparenright}\ F \ {\isasymequiv}\ {\isacharbraceleft}c{\isachardot}\ snd\ {\isacharparenleft}last\ {\isacharparenleft}snd\ c{\isacharparenright}{\isacharparenright}\ {\isasymnotin}\ Fault\ {\isacharbackquote}\ F\ \ {\isasymlongrightarrow}\ {\isacharparenleft}{\isasymforall}i{\isachardot}\ Suc\ i{\isacharless}length\ {\isacharparenleft}snd\ c{\isacharparenright}\ {\isasymlongrightarrow}

\ \ \ \ \ \ \ \ \ \ \ \ \ \ \ \ \ {\isacharparenleft}fst\ c{\isacharparenright}{\isasymturnstile}\isactrlsub c{\isacharparenleft}{\isacharparenleft}snd\ c{\isacharparenright}{\isacharbang}i{\isacharparenright}\ \ {\isasymrightarrow}\ {\isacharparenleft}{\isacharparenleft}snd\ c{\isacharparenright}{\isacharbang}{\isacharparenleft}Suc\ i{\isacharparenright}{\isacharparenright}\ {\isasymlongrightarrow}\ {\isacharparenleft}snd{\isacharparenleft}{\isacharparenleft}snd\ c{\isacharparenright}{\isacharbang}i{\isacharparenright}{\isacharcomma}\ snd{\isacharparenleft}{\isacharparenleft}snd\ c{\isacharparenright}{\isacharbang}{\isacharparenleft}Suc\ i{\isacharparenright}{\isacharparenright}{\isacharparenright}\ {\isasymin}\ guar{\isacharparenright}\ {\isasymand}

\ \ \ \ \ \ \ \ \ \ \ \ \ \ \ \ \ {\isacharparenleft}final\ {\isacharparenleft}last\ {\isacharparenleft}snd\ c{\isacharparenright}{\isacharparenright}\ \ {\isasymlongrightarrow}\ {\isacharparenleft}{\isacharparenleft}fst\ {\isacharparenleft}last\ {\isacharparenleft}snd\ c{\isacharparenright}{\isacharparenright}\ {\isacharequal}\ Skip\ {\isasymand}\ snd\ {\isacharparenleft}last\ {\isacharparenleft}snd\ c{\isacharparenright}{\isacharparenright}\ {\isasymin}\ Normal\ {\isacharbackquote}\ q{\isacharparenright}{\isacharparenright}\ {\isasymor}

\quad \quad \quad \quad \quad \quad \quad 
\quad \quad \quad \quad \quad \quad \quad \ \ 
{\isacharparenleft}fst\ {\isacharparenleft}last\ {\isacharparenleft}snd\ c{\isacharparenright}{\isacharparenright}\ {\isacharequal}\ Throw\ {\isasymand}\ snd\ {\isacharparenleft}last\ {\isacharparenleft}snd\ c{\isacharparenright}{\isacharparenright}\ {\isasymin}\ Normal\ {\isacharbackquote}\ a{\isacharparenright}{\isacharparenright}{\isacharbraceright}

{\isasymGamma}\ {\isasymTurnstile}\isactrlbsub {\isacharslash}F\isactrlesub \ Pr\ sat\ {\isacharbrackleft}pre{\isacharcomma}\ rely{\isacharcomma}\ guar{\isacharcomma}\ q{\isacharcomma}a{\isacharbrackright}\ {\isasymequiv}\ {\isasymforall}s{\isachardot}\ cp\ {\isasymGamma}\ Pr\ s\ {\isasyminter}\ assum{\isacharparenleft}pre{\isacharcomma}\ rely{\isacharparenright}\ {\isasymsubseteq}\ comm{\isacharparenleft}guar{\isacharcomma}\ {\isacharparenleft}q{\isacharcomma}a{\isacharparenright}{\isacharparenright}\ F

{\isasymGamma}{\isacharcomma}{\isasymTheta}\ {\isasymTurnstile}\isactrlbsub {\isacharslash}F\isactrlesub \ Pr\ sat\ {\isacharbrackleft}pre{\isacharcomma}\ rely{\isacharcomma}\ guar{\isacharcomma}\ q{\isacharcomma}a{\isacharbrackright}\ {\isasymequiv}

\quad \quad \quad \quad 
{\isacharparenleft}{\isasymforall}{\isacharparenleft}c{\isacharcomma}p{\isacharcomma}R{\isacharcomma}G{\isacharcomma}q{\isacharcomma}a{\isacharparenright}{\isasymin}\ {\isasymTheta}{\isachardot}\ {\isasymGamma}\ {\isasymTurnstile}\isactrlbsub {\isacharslash}F\isactrlesub \ {\isacharparenleft}Call\ c{\isacharparenright}\ sat\ {\isacharbrackleft}p{\isacharcomma}\ R{\isacharcomma}\ G{\isacharcomma}\ q{\isacharcomma}a{\isacharbrackright}{\isacharparenright}\ {\isasymlongrightarrow}\
{\isasymGamma}\ {\isasymTurnstile}\isactrlbsub {\isacharslash}F\isactrlesub \ Pr\ sat\ {\isacharbrackleft}pre{\isacharcomma}\ rely{\isacharcomma}\ guar{\isacharcomma}\ q{\isacharcomma}a{\isacharbrackright}

\end{isabellec}
	\caption{The Validity in the {\csimpllang} Language in Isabelle/HOL \cite{Sanan17}} 
	\label{fig:csimpl_validity}
\end{flushleft}
\end{figure}

Third, referring to semantics of the final state in the validity of {\csimpllang}, the \emph{commitment} function in {\csimpllang} is defined as shown in {\figprefix} \ref{fig:csimpl_validity}. {\slang} does restrict the final statement to \emph{Skip}, thus exceptions have to be handled at the program level. In the adapter, we define the same \emph{commitment} function as in {\slang}. Then, we restrict the postcondition in the definition of proof rules in the adapter, which will be discussed later. 


Finally, the validity parameter of the {\csimpllang} adapter is define as $({\isasymGamma},{\isasymTheta})\ {\isasymTurnstile}\isactrlsub I\ P\ sat\isactrlsub p\ {\isacharbrackleft}pre{\isacharcomma}\ rely{\isacharcomma}\ guar{\isacharcomma}\ q{\isacharbrackright}$, which has the same form as in the rely-guarantee interface. 

Considering the interface parameter of the proof rule, based on the definition of proof rules {{\isasymGamma}{\isacharcomma}{\isasymTheta}\ {\isasymturnstile}\isactrlbsub {\isacharslash}F\isactrlesub \ P\ sat\ {\isacharbrackleft}pre{\isacharcomma}\ rely{\isacharcomma}\ guar{\isacharcomma}\ q{\isacharcomma}a{\isacharbrackright}} in {\csimpllang}, we adapt proof rules as follows. (1) The validity in {\csimpllang} only concerns preconditions of \emph{Normal} states, so we restrict the precondition $pre$ to \emph{Normal}. (2) Programs of an event body cannot throw exceptions to the event level, so final states when reaching the final statement \emph{Skip} are \emph{Normal}. Thus, we restrict the postcondition $q$ to \emph{Normal}. (3) Events assume the normal execution of their program body, and furthermore the program cannot fall into a \emph{Fault} state. So we assume the \emph{Fault} set $F$ to be empty. In addition, the program $P$ should satisfy its rely-guarantee specification in {\csimpllang}. The abrupt postcondition is set to \emph{False}, to ensure that the CSimpl program never ends in an exceptional state. (4) The environment transition in {\csimpllang} does not allow transitions from a non \emph{Normal} state to a different 	state, we represent it in the rely condition $rely$. (5) Finally, the rely-guarantee specification for each procedure in $\Theta$ has to be satisfied.

{ \isabellestyle{sl}
\begin{isabellebody}
\fontsize{8pt}{0cm} 
\quad {\isachardoublequoteopen}({\isasymGamma},{\isasymTheta})\ {\isasymturnstile}\isactrlsub I\ P\ sat\isactrlsub p\ {\isacharbrackleft}pre{\isacharcomma}\ rely{\isacharcomma}\ guar{\isacharcomma}\ q{\isacharbrackright}\ {\isasymequiv}\ 
$\overbrace{\text{{\isacharparenleft}pre\ {\isasymsubseteq}\ Normal\ {\isacharbackquote}\ UNIV{\isacharparenright}}}^{(1)}$\ {\isasymand}\ 
$\overbrace{\text{{\isacharparenleft}q\ {\isasymsubseteq}\ Normal\ {\isacharbackquote}\ UNIV{\isacharparenright}}}^{(2)}$
\ {\isasymand} \isanewline
$\overbrace{\text{{\isacharparenleft}{\isasymGamma}{\isacharcomma}{\isasymTheta}\ {\isasymturnstile}\isactrlbsub {\isacharslash}{\isacharbraceleft}{\isacharbraceright}\isactrlesub \ P\ sat\ {\isacharbrackleft}{\isacharbraceleft}s $\mid$ Normal\ s\ {\isasymin}\ pre{\isacharbraceright}{\isacharcomma}\ rely{\isacharcomma}\ guar{\isacharcomma}\ {\isacharbraceleft}s $\mid$ Normal\ s\ {\isasymin}\ q{\isacharbraceright}{\isacharcomma}\ \{\ \} {\isacharbrackright}{\isacharparenright}}}^{(3)}$
{\isasymand}\ \isanewline
$\overbrace{\text{{\isacharparenleft}{\isasymforall}{\isacharparenleft}s{\isacharcomma}t{\isacharparenright}{\isasymin}rely{\isachardot}\ s\ {\isasymnotin}\ Normal\ {\isacharbackquote}\ UNIV\ {\isasymlongrightarrow}\ s\ {\isacharequal}\ t{\isacharparenright}}}^{(4)}$
{\isasymand}\  
$\overbrace{\text{{\isacharparenleft}{\isasymforall}{\isacharparenleft}c{\isacharcomma}p{\isacharcomma}R{\isacharcomma}G{\isacharcomma}q{\isacharcomma}a{\isacharparenright}{\isasymin}\ {\isasymTheta}{\isachardot}\ {\isasymGamma}{\isacharcomma}{\isacharbraceleft}{\isacharbraceright}\ {\isasymturnstile}\isactrlbsub {\isacharslash}{\isacharbraceleft}{\isacharbraceright}\isactrlesub \ {\isacharparenleft}Call\ c{\isacharparenright}\ sat\ {\isacharbrackleft}p{\isacharcomma}\ R{\isacharcomma}\ G{\isacharcomma}\ q{\isacharcomma}a{\isacharbrackright}{\isacharparenright}}}^{(5)}$
{\isachardoublequoteclose}
\end{isabellebody}
}

\paragraph{Assumptions}
As we discuss the adapter parameters of {\csimpllang}, it is straightforward to show that the first three interface assumptions ({\tableprefix} \ref{tbl:interface}) are satisfied in the adapter. To show the fourth assumption, it is enough to prove that

{ \isabellestyle{sl} \centering
\begin{isabellebody}

({\isasymGamma},{\isasymTheta})\ {\isasymturnstile}\isactrlsub I\ P\ sat\isactrlsub p\ {\isacharbrackleft}pre{\isacharcomma}\ rely{\isacharcomma}\ guar{\isacharcomma}\ q{\isacharbrackright} $\Longrightarrow$ ({\isasymGamma},{\isasymTheta})\ {\isasymTurnstile}\isactrlsub I\ P\ sat\isactrlsub p\ {\isacharbrackleft}pre{\isacharcomma}\ rely{\isacharcomma}\ guar{\isacharcomma}\ q{\isacharbrackright}

\end{isabellebody}
}


\subsection{{\pcimp}: Integrating {\implang} Language}
\label{subsect:pcimp}
\subsubsection{{\implang} Overview}
{\implang} is a simple imperative language with a rely-guarantee proof system \cite{Nieto03} provided by the Hoare\_Parallel library of Isabelle/HOL. We briefly overview the language since we will use it and its {\slang} instance for the case studies later. 
The abstract and concrete syntax of {\implang} is shown as follows. Since there is no terminal statement in the abstract syntax, the configuration used in the semantics of {\implang} is defined as ${\isacharparenleft}{\isacharparenleft}{\isacharprime}s\ com{\isacharparenright}\ option{\isacharparenright}\ {\isasymtimes}\ {\isacharprime}s$. The terminal statement is denoted as the Isabelle/HOL construct \emph{None}. However, the proof rules of {\implang} denoted as $\rgsat{P}{\rgcond{pre}{rely}{guar}{post}}$ still use the program $P$ of type ${\isacharprime}s\ com$. 

\isabellestyle{sl}
	\begin{isabellebody} 
		\small 
		\isacommand{datatype}\isamarkupfalse%
		\ {\isacharprime}s\ com\ {\isacharequal} 
		\ \ Basic\ {\isachardoublequoteopen}{\isacharprime}s\ {\isasymRightarrow}\ {\isacharprime}s{\isachardoublequoteclose} 
		\quad {\isacharparenleft}{\isasymacute}{\isacharunderscore}\ {\isacharcolon}{\isacharequal}\ {\isacharunderscore}{\isacharparenright}
		
\quad \quad \quad \quad \quad \quad \quad \ \ \ 
{\isacharbar}\ Seq\ {\isachardoublequoteopen}{\isacharprime}s\ com{\isachardoublequoteclose}\ {\isachardoublequoteopen}{\isacharprime}s\ com{\isachardoublequoteclose} 
\quad {\isacharparenleft}{\isacharunderscore}{\isacharsemicolon}{\isacharsemicolon}\ {\isacharunderscore}{\isacharparenright}

\quad \quad \quad \quad \quad \quad \quad \ \ \ 
{\isacharbar}\ Await\ {\isachardoublequoteopen}{\isacharprime}s\ bexp{\isachardoublequoteclose}\ {\isachardoublequoteopen}{\isacharprime}s\ com{\isachardoublequoteclose} 
\quad {\isacharparenleft}\textbf{AWAIT}\ {\isacharunderscore}\ \textbf{THEN}\ {\isacharunderscore}\ \textbf{END}{\isacharparenright}
		
\quad \quad \quad \quad \quad \quad \quad \ \ \  
{\isacharbar}\ Cond\ {\isachardoublequoteopen}{\isacharprime}s\ bexp{\isachardoublequoteclose}\ {\isachardoublequoteopen}{\isacharprime}s\ com{\isachardoublequoteclose}\ \ {\isachardoublequoteopen}{\isacharprime}s\ com{\isachardoublequoteclose}
\quad {\isacharparenleft}\textbf{IF}\ {\isacharunderscore}\ \textbf{THEN}\ {\isacharunderscore}\ \textbf{ELSE}\ {\isacharunderscore}\textbf{FI}{\isacharparenright}

\quad \quad \quad \quad \quad \quad \quad \ \ \  
{\isacharbar}\ While\ {\isachardoublequoteopen}{\isacharprime}s\ bexp{\isachardoublequoteclose}\ {\isachardoublequoteopen}{\isacharprime}s\ com{\isachardoublequoteclose}
\quad {\isacharparenleft}\textbf{WHILE}\ {\isacharunderscore}\ \textbf{DO}\ {\isacharunderscore}\ \textbf{OD}{\isacharparenright}

\isacommand{type{\isacharunderscore}synonym}\isamarkupfalse%
		{\isacharprime}s\ conf\ {\isacharequal}\ {\isachardoublequoteopen}{\isacharparenleft}{\isacharparenleft}{\isacharprime}s\ com{\isacharparenright}\ option{\isacharparenright}\ {\isasymtimes}\ {\isacharprime}s{\isachardoublequoteclose}
		
	\end{isabellebody}%

\subsubsection{Rely-guarantee Adapter for {\implang}}
In the {\implang} adapter, we use ${\isacharparenleft}{\isacharprime}s\ com{\isacharparenright}\ option$ to instantiate the interface type of programs ${\isacharprime}prog$. Since there is no static information $\symbenv$ in {\implang}, we omit this information in the program transition and validity of the adapter. 
We define the proof rules as $\rgsatpI{(Some \ P)}{\rgcond{pre}{rely}{guar}{post}}$ using the original {\implang} proof rules and thus reuse all of soundness proofs to prove the interface assumption (4) in {\tableprefix} \ref{tbl:interface} (Page \pageref{tbl:interface}).  
Other definitions of rely-guarantee components in {\implang} respect to the {\slang} interface are straightforward, and we omit the details here. The interested reader can refer to the Isabelle/HOL sources. 

\section{Formal Verification of Concurrent Memory Management}
\label{sect:mem}

Aiming at the highest evaluation assurance level (EAL 7) evaluation of Common Criteria (CC) \cite{cc}, we develop a low-level design specification of the Zephyr concurrent memory management by using {\pcimp}, the instance of {\slang} for the {\implang} language. The specification closely follows the Zephyr C code, and thus is able to do the \emph{code-to-spec} review required by the EAL 7 evaluation, covering all the data structures and imperative statements present in the implementation. 

The verification conducted in this case study is on Zephyr v1.8.0, released in 2017. The C code of the buddy memory management is $\approx$ 400 lines, not counting blank lines and comments. During the formal verification, we found 3 bugs in the C code of Zephyr: \emph{an incorrect block split}, \emph{an incorrect return}, and \emph{non-termination of a loop} in the \emph{k\_mem\_pool\_alloc} service. 
The first two bugs are critical and have been repaired in the latest release of Zephyr. 

{\slang} has support for concurrent OS kernels features like modelling shared-variable concurrency of multiple threads, interruptable execution of handlers, self-suspending threads, and rescheduling. 
An OS kernel can be considered as a CRS of peer reaction structure. It makes an initialization and enters in an \emph{idle} loop until it receives an interruption which is handled by an interrupt handler. After finishing the execution of the handler, the OS kernel returns to the idle loop waiting for new interruptions. 
OS kernels may provide a large number of services to applications. The peer reaction construct in {\slang} can easily model the behaviour of OS kernels.

\subsection{Concurrent Buddy Memory Management of Zephyr}

In Zephyr, a memory pool is a kernel object that allows memory blocks to be dynamically allocated, from a designated memory region, and released back into the pool. Its definition in the C code is shown as follows. 
A memory pool's buffer ($*buf$) is an $n\_max$-size array of blocks of $max\_sz$ bytes at level $0$, with no wasted space between them. The size of the buffer is thus $n\_max \times max\_sz$ bytes long. Zephyr tries to accomplish a memory request by splitting available blocks into smaller ones fitting as best as possible the requested size. 
Each ``level 0'' block is a quad-block that can be split into four smaller ``level 1'' blocks of equal size. Likewise, each level 1 block is itself a quad-block that can be split again. At each level, the four smaller blocks become \emph{buddies} or \emph{partners} to each other. The block size at level $l$ is thus $max\_sz / 4 ^ l$.

\lstdefinestyle{customc}{
  belowcaptionskip=1\baselineskip,
  breaklines=true,
  frame=none, 
  xleftmargin=0pt, 
  language=C,
  numbers=none,
  stepnumber=1,
  showstringspaces=false,
  basicstyle=\scriptsize\ttfamily, 
  keywordstyle=\bfseries\color{green!40!black},
  commentstyle=\itshape\color{purple!40!black},
  identifierstyle=\color{blue},
  stringstyle=\color{orange},
}
\lstset{escapechar=@,style=customc}
\vspace{-4mm}
\begin{minipage}[t]{1.0\textwidth}
\begin{minipage}[t]{0.48\textwidth}
\begin{lstlisting}
struct k_mem_block_id {
  u32_t pool : 8;
  u32_t level : 4;
  u32_t block : 20;
};
struct k_mem_pool_lvl {
  union {
    u32_t *bits_p;
    u32_t bits;
  };
  sys_dlist_t free_list;
};
\end{lstlisting}
\end{minipage}
~
\begin{minipage}[t]{0.48\textwidth}
\begin{lstlisting}
struct k_mem_block {
  void *data;
  struct k_mem_block_id id;
};
struct k_mem_pool {
  void *buf;
  size_t max_sz;
  u16_t n_max;
  u8_t n_levels;
  u8_t max_inline_level;
  struct k_mem_pool_lvl *levels;
  _wait_q_t wait_q;
};
\end{lstlisting}
\end{minipage}
\end{minipage}

The pool is initially configured with the parameters  $n\_max$ and $max\_sz$, together with a third parameter $min\_sz$. $min\_sz$ defines the minimum size for an allocated block and must be at least $4 \times X$ ($X > 0$) bytes long. Memory pool blocks are recursively split into quarters until blocks of the minimum size are obtained, at which point no further split can occur. 
The depth at which $min\_sz$ blocks are allocated is $n\_levels$ and satisfies that $n\_max = min\_sz \times 4 ^ {n\_levels}$.

Every memory block is composed of a $level$; a $block$ index within the level, ranging from $0$ to $(n\_max \times 4 ^ {level}) - 1$; and the $data$ representing the block start address, which is equal to $buf + (max\_sz / 4 ^ {level}) \times block$. We use a tuple $(level, block)$ to uniquely represent a block within a pool $p$.

A memory pool keeps track of how its buffer space has been split using a linked list \emph{free\_list} with the start address of the free blocks in each level. To improve the performance of coalescing partner blocks, memory pools maintain a bitmap at each level to indicate the allocation status of each block in the level.  This structure is represented by a C union of an integer \emph{bits} and an array \emph{bits\_p}. The implementation can allocate the bitmaps at levels smaller than $max\_inlinle\_levels$ using only an integer \emph{bits}. However, the number of blocks in levels higher than $max\_inlinle\_levels$ make necessary to allocate the bitmap information using the array \emph{bits\_map}. 
In such a design, the levels of bitmaps actually form a forest of complete quadtrees. 
The bit $i$ in the bitmap of level $j$ is set to $1$ for the block $(i,j)$ iff it is a free block, i.e. it is in the free list at level $i$. Otherwise, the bitmap for such block is set to $0$.

\lstdefinestyle{customc}{
  belowcaptionskip=1\baselineskip,
  breaklines=true,
  frame=single, 
  xleftmargin=0pt, 
  language=C,
  numbers=left,
  stepnumber=1,
  showstringspaces=false,
  basicstyle=\scriptsize\ttfamily, 
  keywordstyle=\bfseries\color{green!40!black},
  commentstyle=\itshape\color{purple!40!black},
  identifierstyle=\color{blue},
  stringstyle=\color{orange},
}
\lstset{escapechar=@,style=customc}

\begin{figure}[t]
\begin{lstlisting}
static int pool_alloc(struct k_mem_pool *p,struct k_mem_block *block,size_t size)
{
  ..... //calcuate lsizes[], alloc_l and free_l
  if (alloc_l < 0 || free_l < 0) {
    block->data = NULL;
    return -ENOMEM;
  }
  blk = alloc_block(p, free_l, lsizes[free_l]);
  if (!blk) { return -EAGAIN; }
  /* Iteratively break the smallest enclosing block... */
  for (from_l = free_l; level_empty(p, alloc_l) && from_l < alloc_l; 
  	   from_l++) {
    blk = break_block(p, blk, from_l, lsizes);
  }
  block->id.level = alloc_l; //assign block level to the variable *block
  ......  //assign other block info to the variable *block
  return 0;
}

int k_mem_pool_alloc(struct k_mem_pool *p, struct k_mem_block *block, size_t size, s32_t timeout)
{
  ...... // initialize local vars, calculate the end time for timeout. 
  while (1) {
    ret = pool_alloc(p, block, size);
    if (ret == 0 || timeout == K_NO_WAIT || 
        ret == -EAGAIN || (ret && ret != -ENOMEM)) {
      return ret;
    }
    key = irq_lock();
    _pend_current_thread(&p->wait_q, timeout);
    _Swap(key);
    ...... //if timeout > 0, break the loop if time out
  }
  return -EAGAIN;
}
\end{lstlisting}
\caption{The C Source Code of Memory Allocation in Zephyr v1.8.0}
\label{fig:mem_alloc_code}
\zipfigaft
\end{figure}

Zephyr provides two kernel services \emph{k\_mem\_pool\_alloc} and \emph{k\_mem\_pool\_free}, for memory allocation and release respectively. 
The main part of the C code of \emph{k\_mem\_pool\_alloc} is shown in {\figprefix} \ref{fig:mem_alloc_code}. 
When an application requests for a memory block, Zephyr first computes $alloc\_l$ and $free\_l$. $alloc\_l$ is the level with the size of the smallest block that will satisfy the request, and $free\_l$, with $free\_l \leqslant alloc\_l$, is the lowest level where there are free memory blocks. Since the services are concurrent, when the service tries to allocate a free block \emph{blk} from level $free\_l$ (Line 8), blocks at that level may be allocated or merged into a bigger block by other concurrent threads. In such case, the service will back out (Line 9) and tell the main function \emph{k\_mem\_pool\_alloc} to retry. If $blk$ is successfully locked for allocation, then it is broken down to level $alloc\_l$ (Lines 11 - 14). 
The allocation service \emph{k\_mem\_pool\_alloc} supports a \emph{timeout} parameter to allow threads waiting for that pool for a period of time when the call does not succeed. If the allocation fails (Line 24) and the timeout is not \emph{K\_NO\_WAIT}, the thread is suspended (Line 30) in a linked list \emph{wait\_q}, and the context is switched to another thread (Line 31).

Interruptions are always enabled in both services with the exception of the code for the functions \emph{alloc\_block} and \emph{break\_block}, which invoke \emph{irq\_lock} and \emph{irq\_unlock}  to respectively enable and disable interruptions. 
Similar to \emph{k\_mem\_pool\_alloc}, the execution of \emph{k\_mem\_pool\_free} is interruptable too. 

\subsection{Defining Structures and Properties of Buddy Memory Pools}

As a specification at the design level, we use abstract data types to represent the complete structure of memory pools. 
We use an abstract reference \emph{ref} in Isabelle to define pointers to memory pools. Starting addresses of memory blocks, memory pools, and unsigned integers in the implementation are defined as \emph{natural} numbers (\emph{nat}). Linked lists used in the implementation for the elements \emph{levels} and \emph{free\_list}, together with the bitmaps used in \emph{bits} and \emph{bits\_p}, are defined as a \emph{list} type. 
C \emph{structs} are modelled in Isabelle as \emph{records} of the same name as the implementation and comprising the same data. There are two exceptions to this: (1) $k\_mem\_block\_id$ and $k\_mem\_block$ are merged in one single record, (2) the union in the struct $k\_mem\_pool\_lvl$ is replaced by a single list representing the bitmap, and thus \emph{max\_inline\_level} is removed. 

\begin{figure}[t]
\begin{center}
\includegraphics[width=4.8in]{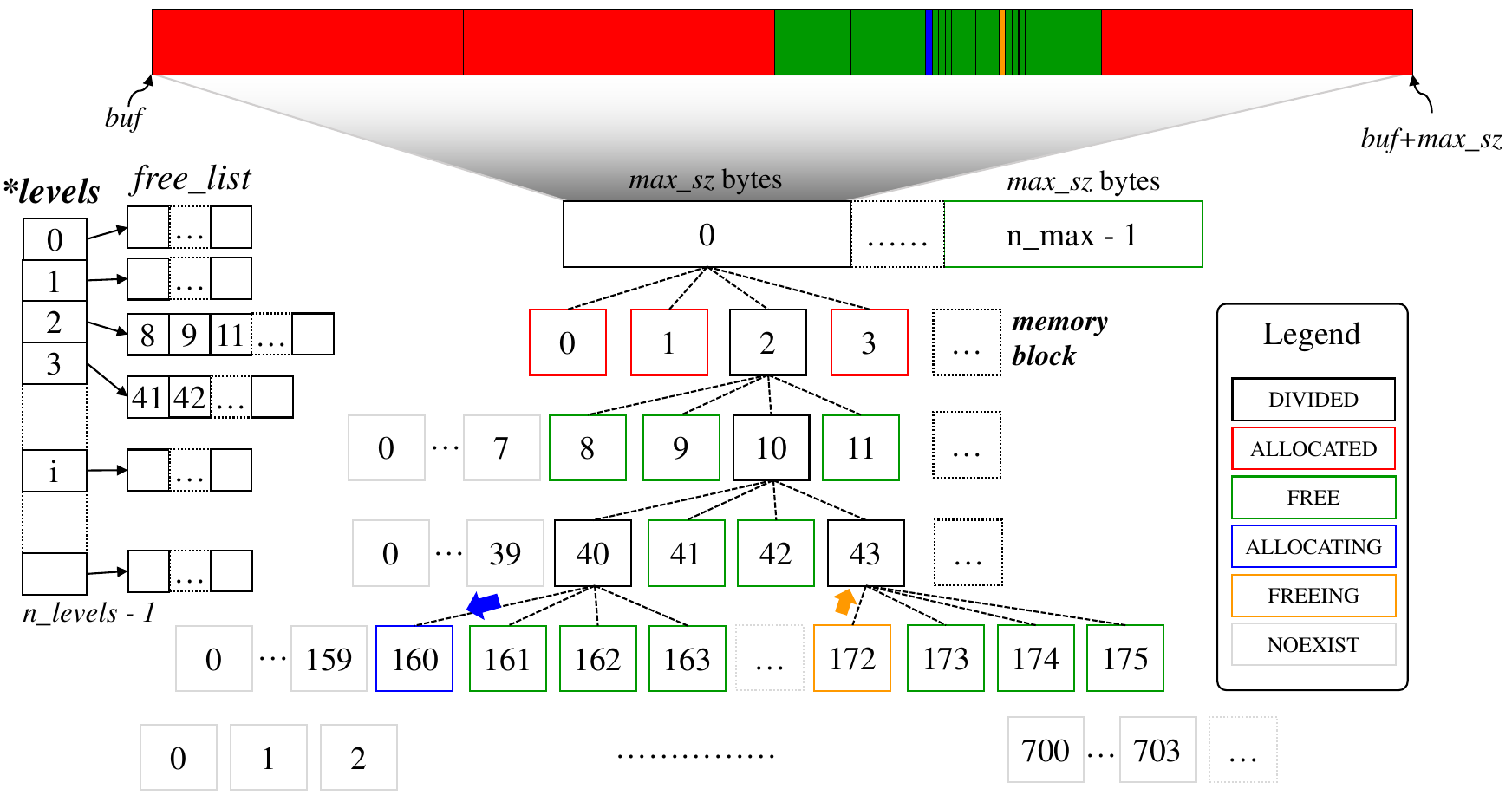}
\end{center}
\caption{Structure of Memory Pools}
\label{fig:mempool}
\zipfigaft
\end{figure}

The Zephyr implementation makes use of a bitmap to represent the state of a memory block. The bit $j$ of the bitmap for level a $i$ is set to $1$ iff the memory address of the memory block $(i,j)$ is in the free list at level $i$. A bit $j$ at a level $i$ is set to $0$ under the following conditions: (1) its corresponding memory block is allocated (\emph{ALLOCATED}), (2) the memory block has been split (\emph{DIVIDED}), (3) the memory block is being split in the allocation service (\emph{ALLOCATING}) (Line 13 in {\figprefix} \ref{fig:mem_alloc_code}), (4) the memory block is being coalesced in the release service (\emph{FREEING}), and (5) the memory block does not exist (\emph{NOEXIST}). Instead of only using a binary representation, our formal specification models the bitmap using a datatype \emph{BlockState} that is composed of these cases together with \emph{FREE}. The reason for this decision is to simplify proving that the bitmap shape is well-formed. In particular, this representation makes it less complex to verify the case in which the descendant of a free block is a non-free block. This is the case where the last free block has not been split, and therefore lower levels do not exist. 
We illustrate the structure of a memory pool in {\figprefix} \ref{fig:mempool}. The top of the figure shows the real memory of the first block at level $0$.

The structural properties clarify the constraints on and consistency of quadtrees, free block lists, the memory pool configuration, and waiting threads. All of them are thought of as invariants on the kernel state and have been formally verified on the formal specification in Isabelle/HOL. 

\paragraph{\textbf{Well-shaped bitmaps.}}
We say that the logical memory block $j$ at a level $i$ physically exists iff the bitmap $j$ for the level $i$ is \emph{ALLOCATED}, \emph{FREE}, \emph{ALLOCATING}, or \emph{FREEING}, represented by the predicate $is\_memblock$. We do not consider blocks marked as \emph{DIVIDED} as physical blocks since it is only a logical block containing other blocks. 
Threads may split and coalesce memory blocks. A valid forest is defined by the following rules: (1) the parent bit of an existing memory block is \emph{DIVIDED} and its child bits are \emph{NOEXIST}, denoted by the predicate $noexist\_bits$ that checks for a given bitmap $b$ and a position $j$ that nodes $b!j$ to $b!(j+3)$ are set as \emph{NOEXIST}; (2) the parent bit of a \emph{DIVIDED} block is also \emph{DIVIDED}; and (3) the child bits of a \emph{NOEXIST} bit are also \emph{NOEXIST} and its parent can not be a \emph{DIVIDED} block. The property is defined as the predicate \isacommand{inv{\isacharunderscore}bitmap}($s$), where $s$ is the state.

There are two additional properties on bitmaps. First, the address space of any memory pool can not be empty, i.e., the bits at level 0 have to be different from \emph{NOEXIST}. Second, the allocation algorithm may split a memory block into smaller ones, but not those blocks at the lowest level (i.e. level $n\_levels - 1$); therefore the bits at the lowest level cannot not be \emph{DIVIDED}. The first property is defined as \isacommand{inv{\isacharunderscore}bitmap{\isadigit{0}}}($s$) and the second as \isacommand{inv{\isacharunderscore}bitmapn}($s$).

\paragraph{\textbf{Consistency of the memory configuration.}}
The configuration of a memory pool is set when it is initialized. Since the minimum block size is aligned to 4 bytes, there must exist an $n > 0$ such that the maximum size of a pool is equal to $4 \times n \times 4 ^ {n\_levels}$, relating the number of levels of a level 0 block with its maximum size. Moreover, the number of blocks at level 0 and the number of levels have to be greater than zero, since the memory pool cannot be empty. The number of levels is equal to the length of the  pool $levels$ list. Finally, the length of the bitmap at level $i$ should be $n\_max \times 4 ^ i$. This property is defined as \isacommand{inv{\isacharunderscore}mempool{\isacharunderscore}info}($s$).

\paragraph{\textbf{Memory partition property.}}
Memory blocks partition the pool they belong to, in this context, there are two critical properties that the allocation and release of blocks must preserve: the absence of any overlapping blocks and the absence of memory leaks. For a memory block of index $j$ at level $i$, its address space is the interval $[j \times (max\_sz / 4 ^ i), (j + 1) \times (max\_sz / 4 ^ i))$. For any relative memory address $addr$ in the memory domain of a memory pool, and hence $addr < n\_max * max\_sz$, there is one and only one memory block whose address space contains $addr$. Here, we use relative address for $addr$.  The property is defined as \isacommand{mem{\isacharunderscore}part}(s).

From the invariants of the bitmap, we derive the general property for the memory partition. 

\begin{theorem}[Memory Partition]
For any kernel state $s$, If the memory pools in $s$ are consistent in their configuration, and their bitmaps are well-shaped, the memory pools satisfy the partition property in $s$: 
\[
\textbf{inv\_mempool\_info}(s) \wedge \textbf{inv\_bitmap}(s) \wedge \textbf{inv\_bitmap0}(s) \wedge \textbf{inv\_bitmapn}(s) \Longrightarrow \textbf{mem\_part}(s)
\]
\end{theorem}

Together with the memory partition property, pools must also satisfy the following:

\paragraph{\textbf{No partner fragmentation.}}
The memory release algorithm in Zephyr coalesces free partner memory blocks into blocks as large as possible for all the descendants from the root level, without including it. Thus, a memory pool does not contain four \emph{FREE} partner bits.

\paragraph{\textbf{Validity of free block lists.}}
The free list at one level keeps the starting address of free memory blocks. The memory management ensures that the addresses in the list are valid, i.e., they are different from each other and aligned to the \emph{block size}, which at a level $i$ is given by ($max\_sz / 4 ^ i$). 
Moreover, a memory block is in the free list iff the corresponding bit of the bitmap is \emph{FREE}. 

\paragraph{\textbf{Non-overlapping of memory pools.}}
The memory spaces of the set of pools defined in a system must be disjoint, so the memory addresses of a pool do not belong to the memory space of any other pool. 

\paragraph{\textbf{Other properties.}} 
The state of a suspended thread in \emph{wait\_q} has to be consistent with the threads waiting for a memory pool. Threads can only be blocked once, and those threads waiting for available memory blocks have to be in a \emph{BLOCKED} state.
During allocation and free of a memory block, blocks of the tree may temporally be manipulated during the coalescing and division process. Only one thread can manipulate a give bloack at a time, and the state bit of a block being temporally manipulated has to be \emph{FREEING} or \emph{ALLOCATING}. 

\subsection{Formalizing Zephyr Memory Management}

\subsubsection{Event-based Execution Model of Zephyr}
\begin{figure}[t]
\begin{center}
\includegraphics[width=5in]{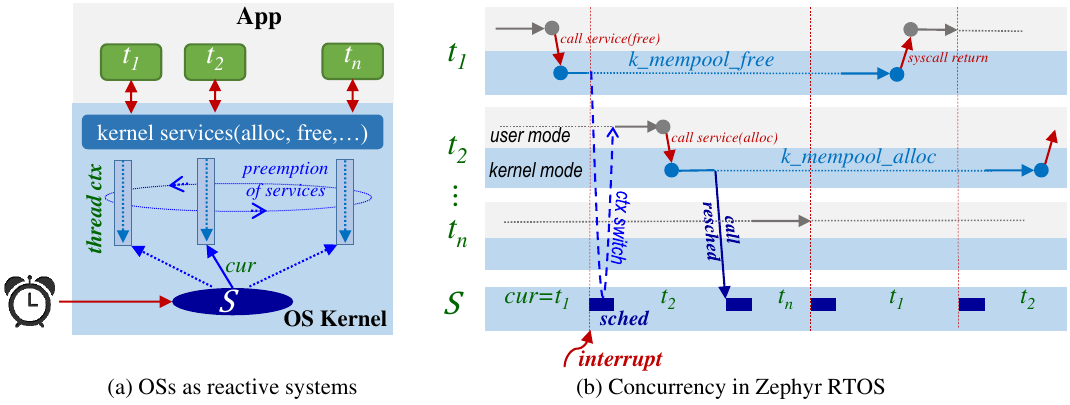}
\end{center}
\caption{An Execution Model of Zephyr Memory Management}
\label{fig:kernel_model}
\zipfigaft
\end{figure}

As shown in {\figprefix} \ref{fig:kernel_model} (a), if we do not consider its initialization, an OS kernel can be seen as a reactive system that is in an \emph{idle} loop until it receives an interruption (e.g. syscalls, timer interrupts) which is handled by an interruption handler. 
Whilst interrupt handlers execution is atomic in sequential kernels, it can be interrupted in concurrent kernels~\cite{Chen16,Xu16} allowing services invoked by threads to be interrupted and resumed later. 
In the execution model of Zephyr, we consider a scheduler $\mathcal{S}$ and a set of threads $t_1, ..., t_n$. In this model, the execution of the scheduler is atomic since kernel services can not interrupt it. But kernel services can be interrupted via the scheduler, i.e., the execution of a memory service invoked by a thread $t_i$ may be interrupted by the kernel scheduler to execute a thread $t_j$. {\figprefix}\ref{fig:kernel_model} (b) illustrates the Zephyr execution model, where solid lines represent execution steps of the threads/kernel services and dotted lines mean the suspension of the thread/code. For instance, the execution of \emph{k\_mempool\_free} in thread $t_1$ is interrupted by the scheduler, and the context is switched to thread $t_2$ which invokes \emph{k\_mempool\_alloc}. During the execution of $t_2$, the kernel service may suspend the thread and switch to another thread $t_n$ by calling \emph{rescheduling}. Later, the execution is switched back to $t_1$ and continues the execution of \emph{k\_mempool\_free} in a different state from when it was interrupted.

Each user thread $t_i$ invokes the allocation and release services; thus the event system for $t_i$ is 
\begin{equation*}
\begin{aligned}
esys_{t_i} \equiv & \evtrec{\mathbf{True}}{((\stmtevthead{mem\_pool\_free}{blk}{t_i})} \\
 & \quad \quad \quad\ \ \symbchoice (\stmtevthead{mem\_pool\_alloc}{p,sz,tmout}{t_i}))
\end{aligned}
\end{equation*}

which is the iteration of \emph{alloc} or \emph{free} events, where the input parameters for these events correspond with the arguments of the service implementation in the C code. 
Events are parametrized by a thread identifier $t_i$ used to control access to the execution context of the thread invoking it. 
Together with threads, we model the event service for the scheduler $esys_{sched}$ consisting of a unique event \emph{sched} whose argument is a thread $t$ to be scheduled when it is in the \emph{READY} state. 
The formal specification of the memory management is thus defined as follows, which is much simpler than the specification obtained from a non-event oriented language. 
\[
Kernel\_Spec \equiv \lambda \symbcore. \ \textbf{case} \ \symbcore \ \textbf{of} \ (\isasymT \ t_i)\ \isasymRightarrow\ esys_{t_i} \ \isacharbar \ \textbf{sched} \ \isasymRightarrow\ esys_{sched}
\]

The execution of a service executed by a thread can be stopped by the scheduler to be resumed later. This behaviour is modelled by using a global variable $cur$ that indicates the thread being currently has been scheduled and is being executed, and conditioning the execution of parametrized events in $t$ only when $t$ is scheduled. This is achieved by using the expression 
\[ \stmtirq{t}{p} \equiv \stmtawait{cur = t}{p}\]
so an event invoked by a thread $t$ only progresses when $t$ is scheduled. This scheme allows using rely-guarantee for concurrent execution of threads on mono-core architectures, where only the scheduled thread is able to modify the memory.

\subsubsection{Specification of Memory Management Services}
The C code of Zephyr uses the recursive function \emph{free\_block} to coalesce free partner blocks and the \emph{break} statement to stop the execution of a loop statements, which are not supported by the {\implang} language in {\pcimp}. The formal specification overcomes this by transforming the recursion into a loop controlled by the recursion condition, and using a control variable to exit loops with breaks when the condition to execute the loop break is satisfied. 
The memory management services use the atomic body \emph{irq\_lock(); P; irq\_unlock();} to keep interruption handlers \emph{reentrant} by disabling interruptions. 
We simplify this behaviour in the specification using an \textbf{ATOM} statement, avoiding that the service is interrupted at that point. The rest of the formal specification closely follows the implementation, where variables are modified using higher order functions changing the state as the code does it. The reason for using Isabelle/HOL functions is that {\slang} does not provide a semantic for expressions, using instead state transformer relying on high order functions to change the state.

\begin{figure}[t]
\begin{flushleft}
\begin{isabellec} \fontsize{7pt}{0cm} 
\ \ 1 \ \isacommand{WHILE}\ {\isasymacute}free{\isacharunderscore}block{\isacharunderscore}r\ t\ \isacommand{DO}\isanewline
\ \ 2 \ \quad t\ {\isactrlenum} \ {\isasymacute}lsz\ {\isacharcolon}{\isacharequal}\ {\isasymacute}lsz\ {\isacharparenleft}t\ {\isacharcolon}{\isacharequal}\ {\isasymacute}lsizes\ t\ {\isacharbang}\ {\isacharparenleft}{\isasymacute}lvl\ t{\isacharparenright}{\isacharparenright}{\isacharsemicolon}{\isacharsemicolon}\isanewline
\ \ 3 \ \quad t\ {\isactrlenum} \ {\isasymacute}blk\ {\isacharcolon}{\isacharequal}\ {\isasymacute}blk\ {\isacharparenleft}t\ {\isacharcolon}{\isacharequal}\ block{\isacharunderscore}ptr\ {\isacharparenleft}{\isasymacute}mem{\isacharunderscore}pool{\isacharunderscore}info\ {\isacharparenleft}pool\ b{\isacharparenright}{\isacharparenright}\ {\isacharparenleft}{\isasymacute}lsz\ t{\isacharparenright}\ {\isacharparenleft}{\isasymacute}bn\ t{\isacharparenright}{\isacharparenright}{\isacharsemicolon}{\isacharsemicolon}\isanewline
\ \ 4 \ \quad t\ {\isactrlenum} \ \isacommand{ATOM}\isanewline
\ \ 5 \ \quad \quad {\isasymacute}mem{\isacharunderscore}pool{\isacharunderscore}info\ {\isacharcolon}{\isacharequal}\ set{\isacharunderscore}bit{\isacharunderscore}free\ {\isasymacute}mem{\isacharunderscore}pool{\isacharunderscore}info\ {\isacharparenleft}pool\ b{\isacharparenright}\ {\isacharparenleft}{\isasymacute}lvl\ t{\isacharparenright}\ {\isacharparenleft}{\isasymacute}bn\ t{\isacharparenright}{\isacharsemicolon}{\isacharsemicolon}\isanewline
\ \ 6 \ \quad \quad {\isasymacute}freeing{\isacharunderscore}node\ {\isacharcolon}{\isacharequal}\ {\isasymacute}freeing{\isacharunderscore}node\ {\isacharparenleft}t\ {\isacharcolon}{\isacharequal}\ None{\isacharparenright}{\isacharsemicolon}{\isacharsemicolon}\isanewline
\ \ 7 \ \quad \quad \isacommand{IF}\ {\isasymacute}lvl\ t\ {\isachargreater}\ {\isadigit{0}}\ {\isasymand}\ partner{\isacharunderscore}bits\ {\isacharparenleft}{\isasymacute}mem{\isacharunderscore}pool{\isacharunderscore}info\ {\isacharparenleft}pool\ b{\isacharparenright}{\isacharparenright}\ {\isacharparenleft}{\isasymacute}lvl\ t{\isacharparenright}\ {\isacharparenleft}{\isasymacute}bn\ t{\isacharparenright}\ \isacommand{THEN}\isanewline
\ \ 8 \ \quad \quad \quad \isacommand{FOR}\ {\isasymacute}i\ {\isacharcolon}{\isacharequal}\ {\isasymacute}i{\isacharparenleft}t\ {\isacharcolon}{\isacharequal}\ {\isadigit{0}}{\isacharparenright}{\isacharsemicolon}\ {\isasymacute}i\ t\ {\isacharless}\ {\isadigit{4}}{\isacharsemicolon}\ {\isasymacute}i\ {\isacharcolon}{\isacharequal}\ {\isasymacute}i{\isacharparenleft}t\ {\isacharcolon}{\isacharequal}\ {\isasymacute}i\ t\ {\isacharplus}\ {\isadigit{1}}{\isacharparenright}\ \isacommand{DO}\isanewline
\ \ 9 \ \quad \quad \quad \quad {\isasymacute}bb\ {\isacharcolon}{\isacharequal}\ {\isasymacute}bb\ {\isacharparenleft}t\ {\isacharcolon}{\isacharequal}\ {\isacharparenleft}{\isasymacute}bn\ t\ div\ {\isadigit{4}}{\isacharparenright}\ {\isacharasterisk}\ {\isadigit{4}}\ {\isacharplus}\ {\isasymacute}i\ t{\isacharparenright}{\isacharsemicolon}{\isacharsemicolon}\isanewline
10 \ \quad \quad \quad \quad {\isasymacute}mem{\isacharunderscore}pool{\isacharunderscore}info\ {\isacharcolon}{\isacharequal}\ set{\isacharunderscore}bit{\isacharunderscore}noexist\ {\isasymacute}mem{\isacharunderscore}pool{\isacharunderscore}info\ {\isacharparenleft}pool\ b{\isacharparenright}\ {\isacharparenleft}{\isasymacute}lvl\ t{\isacharparenright}\ {\isacharparenleft}{\isasymacute}bb\ t{\isacharparenright}{\isacharsemicolon}{\isacharsemicolon}\isanewline
11 \ \quad \quad \quad \quad {\isasymacute}block{\isacharunderscore}pt\ {\isacharcolon}{\isacharequal}\ {\isasymacute}block{\isacharunderscore}pt\ {\isacharparenleft}t\ {\isacharcolon}{\isacharequal}\ block{\isacharunderscore}ptr\ {\isacharparenleft}{\isasymacute}mem{\isacharunderscore}pool{\isacharunderscore}info\ {\isacharparenleft}pool\ b{\isacharparenright}{\isacharparenright}\ {\isacharparenleft}{\isasymacute}lsz\ t{\isacharparenright}\ {\isacharparenleft}{\isasymacute}bb\ t{\isacharparenright}{\isacharparenright}{\isacharsemicolon}{\isacharsemicolon}\isanewline
12 \ \quad \quad \quad \quad \isacommand{IF}\ {\isasymacute}bn\ t\ {\isasymnoteq}\ {\isasymacute}bb\ t\ {\isasymand}\ block{\isacharunderscore}fits\ {\isacharparenleft}{\isasymacute}mem{\isacharunderscore}pool{\isacharunderscore}info\ {\isacharparenleft}pool\ b{\isacharparenright}{\isacharparenright}\ {\isacharparenleft}{\isasymacute}block{\isacharunderscore}pt\ t{\isacharparenright}\ {\isacharparenleft}{\isasymacute}lsz\ t{\isacharparenright}\ \isacommand{THEN}\isanewline
13 \ \quad \quad \quad \quad \quad {\isasymacute}mem{\isacharunderscore}pool{\isacharunderscore}info\ {\isacharcolon}{\isacharequal}\ {\isasymacute}mem{\isacharunderscore}pool{\isacharunderscore}info\ {\isacharparenleft}{\isacharparenleft}pool\ b{\isacharparenright}\ {\isacharcolon}{\isacharequal}\ \isanewline
14 \ \quad \quad \quad \quad \quad \quad \quad remove{\isacharunderscore}free{\isacharunderscore}list\ {\isacharparenleft}{\isasymacute}mem{\isacharunderscore}pool{\isacharunderscore}info\ {\isacharparenleft}pool\ b{\isacharparenright}{\isacharparenright}\ {\isacharparenleft}{\isasymacute}lvl\ t{\isacharparenright}\ {\isacharparenleft}{\isasymacute}block{\isacharunderscore}pt\ t{\isacharparenright}{\isacharparenright}\isanewline
15 \ \quad \quad \quad \quad \isacommand{FI}\isanewline
16 \ \quad \quad \quad \isacommand{ROF}{\isacharsemicolon}{\isacharsemicolon}\isanewline
17 \ \quad \quad \quad {\isasymacute}lvl\ {\isacharcolon}{\isacharequal}\ {\isasymacute}lvl\ {\isacharparenleft}t\ {\isacharcolon}{\isacharequal}\ {\isasymacute}lvl\ t\ {\isacharminus}\ {\isadigit{1}}{\isacharparenright}{\isacharsemicolon}{\isacharsemicolon}\isanewline
18 \ \quad \quad \quad {\isasymacute}bn\ {\isacharcolon}{\isacharequal}\ {\isasymacute}bn\ {\isacharparenleft}t\ {\isacharcolon}{\isacharequal}\ {\isasymacute}bn\ t\ div\ {\isadigit{4}}{\isacharparenright}{\isacharsemicolon}{\isacharsemicolon}\isanewline
19 \ \quad \quad \quad {\isasymacute}mem{\isacharunderscore}pool{\isacharunderscore}info\ {\isacharcolon}{\isacharequal}\ set{\isacharunderscore}bit{\isacharunderscore}freeing\ {\isasymacute}mem{\isacharunderscore}pool{\isacharunderscore}info\ {\isacharparenleft}pool\ b{\isacharparenright}\ {\isacharparenleft}{\isasymacute}lvl\ t{\isacharparenright}\ {\isacharparenleft}{\isasymacute}bn\ t{\isacharparenright}{\isacharsemicolon}{\isacharsemicolon}\isanewline
20 \ \quad \quad \quad {\isasymacute}freeing{\isacharunderscore}node\ {\isacharcolon}{\isacharequal}\ {\isasymacute}freeing{\isacharunderscore}node\ {\isacharparenleft}t\ {\isacharcolon}{\isacharequal}\ Some\ {\isasymlparr}pool\ {\isacharequal}\ {\isacharparenleft}pool\ b{\isacharparenright}{\isacharcomma}\ level\ {\isacharequal}\ {\isacharparenleft}{\isasymacute}lvl\ t{\isacharparenright}{\isacharcomma}\ \isanewline
21 \ \ \ \ \ \ \ \ \ \ \ \ \ \ \ \ \ \ \ \ \ block\ {\isacharequal}\ {\isacharparenleft}{\isasymacute}bn\ t{\isacharparenright}{\isacharcomma}\ 
data\ {\isacharequal}\ block{\isacharunderscore}ptr\ {\isacharparenleft}{\isasymacute}mem{\isacharunderscore}pool{\isacharunderscore}info\ {\isacharparenleft}pool\ b{\isacharparenright}{\isacharparenright}\ \isanewline
22 \ \ \ \ \ \ \ \ \ \ \ \ \ \ \ \ \ \ \ \ \ \ \  {\isacharparenleft}{\isacharparenleft}{\isacharparenleft}ALIGN{\isadigit{4}}\ {\isacharparenleft}max{\isacharunderscore}sz\ {\isacharparenleft}{\isasymacute}mem{\isacharunderscore}pool{\isacharunderscore}info\ {\isacharparenleft}pool\ b{\isacharparenright}{\isacharparenright}{\isacharparenright}{\isacharparenright}\ div\ {\isacharparenleft}{\isadigit{4}}\ {\isacharcircum}\ {\isacharparenleft}{\isasymacute}lvl\ t{\isacharparenright}{\isacharparenright}{\isacharparenright}{\isacharparenright}\ 
{\isacharparenleft}{\isasymacute}bn\ t{\isacharparenright}\ {\isasymrparr}{\isacharparenright}\isanewline
23 \ \quad \quad \isacommand{ELSE}\isanewline
24 \ \quad \quad \quad \isacommand{IF}\ block{\isacharunderscore}fits\ {\isacharparenleft}{\isasymacute}mem{\isacharunderscore}pool{\isacharunderscore}info\ {\isacharparenleft}pool\ b{\isacharparenright}{\isacharparenright}\ {\isacharparenleft}{\isasymacute}blk\ t{\isacharparenright}\ {\isacharparenleft}{\isasymacute}lsz\ t{\isacharparenright}\ \isacommand{THEN}\isanewline
25 \ \quad \quad \quad \quad {\isasymacute}mem{\isacharunderscore}pool{\isacharunderscore}info\ {\isacharcolon}{\isacharequal}\ {\isasymacute}mem{\isacharunderscore}pool{\isacharunderscore}info\ {\isacharparenleft}{\isacharparenleft}pool\ b{\isacharparenright}\ {\isacharcolon}{\isacharequal}\ \isanewline
26 \ \quad \quad \quad \quad \quad append{\isacharunderscore}free{\isacharunderscore}list\ {\isacharparenleft}{\isasymacute}mem{\isacharunderscore}pool{\isacharunderscore}info\ {\isacharparenleft}pool\ b{\isacharparenright}{\isacharparenright}\ {\isacharparenleft}{\isasymacute}lvl\ t{\isacharparenright}\ {\isacharparenleft}{\isasymacute}blk\ t{\isacharparenright}\ {\isacharparenright}\isanewline
27 \ \quad \quad \quad \isacommand{FI}{\isacharsemicolon}{\isacharsemicolon}\isanewline
28 \ \quad \quad \quad {\isasymacute}free{\isacharunderscore}block{\isacharunderscore}r\ {\isacharcolon}{\isacharequal}\ {\isasymacute}free{\isacharunderscore}block{\isacharunderscore}r\ {\isacharparenleft}t\ {\isacharcolon}{\isacharequal}\ False{\isacharparenright}\isanewline
29 \ \quad \quad \isacommand{FI}\isanewline
30 \ \quad \isacommand{END}\isanewline
31 \ \isacommand{OD}
\end{isabellec} 
\end{flushleft}
\caption{The {\slang} Specification of \emph{free\_block}}
\label{fig:freeblock}
\end{figure}

{\figprefix} \ref{fig:freeblock} illustrates the {\slang} specification of the \emph{free\_block} function invoked by \emph{k\_mem\_pool\_free} when releasing a memory block. 
We refer readers to {\appendixprefix} \ref{appx:mempoolfree_c} for the main part of the C code of the \emph{k\_mem\_pool\_free} service and {\appendixprefix} \ref{appx:mempoolfree} for the complete specification of the service.

The code accesses  the following variables: $lsz$, $lsize$, and $lvl$ to keep information about the current level; $blk$, $bn$, and $bb$ to represent the address and number of the block currently being accessed; $freeing\_node$ to represent the node being freeing; and $i$ to iterate blocks. Additionally, the model includes the component $free\_block\_r$ to model the recursion condition. To simplify the representation, the model uses predicates and functions to access and modify the state. Due to space constraints, we are unable to provide a detailed explanation of these functions. However, the name of the functions can help the reader to better understand their functionality. 

In the C code,  \emph{free\_block} is a recursive function with two conditions: (1) the block being released belongs to a level higher than zero, since blocks at level zero cannot be merged; and (2) the partners bits of the block being released are FREE, so they can be merged into a bigger block. We represent (1) with the predicate ${\isasymacute}lvl\ t\ {\isachargreater}\ {\isadigit{0}}$ and (2) with the predicate $partner\_bit\_free$. The formal specification follows the same structure translating the recursive function into a loop that is controlled by a variable mimicking the recursion. 

The formal specification for  \emph{free\_block} first releases an allocated memory block $bn$ setting it to \emph{FREEING}. Then, the loop statement sets \emph{free\_block} to \emph{FREE} (Line 5), and also checks that the iteration/recursive condition holds in Line 7. If the condition holds, the partner bits are set  to \emph{NOEXIST}, and remove their addresses from the free list for this level (Lines 12 - 14). Then, it sets the parent block bit to \emph{FREEING} (Lines 17 - 22), and updates the variables controlling the current block and level numbers, before going back to the beginning of the loop again. If the iteration condition is not true, it sets the bit to \emph{FREE} and add the block to the free list (Lines 24 - 28) and sets the loop condition to false to end the procedure. 
This function is illustrated in {\figprefix} \ref{fig:mempool}. A thread releases the block $172$, and since its partner blocks (block $173 - 175$) are free, Zephyr coalesces the four blocks and sets their parent block $43$ as \emph{FREEING}. The coalescence continues iteratively if the partners of block $43$ are all free. 

\subsection{Correctness and Rely-guarantee Proof}

Using the compositional reasoning of {\slang}, correctness of Zephyr memory management can be specified and verified with the rely-guarantee specification of each event.
A pair pre-post-condition for a kernel service specifies its functional correctnes.
Invariant preservation, memory configuration, and separation of local variables are specified in the guarantee condition of each service. 
The guarantee condition for both memory services is defined as:

\begin{isabellec}
\isacommand{Mem{\isacharunderscore}pool{\isacharunderscore}alloc{\isacharunderscore}guar}\ t\ {\isasymequiv} 
$\overbrace{Id}^{(1)}$
 \ {\isasymunion}\ {\isacharparenleft}
$\overbrace{gvars\_conf\_stable}^{(2)}$
 {\isasyminter}\ \isanewline
\quad {\isacharbraceleft}{\isacharparenleft}s{\isacharcomma}r{\isacharparenright}{\isachardot}\ {\isacharparenleft}
$\overbrace{cur\ s\ {\isasymnoteq}\ Some\ t\ {\isasymlongrightarrow}\ gvars{\isacharunderscore}nochange\ s\ r\ {\isasymand}\ lvars{\isacharunderscore}nochange\ t\ s\ r}^{(3.1)}$
{\isacharparenright}\isanewline
\quad {\isasymand}\ {\isacharparenleft}
$\overbrace{cur\ s\ {\isacharequal}\ Some\ t\ {\isasymlongrightarrow}\ inv\ s\ {\isasymlongrightarrow}\ inv\ r}^{(3.2)}$
{\isacharparenright}\ {\isasymand}\ {\isacharparenleft}
$\overbrace{{\isasymforall}t{\isacharprime}{\isachardot}\ t{\isacharprime}\ {\isasymnoteq}\ t\ {\isasymlongrightarrow}\ lvars{\isacharunderscore}nochange\ t{\isacharprime}\ s\ r}^{(4)}$
{\isacharparenright}\ {\isacharbraceright}{\isacharparenright}
\end{isabellec}

This relation states that the \emph{alloc} and \emph{free} services may not change the state (1), e.g., a blocked await or selecting branch on a conditional statement. If it changes the state then: (2) the static configuration of memory pools in the model do not change; (3.1) if the scheduled thread is not the thread invoking the event then variables for that thread do not change, since it is blocked in an \emph{Await}; (3.2) if it is, then the relation preserves the memory invariant, and consequently each step of the event needs to preserve the invariant $inv$; (4) a thread does not change the local variables of other threads.

Using the {\slang} proof rules, we verify that all the events reserve the invariant. Additionally, we prove that when starting in a valid memory configuration given by the invariant, then the service returns a valid memory block with size bigger or equal than the requested capacity, when the returned value is not an error code. The following postcondition specifies the property: 

\begin{isabellec}
\isacommand{Mem{\isacharunderscore}pool{\isacharunderscore}alloc{\isacharunderscore}pre}\ t\ {\isasymequiv}\ {\isacharbraceleft}s{\isachardot}\ inv\ s\ {\isasymand}\ allocating{\isacharunderscore}node\ s\ t\ {\isacharequal}\ None\ {\isasymand}\ freeing{\isacharunderscore}node\ s\ t\ {\isacharequal}\ None{\isacharbraceright}

\isacommand{Mem{\isacharunderscore}pool{\isacharunderscore}alloc{\isacharunderscore}post}\ t\ p\ sz\ timeout\ {\isasymequiv}\ {\isacharbraceleft}s{\isachardot}\ inv\ s\ {\isasymand}\ allocating{\isacharunderscore}node\ s\ t\ {\isacharequal}\ None\ {\isasymand}\ freeing{\isacharunderscore}node\ s\ t\ {\isacharequal}\ None\isanewline
\ \ \ \ \ \ {\isasymand}\ {\isacharparenleft}timeout\ {\isacharequal}\ FOREVER\ {\isasymlongrightarrow}\isanewline
\quad \quad \quad {\isacharparenleft}ret\ s\ t\ {\isacharequal}\ ESIZEERR\ {\isasymand}\ mempoolalloc{\isacharunderscore}ret\ s\ t\ {\isacharequal}\ None\ {\isasymor}\isanewline
\quad \quad \quad ret\ s\ t\ {\isacharequal}\ OK\ {\isasymand}\ {\isacharparenleft}{\isasymexists}mblk{\isachardot}\ mempoolalloc{\isacharunderscore}ret\ s\ t\ {\isacharequal}\ Some\ mblk\ {\isasymand}\ mblk{\isacharunderscore}valid\ s\ p\ sz\ mblk{\isacharparenright}{\isacharparenright}{\isacharparenright}\isanewline
\ \ \ \ \ \ {\isasymand}\ {\isacharparenleft}timeout\ {\isacharequal}\ NOWAIT\ {\isasymlongrightarrow}\isanewline
\quad \quad \quad {\isacharparenleft}{\isacharparenleft}ret\ s\ t\ {\isacharequal}\ ENOMEM\ {\isasymor}\ ret\ s\ t\ {\isacharequal}\ ESIZEERR{\isacharparenright}\ {\isasymand}\ mempoolalloc{\isacharunderscore}ret\ s\ t\ {\isacharequal}\ None{\isacharparenright} {\isasymor} \isanewline
\quad \quad \quad {\isacharparenleft}ret\ s\ t\ {\isacharequal}\ OK\ {\isasymand}\ {\isacharparenleft}{\isasymexists}mblk{\isachardot}\ mempoolalloc{\isacharunderscore}ret\ s\ t\ {\isacharequal}\ Some\ mblk\ {\isasymand}\ mblk{\isacharunderscore}valid\ s\ p\ sz\ mblk{\isacharparenright}{\isacharparenright}{\isacharparenright}\isanewline
\ \ \ \ \ \ {\isasymand}\ {\isacharparenleft}timeout\ {\isachargreater}\ {\isadigit{0}}\ {\isasymlongrightarrow}\isanewline
\quad \quad \quad {\isacharparenleft}{\isacharparenleft}ret\ s\ t\ {\isacharequal}\ ETIMEOUT\ {\isasymor}\ ret\ s\ t\ {\isacharequal}\ ESIZEERR{\isacharparenright}\ {\isasymand}\ mempoolalloc{\isacharunderscore}ret\ s\ t\ {\isacharequal}\ None{\isacharparenright} {\isasymor} \isanewline
\quad \quad \quad {\isacharparenleft}ret\ s\ t\ {\isacharequal}\ OK\ {\isasymand}\ {\isacharparenleft}{\isasymexists}mblk{\isachardot}\ mempoolalloc{\isacharunderscore}ret\ s\ t\ {\isacharequal}\ Some\ mblk \
{\isasymand}\ mblk{\isacharunderscore}valid\ s\ p\ sz\ mblk{\isacharparenright}{\isacharparenright}{\isacharparenright}{\isacharbraceright}
\end{isabellec}

If a thread requests a memory block in mode \emph{FOREVER}, it may successfully allocate a valid memory block, or fail (\emph{ESIZEERR}) if the requested size is larger than the size of the memory pool. If the thread is requesting a memory pool in mode \emph{NOWAIT}, it may also get the result of \emph{ENOMEM} if there are no available blocks. But if the thread is requesting the service in mode \emph{TIMEOUT}, it will get the result of \emph{ETIMEOUT} if there are no available blocks in \emph{timeout} milliseconds.

In {\slang}, verification of a rely-guarantee specification proving a property is carried out by inductively applying the proof rules for each system event and discharging the proof obligations the rules generate. Typically, these proof obligations require to prove stability of the pre and postcondition to check that changes of the environment preserve them, and to show that a statement modifying a state from the precondition gets a state belonging to the postcondition. 
A detailed proof sketch of the \emph{k\_mempool\_free} service is shown in {\appendixprefix} \ref{appx:mempoolfree}.

To prove loop termination, loop invariants are parametrized with a logical variable $\alpha$. It suffices to show total correctness of a loop statement by the following proposition where $loopinv(\alpha)$ is the parametrize invariant, in which the logical variable is used to find a convergent relation to show that the number of iterations of the loop is finite. 
\[ 
\begin{aligned} 
&\rgsatpI{(Some \ P)}{\rgcond{loopinv(\alpha) \cap \{s.\ \alpha > 0\}}{R}{G}{\exists \beta < \alpha. \ loopinv(\beta)}} \\
& \wedge\  loopinv(\alpha) \cap \{s.\ \alpha > 0\} \subseteq \{s.\ b\} \ \wedge \ loopinv(0) \subseteq \{s.\ \neg b\} \\
& \wedge \ \forall s \in loopinv(\alpha).\ (s,t) \in R \longrightarrow \exists \beta \leqslant \alpha. \ t \in loopinv(\beta)
\end{aligned}
\]

For instance, to prove termination of the loop statement in \emph{free\_block} shown in {\figprefix} \ref{fig:freeblock}, we define the loop invariant with the logical variable $\alpha$ as follows. 

\begin{isabellec}
\isacommand{mp{\isacharunderscore}free{\isacharunderscore}loopinv}\ t\ b\ {\isasymalpha}\ {\isasymequiv} {\isacharbraceleft}s.\ inv\ s {\isasymand} level\ b\ {\isacharless}\ length\ {\isacharparenleft}lsizes\ s\ t{\isacharparenright}\isanewline
\quad {\isasymand}\ {\isacharparenleft}{\isasymforall}ii{\isacharless}length\ {\isacharparenleft}lsizes\ s\ t{\isacharparenright}{\isachardot}\ lsizes\ s\ t\ {\isacharbang}\ ii\ {\isacharequal}\ {\isacharparenleft}max{\isacharunderscore}sz\ {\isacharparenleft}mem{\isacharunderscore}pool{\isacharunderscore}info\ s\ {\isacharparenleft}pool\ b{\isacharparenright}{\isacharparenright}{\isacharparenright}\ div\ {\isacharparenleft}{\isadigit{4}}\ {\isacharcircum}\ ii{\isacharparenright}{\isacharparenright}\isanewline
\quad {\isasymand}\ bn\ s\ t\ {\isacharless}\ length\ {\isacharparenleft}bits\ {\isacharparenleft}levels\ {\isacharparenleft}mem{\isacharunderscore}pool{\isacharunderscore}info\ s\ {\isacharparenleft}pool\ b{\isacharparenright}{\isacharparenright}{\isacharbang}{\isacharparenleft}lvl\ s\ t{\isacharparenright}{\isacharparenright}{\isacharparenright}\isanewline
\quad {\isasymand}\ bn\ s\ t\ {\isacharequal}\ {\isacharparenleft}block\ b{\isacharparenright}\ div\ {\isacharparenleft}{\isadigit{4}}\ {\isacharcircum}\ {\isacharparenleft}level\ b\ {\isacharminus}\ lvl\ s\ t{\isacharparenright}{\isacharparenright} {\isasymand} lvl\ s\ t\ {\isasymle}\ level\ b

\quad {\isasymand}\ {\isacharparenleft}free{\isacharunderscore}block{\isacharunderscore}r\ s\ t\ {\isasymlongrightarrow}\  {\isacharparenleft}{\isasymexists}blk{\isachardot}\ freeing{\isacharunderscore}node\ s\ t\ {\isacharequal}\ Some\ blk\ {\isasymand}\ pool\ blk\ {\isacharequal}\ pool\ b\isanewline
\quad \quad \quad \quad \quad \quad \quad \quad \quad \quad \quad \quad \quad \quad \quad 
{\isasymand}\ level\ blk\ {\isacharequal}\ lvl\ s\ t\ {\isasymand}\ block\ blk\ {\isacharequal}\ bn\ s\ t{\isacharparenright}

\quad \quad \quad \quad \quad \quad \quad \quad \quad \quad {\isasymand}\ alloc{\isacharunderscore}memblk{\isacharunderscore}data{\isacharunderscore}valid\ s\ {\isacharparenleft}pool\ b{\isacharparenright}\ {\isacharparenleft}the\ {\isacharparenleft}freeing{\isacharunderscore}node\ s\ t{\isacharparenright}{\isacharparenright}{\isacharparenright}

\quad {\isasymand}\ {\isacharparenleft}{\isasymnot}\ free{\isacharunderscore}block{\isacharunderscore}r\ s\ t\ {\isasymlongrightarrow}\ freeing{\isacharunderscore}node\ s\ t\ {\isacharequal}\ None{\isacharparenright} \ {\isasymand}\ ... \ {\isacharbraceright}\ {\isasyminter} 

\quad {\isacharbraceleft}s.\ {\isasymalpha}\ {\isacharequal}\ {\isacharparenleft}\textbf{if}\ freeing{\isacharunderscore}node\ s\ t\ {\isasymnoteq}\ None\ \textbf{then}\ lvl\ s\ t\ {\isacharplus}\ {\isadigit{1}}\ \textbf{else}\ {\isadigit{0}}{\isacharparenright}\ {\isacharbraceright}
\end{isabellec}

$freeing\_node$ and $lvt$ are local variables respectively storing the node being free and the level that the node belongs to. 
In the body of the loop, if $lvl\ t\ {\isachargreater}\ {\isadigit{0}}$ and $partner\_bit$ is \emph{true}, then $lvl = lvl - 1$ at the end of the body. Otherwise, $freeing\_node \ t= None$. So at the end of the loop body, $\alpha$ decreases or $\alpha = 0$. If $\alpha = 0$, we have $freeing\_node\ t = None$, and thus the negation of the loop condition  $\neg free\_block\_r \ t$, concluding termination of \emph{free\_block}. 

\section{Verified Translation from BPEL to {\slang}}
\label{sect:bpel}

In the case study of Zephyr concurrent memory management, we have shown the simplicity of the formal specification of {\slang} for concurrent OS kernels as well as the power of {\slang} for compositional reasoning of CRSs with complex structures and algorithms. 
This section presents another class of CRSs, i.e. concurrent business processes in BPEL language \cite{bpel}, which has more complex event patterns. This kind of CRSs has a flow reaction structure. 
Due to the wide application of business process programming, especially in safety/security critical systems, it is essential to ensure correctness and safety/security of business processes in order to avoid errors that may cause critical losses to the involved organizations. 

In this case study, we create an abstract syntax for BPEL and implement a formally verified translation from BPEL into {\slang}, rather than use the {\slang} to model the BPEL programs directly. This approach alleviates the efforts of manually modeling BPEL programs by automatic translation and provides high confidence with correctness proof. By the correct translation, the functional correctness and safety/security of BPEL programs can be compositionally verified in {\slang}. 

The correctness of the translation or compilation is usually defined as a simulation relation between the source and target languages, \cite{Leroy2009,LiangFF12}. Thus, we define a strong bisimulation relation for the correctness of the translation from BPEL to {\slang}. {\slang} is designed having reaction structures in mind, and it is expressive enough to simulate each small step of BPEL processes. 

In this section, we first introduce the BPEL language in brief. Then, we present an abstract syntax and its operational semantics of a core subset of BPEL in Isabelle/HOL. Next, we discuss the translation of BPEL into {\pcimp} and its correctness proof. 

\subsection{Background of BPEL}
Service-oriented computing represents a paradigm for building distributed computing applications over the Internet, where services are the basic blocks and applications are constructed by interoperations among services. From a business perspective, services composition drastically reduces the cost and risks of building new business applications in the sense that existing business logics are represented as Web services and could be reused \cite{Sheng14}. Service compositions are programmed as executable business processes in languages like BPEL \cite{bpel}, which is an imperative, XML-based language. 
In BPEL, activities are nested within concurrency, compensation and event handling constructs that cause an overwhelming number of execution paths. 

In {\figprefix} \ref{fig:bpel_example}, we illustrate the structure of business processes in BPEL with a simple example adopted from the BPEL standard. The labels in blue color show the type of each activity. 
Solid arrows in black color represent sequencing. Solid arrows in red color represent control links used for synchronization across concurrent activities. 
Activities in BPEL perform the process logic. Basic activities are those describing elemental steps of the process behaviour, such as \textbf{invoke}, \textbf{receive}, \textbf{reply}, \textbf{assign}, \textbf{wait}. The \textbf{wait} activity specifies a delay for a certain period of time or until a certain deadline is reached. 
Structured activities encode control-flow logic, and therefore can contain other basic and/or structured activities recursively, such as \textbf{sequence}, \textbf{if}, \textbf{while}, \textbf{repeatUntil}, \textbf{forEach}, \textbf{pick}, \textbf{flow}. The \textbf{pick} activity waits for the occurrence of exactly one event from a set of events, then executes the activity associated with that event. The \textbf{flow} activity provides concurrency and synchronization using \emph{links}.

\begin{figure}
\begin{center}
\includegraphics[width=3.2in]{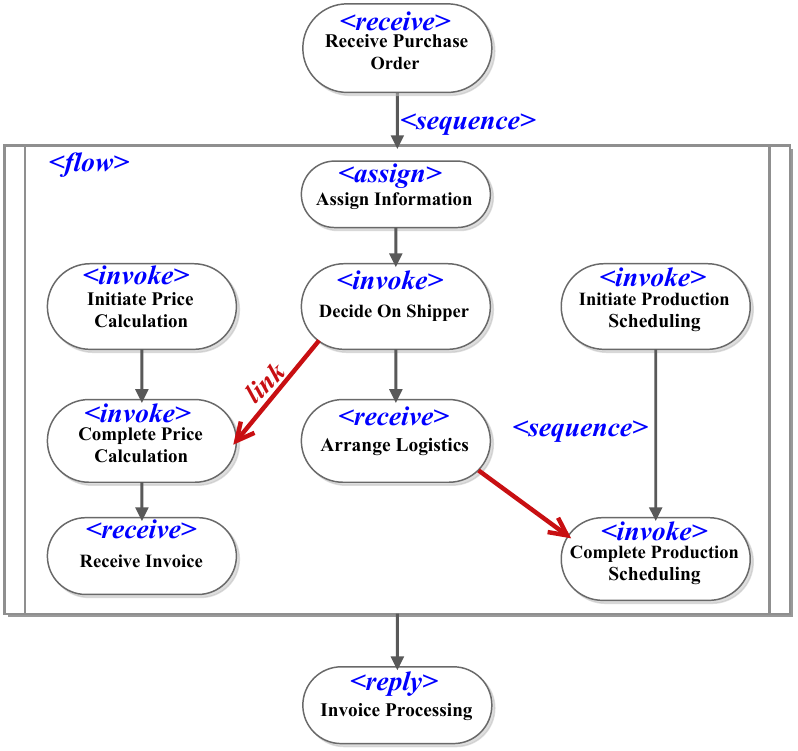}
\end{center}
\caption{An Example of BPEL \cite{bpel}}
\label{fig:bpel_example}
\end{figure}


\subsection{Formalizing the BPEL Language in Isabelle/HOL}

We design an abstract syntax of a core subset of the BPEL language as shown in {\figprefix} \ref{fig:bpel_syntax}. 
The \textbf{invoke} activity is used to call a Web Service offered by service providers, which is specified by the service's partner link \emph{ptlink}, port type \emph{pttype} and operation \emph{op}. The invocation returns normally, and the state is updated by the invocation result, which is specified as a mapping \emph{spec} of states. It also may return a fault message identified by a \emph{QName}. 
The \textbf{catch} construct attached to the \textbf{invoke} activity provides a way to define a set of custom fault handling activities, where each of them is defined to intercept a specific kind of fault. A default handler \textbf{catchAll} clause can be added to catch any fault not caught by a more specific fault handler. 
Whenever a \textbf{catchAll} fault handler is missing, it must be implicitly created. Thus, we assume there is a \textbf{catchAll} handler for each \textbf{invoke} activity.

\begin{figure}[t] 
\begin{flushleft}
\isabellestyle{sl}
	\begin{isabellebody} 
		\small 


\isacommand{record}\isamarkupfalse%
\ {\isacharparenleft}{\isacharprime}s{\isacharcomma}{\isacharprime}l{\isacharcomma}{\isacharprime}c{\isacharparenright}\ Flow{\isacharunderscore}Ele\ {\isacharequal}\ targets\ {\isacharcolon}{\isacharcolon}\ {\isachardoublequoteopen}{\isacharparenleft}{\isacharparenleft}{\isacharparenleft}{\isacharprime}s{\isacharcomma}{\isacharprime}l{\isacharcomma}{\isacharprime}c{\isacharparenright}\ State\ {\isasymRightarrow}\ bool{\isacharparenright}\ {\isasymtimes}\ {\isacharprime}l\ list{\isacharparenright}\ option{\isachardoublequoteclose}\ 

\quad \quad \quad \quad \quad \quad \quad \quad \quad \quad \quad 
sources\ {\isacharcolon}{\isacharcolon}\ {\isachardoublequoteopen}{\isacharparenleft}{\isacharprime}l\ {\isasymtimes}\ {\isacharparenleft}{\isacharparenleft}{\isacharprime}s{\isacharcomma}{\isacharprime}l{\isacharcomma}{\isacharprime}c{\isacharparenright}\ State\ {\isasymRightarrow}\ bool{\isacharparenright}{\isacharparenright}\ list\ option{\isachardoublequoteclose}

\isacommand{datatype}\isamarkupfalse%
\ {\isacharparenleft}{\isacharprime}s{\isacharcomma}{\isacharprime}l{\isacharcomma}{\isacharprime}c{\isacharparenright}\ Activity\ {\isacharequal}\ \isanewline
\quad \ \ \textbf{\textit{Invoke}} \ {\isacharparenleft}fe{\isacharcolon}{\isachardoublequoteopen}{\isacharparenleft}{\isacharprime}s{\isacharcomma}{\isacharprime}l{\isacharcomma}{\isacharprime}c{\isacharparenright}\ Flow{\isacharunderscore}Ele{\isachardoublequoteclose}{\isacharparenright}\ {\isacharparenleft}ptlink{\isacharcolon}NCName{\isacharparenright}\ {\isacharparenleft}pttype{\isacharcolon}QName{\isacharparenright}\ {\isacharparenleft}op{\isacharcolon}NCName{\isacharparenright}

\quad \quad \quad \quad \quad  {\isacharparenleft}spec{\isacharcolon}{\isachardoublequoteopen}{\isacharparenleft}{\isacharprime}s{\isacharcomma}{\isacharprime}l{\isacharcomma}{\isacharprime}c{\isacharparenright}\ State\ {\isasymRightarrow}\ {\isacharparenleft}{\isacharprime}s{\isacharcomma}{\isacharprime}l{\isacharcomma}{\isacharprime}c{\isacharparenright}\ State{\isachardoublequoteclose}{\isacharparenright}

\quad \quad \quad \quad \quad  {\isacharparenleft}catches{\isacharcolon}{\isachardoublequoteopen}{\isacharparenleft}QName\ {\isasymtimes}\ {\isacharparenleft}{\isacharparenleft}{\isacharprime}s{\isacharcomma}{\isacharprime}l{\isacharcomma}{\isacharprime}c{\isacharparenright}\ Activity{\isacharparenright}{\isacharparenright}\ list{\isachardoublequoteclose}{\isacharparenright}\ {\isacharparenleft}catchall{\isacharcolon}{\isachardoublequoteopen}{\isacharparenleft}{\isacharprime}s{\isacharcomma}{\isacharprime}l{\isacharcomma}{\isacharprime}c{\isacharparenright}\ Activity{\isachardoublequoteclose}{\isacharparenright}

\quad {\isacharbar}\ \textbf{\textit{Receive}} \ {\isacharparenleft}fe{\isacharcolon}{\isachardoublequoteopen}{\isacharparenleft}{\isacharprime}s{\isacharcomma}{\isacharprime}l{\isacharcomma}{\isacharprime}c{\isacharparenright}\ Flow{\isacharunderscore}Ele{\isachardoublequoteclose}{\isacharparenright}\ {\isacharparenleft}ptlink{\isacharcolon}NCName{\isacharparenright}\ {\isacharparenleft}pttype{\isacharcolon}QName{\isacharparenright}\ {\isacharparenleft}op{\isacharcolon}NCName{\isacharparenright}

\quad \quad \quad \quad \quad  {\isacharparenleft}spec{\isacharcolon}{\isachardoublequoteopen}{\isacharparenleft}{\isacharprime}s{\isacharcomma}{\isacharprime}l{\isacharcomma}{\isacharprime}c{\isacharparenright}\ State\ {\isasymRightarrow}\ {\isacharparenleft}{\isacharprime}s{\isacharcomma}{\isacharprime}l{\isacharcomma}{\isacharprime}c{\isacharparenright}\ State{\isachardoublequoteclose}{\isacharparenright}

\quad {\isacharbar}\ \textbf{\textit{Reply}} \ {\isacharparenleft}fe{\isacharcolon}{\isachardoublequoteopen}{\isacharparenleft}{\isacharprime}s{\isacharcomma}{\isacharprime}l{\isacharcomma}{\isacharprime}c{\isacharparenright}\ Flow{\isacharunderscore}Ele{\isachardoublequoteclose}{\isacharparenright}\ {\isacharparenleft}ptlink{\isacharcolon}NCName{\isacharparenright}\ {\isacharparenleft}pttype{\isacharcolon}QName{\isacharparenright}\ {\isacharparenleft}op{\isacharcolon}NCName{\isacharparenright}

\quad {\isacharbar}\ \textbf{\textit{Assign}} \ {\isacharparenleft}fe{\isacharcolon}{\isachardoublequoteopen}{\isacharparenleft}{\isacharprime}s{\isacharcomma}{\isacharprime}l{\isacharcomma}{\isacharprime}c{\isacharparenright}\ Flow{\isacharunderscore}Ele{\isachardoublequoteclose}{\isacharparenright}\ \ \ {\isachardoublequoteopen}{\isacharparenleft}{\isacharprime}s{\isacharcomma}{\isacharprime}l{\isacharcomma}{\isacharprime}c{\isacharparenright}\ State\ {\isasymRightarrow}\ {\isacharparenleft}{\isacharprime}s{\isacharcomma}{\isacharprime}l{\isacharcomma}{\isacharprime}c{\isacharparenright}\ State{\isachardoublequoteclose}

\quad
{\isacharbar}\ \textbf{\textit{Wait}} \ {\isacharparenleft}fe{\isacharcolon}{\isachardoublequoteopen}{\isacharparenleft}{\isacharprime}s{\isacharcomma}{\isacharprime}l{\isacharcomma}{\isacharprime}c{\isacharparenright}\ Flow{\isacharunderscore}Ele{\isachardoublequoteclose}{\isacharparenright}\ Time

\quad {\isacharbar}\ \textbf{\textit{Empty}} \ {\isacharparenleft}fe{\isacharcolon}{\isachardoublequoteopen}{\isacharparenleft}{\isacharprime}s{\isacharcomma}{\isacharprime}l{\isacharcomma}{\isacharprime}c{\isacharparenright}\ Flow{\isacharunderscore}Ele{\isachardoublequoteclose}{\isacharparenright}

\quad
{\isacharbar}\ \textbf{\textit{Seq}} \ \ {\isachardoublequoteopen}{\isacharparenleft}{\isacharprime}s{\isacharcomma}{\isacharprime}l{\isacharcomma}{\isacharprime}c{\isacharparenright}\ Activity{\isachardoublequoteclose}\ \ \ {\isachardoublequoteopen}{\isacharparenleft}{\isacharprime}s{\isacharcomma}{\isacharprime}l{\isacharcomma}{\isacharprime}c{\isacharparenright}\ Activity{\isachardoublequoteclose}

\quad {\isacharbar}\ \textbf{\textit{If}} \ \ {\isacharparenleft}cond{\isacharcolon}{\isachardoublequoteopen}{\isacharparenleft}{\isacharprime}s{\isacharcomma}{\isacharprime}l{\isacharcomma}{\isacharprime}c{\isacharparenright}\ State\ set{\isachardoublequoteclose}{\isacharparenright}\ \ \ {\isachardoublequoteopen}{\isacharparenleft}{\isacharprime}s{\isacharcomma}{\isacharprime}l{\isacharcomma}{\isacharprime}c{\isacharparenright}\ Activity{\isachardoublequoteclose}\ \ \ {\isachardoublequoteopen}{\isacharparenleft}{\isacharprime}s{\isacharcomma}{\isacharprime}l{\isacharcomma}{\isacharprime}c{\isacharparenright}\ Activity{\isachardoublequoteclose}

\quad {\isacharbar}\ \textbf{\textit{While}} \ \ {\isacharparenleft}cond{\isacharcolon}{\isachardoublequoteopen}{\isacharparenleft}{\isacharprime}s{\isacharcomma}{\isacharprime}l{\isacharcomma}{\isacharprime}c{\isacharparenright}\ State\ set{\isachardoublequoteclose}{\isacharparenright}\ \ \ {\isachardoublequoteopen}{\isacharparenleft}{\isacharprime}s{\isacharcomma}{\isacharprime}l{\isacharcomma}{\isacharprime}c{\isacharparenright}\ Activity{\isachardoublequoteclose}

\quad
{\isacharbar}\ \textbf{\textit{Flow}}  \ {\isachardoublequoteopen}{\isacharparenleft}{\isacharprime}s{\isacharcomma}{\isacharprime}l{\isacharcomma}{\isacharprime}c{\isacharparenright}\ Activity{\isachardoublequoteclose}\ \ \ {\isachardoublequoteopen}{\isacharparenleft}{\isacharprime}s{\isacharcomma}{\isacharprime}l{\isacharcomma}{\isacharprime}c{\isacharparenright}\ Activity{\isachardoublequoteclose}

\quad {\isacharbar}\ \textbf{\textit{Pick}} \ \ {\isachardoublequoteopen}{\isacharparenleft}{\isacharparenleft}{\isacharprime}s{\isacharcomma}{\isacharprime}l{\isacharcomma}{\isacharprime}c{\isacharparenright}\ EventHandler{\isacharparenright}{\isachardoublequoteclose} \ \ {\isachardoublequoteopen}{\isacharparenleft}{\isacharparenleft}{\isacharprime}s{\isacharcomma}{\isacharprime}l{\isacharcomma}{\isacharprime}c{\isacharparenright}\ EventHandler{\isacharparenright}{\isachardoublequoteclose}

\quad {\isacharbar}\ \textbf{\textit{\bpelfin}}


\isakeyword{and}\ {\isacharparenleft}{\isacharprime}s{\isacharcomma}{\isacharprime}l{\isacharcomma}{\isacharprime}c{\isacharparenright}\ EventHandler\ {\isacharequal}

\quad \ \ \textbf{\textit{OnMessage}} \ {\isacharparenleft}ptlink{\isacharcolon}NCName{\isacharparenright}\ {\isacharparenleft}pttype{\isacharcolon}QName{\isacharparenright}\ {\isacharparenleft}op{\isacharcolon}NCName{\isacharparenright}

\quad \quad \ \ \ \ \ \ \ \ \ \ \ \ {\isacharparenleft}spec{\isacharcolon}{\isachardoublequoteopen}{\isacharparenleft}{\isacharprime}s{\isacharcomma}{\isacharprime}l{\isacharcomma}{\isacharprime}c{\isacharparenright}\ State\ {\isasymRightarrow}\ {\isacharparenleft}{\isacharprime}s{\isacharcomma}{\isacharprime}l{\isacharcomma}{\isacharprime}c{\isacharparenright}\ State{\isachardoublequoteclose}{\isacharparenright}\ \ \ {\isachardoublequoteopen}{\isacharparenleft}{\isacharprime}s{\isacharcomma}{\isacharprime}l{\isacharcomma}{\isacharprime}c{\isacharparenright}\ Activity{\isachardoublequoteclose}

\quad {\isacharbar}\ \textbf{\textit{OnAlarm}} \ Time\ \ \ {\isachardoublequoteopen}{\isacharparenleft}{\isacharprime}s{\isacharcomma}{\isacharprime}l{\isacharcomma}{\isacharprime}c{\isacharparenright}\ Activity{\isachardoublequoteclose}

\textbf{\textit{repeatUtil}}\ c\ P $\equiv$ \textbf{\textit{Seq}} \ P \ (\textbf{\textit{While}}\ c\ P)

\textbf{\textit{forEach}}\ m\ n\ P $\equiv$ 
\textbf{\textit{if}}\ m = n \textbf{\textit{then}}\ P

\quad \quad \quad \quad \quad \quad \quad 
\textbf{\textit{else}}\ \textbf{\textit{if}}\ m > n \ \textbf{\textit{then}}\ \textbf{\textit{\bpelfin}}

\quad \quad \quad \quad \quad \quad \quad 
\textbf{\textit{else}}\ \textbf{\textit{Seq}} \ P (\textbf{\textit{forEach}}\ (m + 1)\ n\ P)

	\end{isabellebody}%
	\caption{Syntax of BPEL in Isabelle/HOL} 
	\label{fig:bpel_syntax}
\end{flushleft}
\end{figure}

A business process provides services to its partners through inbound message activities (e.g. \textbf{receive} and \textbf{pick}) and corresponding \textbf{reply} activities. 
A \textbf{receive} activity specifies a partner link \emph{ptlink} to receive messages, the port type \emph{pttype} and operation \emph{op} that it expects the partner to invoke. A \textbf{receive} activity is blocked until it receives a message and the state is updated by assigning the message to a variable. 
The \textbf{reply} activity is used to send a response to a request previously accepted through an inbound message activity such as a \textbf{receive} activity. This activity only queries the current state for sending messages; however, the state is not changed.  

The \textbf{assign} activity can be used to copy data from one variable to another, as well as to construct and insert new data using expressions. A state mapping specifies this activity. The \textbf{wait} activity specifies a delay until a certain deadline is reached. The \textbf{empty} activity specifies an activity that does nothing, for example when a fault needs to be caught and suppressed. Another use of \textbf{empty} is to provide a synchronization point in a \textbf{flow} structure. 

The \textbf{if}, \textbf{seq} and \textbf{while} activities are standard, and \textbf{repeatUntil} is similar to \emph{do \{ ... \} while(b)} statement. 
The \textbf{forEach} activity has two execution modes. In sequential mode, when its parallel attribute is \emph{false}, the \textbf{forEach} activity will execute its contained activity exactly $N+1$ times where $N = endv - startv$. Otherwise, the activity is a parallel \textbf{forEach}. The enclosed activity must be concurrently executed $N+1$ times. Essentially, it is the same as a \textbf{flow} activity with $N+1$ copies of the \textbf{forEach}'s enclosed activity as children. Thus, we could only consider sequential \textbf{forEach} activities for formal verification. 

The \textbf{pick} activity waits for the occurrence of exactly one event from a set of events, then executes the activity associated with that event. After an event has been selected, the other events are no longer accepted by that \textbf{pick}. The activity is comprised of a set of event handlers, each containing an event-activity pair in two forms: (1) The \textbf{onMessage} is similar to a \textbf{receive} activity, in that it waits for the receipt of an inbound message. (2) The \textbf{onAlarm} corresponds to a timer-based alarm. 

A \textbf{flow} activity completes when all of the activities enclosed by the \textbf{flow} have completed.
The synchronization of concurrent activities is specified by a \emph{flow element} attached to basic activities. 
An activity can declare itself to be the source of one or more links by including one or more \emph{source} elements (\emph{sources}). 
Each \emph{source} element can specify an optional \emph{transition condition} as a guard for the specified link. After executing the activity, the \emph{source}'s specified link is flagged as \emph{fired}, if the transition condition of a \emph{source} is evaluated to \emph{true}. 
Similarly, an activity can declare itself to be the target of one or more links by including one or more \emph{target} elements (\emph{targets}). \emph{Targets}, as a whole, can specify an optional \emph{join condition}, which is a boolean expression. If the condition is \emph{true}, the activity must be blocked until all links specified in \emph{targets} are fired. If no \emph{join condition} is specified, the condition is the disjunction of the link status of all incoming links of this activity. 
Finally, we add {\bpelfin} to denote the termination of BPEL execution. 

Typical transition rules of BPEL are shown in {\figprefix} \ref{fig:bpel_semantics}. The invoked service nondeterministically decides the occurrence of faults when a BPEL process invoking a Web services, so we choose one fault handler to be executed (\bpelsemrulename{InvokeFault}). Other rules refer to the explanation of each activity above.

\begin{figure}
\small 

\begin{tabular}{c}
\bpelsemrule{InvokeSuc}
{
	\infer{ \bpeltran{\textbf{Invoke}\ fls\ ptl\ ptt\ opn\ spc\ cts\ cta}{s}{\bpelfin}{t} }
		  { \text{targets\_sat}\ (targets\ fls)\ s 
		  & r = spc\ s 
		  & t = \text{fire\_sources}\ (sources\ fls)\ r }
}
\end{tabular}
\vspace{0.1cm}

\begin{tabular}{c}
\bpelsemrule{InvokeFault}
{
	\infer{ \bpeltran{\textbf{Invoke}\ fls\ ptl\ ptt\ opn\ spc\ cts\ cta}{s}{evh}{t} }
		  { 
		  \begin{tabular}{l}
			$\text{targets\_sat}\ (targets\ fls)\ s$ \quad $t = \text{fire\_sources}\ (sources\ fls)\ s$ \quad
			$evh \in set\ (map\ snd\ cts) \cup \{cta\}$
		  \end{tabular}
		  }
}
\end{tabular}
\vspace{0.1cm}

\begin{tabular}{cc}
\bpelsemrule{PickR}
{
	\infer{ \bpeltran{\textbf{Pick}\ P \ Q}{s}{Q'}{t} }
		  { \bpelehtran{Q}{s}{Q'}{t} }
}
&
\bpelsemrule{FlowR}
{
	\infer{ \bpeltran{\textbf{Flow}\ P\ Q}{s}{\textbf{Flow}\ P\ O'}{t} }
		  { \bpeltran{Q}{s}{Q'}{t} }
}
\end{tabular}
\vspace{0.1cm}

\begin{tabular}{c}
\bpelsemrule{FlowFin}
{
	\infer{ \bpeltran{\textbf{Flow}\ \bpelfin\ \bpelfin}{s}{\bpelfin}{s} }
		  { - }
}
\end{tabular}
\vspace{0.1cm}

\begin{tabular}{cc}
\bpelsemrule{OnMessage}
{
	\infer{ \bpelehtran{\textbf{OnMessage}\ ptl\ ptt\ opn\ spc\ A}{s}{A}{t} }
		  { t = spc\ s }
}
&
\bpelsemrule{OnAlarm}
{
	\infer{ \bpelehtran{\textbf{OnAlarm}\ time\ A}{s}{A}{s} }
		  { time > tick\ s }
}
\end{tabular}

\caption{Operational Semantics of BPEL}
\label{fig:bpel_semantics}
\end{figure}

\subsection{Translating BPEL into {\slang}}
We have developed a translator for the formal verification of BPEL programs using the {\slang} framework. In order to show that {\slang} is expressive to model concurrent business processes in BPEL, we use the event compositions to represent complex activities, rather than the complex statements in integrated languages such as {\csimpllang}. We only use the \emph{Basic} statement in the {\implang} language ({\sectprefix} \ref{subsect:pcimp}) to change the state.  
For the formal verification, we use the same state definition in the translated {\slang} specification as in the BPEL program. 
BPEL composes Web services in the \emph{orchestration} manner \cite{Peltz03}, which always represents control from one party's perspective.
Formal verification of BPEL can only concern one BPEL program, rather than a set of composed programs, which would be a Web services \emph{choreography}. 
Moreover, a BPEL program composing a set of external services has embedded concurrent internal sub-processes. Therefore, we use event systems to represent BPEL processes but not parallel event systems. In this way, the events translated from BPEL activities do not have the parameter ``$\textbf{OF}\ t$'' to indicate the execution context. 

We implement a \emph{compile} function to translate a BPEL process into a {\slang} specification. 
We first introduce the translation of the \textbf{invoke} activity as shown in {\figprefix} \ref{fig:compile_invoke}. Other translation rules are shown in {\tableprefix} \ref{tbl:translate_basic}, which will be discussed later.


%
%

For each basic activity, the translated events are guarded by a condition 
\[
{\isasymacute}{\isacharparenleft}{\isasymlambda}s{\isachardot}\ targets{\isacharunderscore}sat\ {\isacharparenleft}targets\ fls{\isacharparenright}\ s{\isacharparenright}
\]
which constrains the state that all targeted links of the activity are fired. Meanwhile, in the body of the translated events, the links sourced by the activity are fired by a statement 
\[
Basic\ {\isacharparenleft}{\isasymlambda}s{\isachardot}\ fire{\isacharunderscore}sources\ {\isacharparenleft}sources\ fls{\isacharparenright}\ s{\isacharparenright}
\]
where $Basic\ spc$ represents an atomic state transformation, for example, an assignment, a multiple assignments, or the $SKIP$ command ($SKIP \equiv Basic\ Id$).

For the \textbf{invoke} activity, we use the \textbf{$\textbf{Event}_{suc}$} and \textbf{$\textbf{Event}_{fail}$} events to simulate the cases of normal and faulty invocation respectively. Since the execution of a \textbf{invoke} activity is one step in BPEL semantics, we use atomic events (\textbf{EVENT\isactrlsub A}) to model it. In the normal case, the event updates the state by the invocation result and flags its source links. When the invocation returns a fault, it is nondeterministically resolved by special handlers and the default handler. 
For the list of handlers in an invocation, we use the \lchoice{}operator to compose the events in a list by the choice construct. Then, we use the $\oplus$ and \lchoice{}operators to simulate the nondeterministism among successful invocation, failed invocations caught by the default and custom fault handlers. The $\textbf{Event}_{suc}$ event represents the successful invocation in {\slang}. The sequence $\evtseq{ \textbf{Event\isactrlsub {fail}} }{ \textbf{compile}(cta) }$ represents a failed invocation and the default handler  ($cta$) deals with the fault. The ``map\ (($\evtseq{\textbf{Event\isactrlsub {fail}}}{}$) $\circ$ \textbf{compile})\ (map\ snd\ cts)'' expression constructs a list of event sequence, each of which is $\evtseq{ \textbf{Event\isactrlsub {fail}} }{ \textbf{compile}(ct) }$ corresponding to each $(name,ct) \in cts$. 
The partner link, port type and operation name are used to construct the event name.

\begin{figure}
\begin{flushleft}
\isabellestyle{sl}
\small
$
\left\{
\begin{aligned}
& \lchoice{[]} = \fin \\
& \lchoice{[a]} = a \\
& \lchoice{(a \ \# \ b \ \# \ xs)} = \evtchoice{a}{(\lchoice{(b \ \# \ xs)})}
\end{aligned}
\right.
$
\vspace{0.2cm} 

\begin{minipage}{.49\linewidth}
{\slshape
\textbf{Event\isactrlsub {suc}} {\isasymequiv} 
\textbf{EVENT\isactrlsub A}\ {\isacharparenleft}{\isacharprime}{\isacharprime}Invoke{\isacharprime}{\isacharprime}{\isacharat}ptl{\isacharat}ptt{\isacharat}opn{\isacharparenright}

\quad \quad \quad \quad \quad
\textbf{WHEN}\ {\isacharparenleft}{\isasymacute}{\isacharparenleft}{\isasymlambda}s{\isachardot}\ targets{\isacharunderscore}sat\ {\isacharparenleft}targets\ fls{\isacharparenright}\ s{\isacharparenright}{\isacharparenright}\ 

\quad \quad \quad \quad \quad
\textbf{THEN}

\quad \quad \quad \quad \quad \quad
{\isacharparenleft}Basic\ spc{\isacharparenright}{\isacharsemicolon}{\isacharsemicolon}

\quad \quad \quad \quad \quad \quad
{\isacharparenleft}Basic\ {\isacharparenleft}{\isasymlambda}s{\isachardot}\ fire{\isacharunderscore}sources\ {\isacharparenleft}sources\ fls{\isacharparenright}\ s{\isacharparenright}{\isacharparenright}

\quad \quad \quad \quad \quad
\textbf{END} 
}

\end{minipage}
\begin{minipage}{.5\linewidth}
\vspace{-3mm}
{\slshape
\textbf{Event\isactrlsub {fail}} {\isasymequiv} 
\textbf{EVENT\isactrlsub A}\ {\isacharparenleft}{\isacharprime}{\isacharprime}Invoke{\isacharprime}{\isacharprime}{\isacharat}ptl{\isacharat}ptt{\isacharat}opn{\isacharparenright}
    
\quad \quad \quad \quad \quad
\textbf{WHEN}\ {\isacharparenleft}{\isasymacute}{\isacharparenleft}{\isasymlambda}s{\isachardot}\ targets{\isacharunderscore}sat\ {\isacharparenleft}targets\ fls{\isacharparenright}\ s{\isacharparenright}{\isacharparenright}\ 

\quad \quad \quad \quad \quad
\textbf{THEN}

\quad \quad \quad \quad \quad \quad
{\isacharparenleft}Basic\ {\isacharparenleft}{\isasymlambda}s{\isachardot}\ fire{\isacharunderscore}sources\ {\isacharparenleft}sources\ fls{\isacharparenright}\ s{\isacharparenright}{\isacharparenright}

\quad \quad \quad \quad \quad
\textbf{END}
}
\end{minipage}
\vspace{0.2cm} 

\textbf{compile}(\textbf{Invoke} \ fls\ ptl\ ptt\ opn\ spc\ cts\ cta) {\isasymequiv} 

\quad 
\textbf{Event\isactrlsub {succ}}\ $\oplus$
{\isacharparenleft}
\lchoice{(($\evtseq{ \textbf{Event\isactrlsub {fail}} }{ \textbf{compile}(cta) }$) \# (\textbf{map}\ (($\evtseq{\textbf{Event\isactrlsub {fail}}}{}$) $\circ$ \textbf{compile})\ (\textbf{map}\ snd\ cts)  ))}
{\isacharparenright}


\caption{Translation of \textbf{Invoke} Activity}
\label{fig:compile_invoke}
\end{flushleft}
\end{figure}

\begin{table}[t]
  \centering
  \small 
  \caption{The \textbf{compile} function} 
  \begin{tabularx}{\textwidth}{|l|X|}
  \hline
    \textbf{BPEL Activity} & \textbf{Translated {\slang} Specification} 
\\ \hline

\textbf{Receive}\ fls\ ptl\ ptt\ opn\ spc
	&
	{\slshape
\textbf{EVENT\isactrlsub A} \ {\isacharparenleft}{\isacharprime}{\isacharprime}Receive{\isacharprime}{\isacharprime}{\isacharat}ptl{\isacharat}ptt{\isacharat}opn{\isacharparenright}\ \isanewline
\textbf{WHEN} \ {\isacharparenleft}{\isasymacute}{\isacharparenleft}{\isasymlambda}s{\isachardot}\ targets{\isacharunderscore}sat\ {\isacharparenleft}targets\ fls{\isacharparenright}\ s{\isacharparenright}{\isacharparenright}\ \isanewline
\textbf{THEN} \ \
{\isacharparenleft}Basic\ spc{\isacharparenright}{\isacharsemicolon}{\isacharsemicolon}\ {\isacharparenleft}Basic\ {\isacharparenleft}{\isasymlambda}s{\isachardot}\ fire{\isacharunderscore}sources\ {\isacharparenleft}sources\ fls{\isacharparenright}\ s{\isacharparenright}{\isacharparenright}\ \
\textbf{END}
    }
\\ \hline

	\textbf{Reply}\ fls\ ptl\ ptt\ opn
	&
	{\slshape
\textbf{EVENT\isactrlsub A} \ {\isacharparenleft}{\isacharprime}{\isacharprime}Reply{\isacharprime}{\isacharprime}{\isacharat}ptl{\isacharat}ptt{\isacharat}opn{\isacharparenright}\ \isanewline
\textbf{WHEN} \ {\isacharparenleft}{\isasymacute}{\isacharparenleft}{\isasymlambda}s{\isachardot}\ targets{\isacharunderscore}sat\ {\isacharparenleft}targets\ fls{\isacharparenright}\ s{\isacharparenright}{\isacharparenright}\ \isanewline
\textbf{THEN} \ \
{\isacharparenleft}Basic\ {\isacharparenleft}{\isasymlambda}s{\isachardot}\ fire{\isacharunderscore}sources\ {\isacharparenleft}sources\ fls{\isacharparenright}\ s{\isacharparenright}{\isacharparenright}\ \
\textbf{END}
    }
\\ \hline

\textbf{Assign}\ fls\ spc
	&
	{\slshape
\textbf{EVENT\isactrlsub A} \ {\isacharparenleft}{\isacharprime}{\isacharprime}Assign{\isacharprime}{\isacharprime}{\isacharparenright}\ \isanewline
\textbf{WHEN} \ {\isacharparenleft}{\isasymacute}{\isacharparenleft}{\isasymlambda}s{\isachardot}\ targets{\isacharunderscore}sat\ {\isacharparenleft}targets\ fls{\isacharparenright}\ s{\isacharparenright}{\isacharparenright}\ \isanewline
\textbf{THEN} \ \
{\isacharparenleft}Basic\ spc{\isacharparenright}{\isacharsemicolon}{\isacharsemicolon}\ {\isacharparenleft}Basic\ {\isacharparenleft}{\isasymlambda}s{\isachardot}\ fire{\isacharunderscore}sources\ {\isacharparenleft}sources\ fls{\isacharparenright}\ s{\isacharparenright}{\isacharparenright}\ \
\textbf{END}
    }
\\ \hline

	\textbf{Wait}\ fls\ t
	&
	{\slshape
\textbf{EVENT\isactrlsub A}\ {\isacharparenleft}{\isacharprime}{\isacharprime}Wait{\isacharprime}{\isacharprime}{\isacharparenright}\ \isanewline
\textbf{WHEN}\ t\ < \ {\isasymacute}tick\ {\isasymand}\ {\isacharparenleft}{\isasymacute}{\isacharparenleft}{\isasymlambda}s{\isachardot}\ targets{\isacharunderscore}sat\ {\isacharparenleft}targets\ fls{\isacharparenright}\ s{\isacharparenright}{\isacharparenright}\ \isanewline
\textbf{THEN}\ \
{\isacharparenleft}Basic\ {\isacharparenleft}{\isasymlambda}s{\isachardot}\ fire{\isacharunderscore}sources\ {\isacharparenleft}sources\ fls{\isacharparenright}\ s{\isacharparenright}{\isacharparenright}\ \
\textbf{END}
    }
\\ \hline

\textbf{Empty}\ fls
	&
	{\slshape
\textbf{EVENT\isactrlsub A} \ {\isacharparenleft}{\isacharprime}{\isacharprime}Empty{\isacharprime}{\isacharprime}{\isacharparenright}\ \isanewline
\textbf{WHEN} \ {\isacharparenleft}{\isasymacute}{\isacharparenleft}{\isasymlambda}s{\isachardot}\ targets{\isacharunderscore}sat\ {\isacharparenleft}targets\ fls{\isacharparenright}\ s{\isacharparenright}{\isacharparenright}\ \isanewline
\textbf{THEN} \ \
{\isacharparenleft}Basic\ {\isacharparenleft}{\isasymlambda}s{\isachardot}\ fire{\isacharunderscore}sources\ {\isacharparenleft}sources\ fls{\isacharparenright}\ s{\isacharparenright}{\isacharparenright}\ \
\textbf{END}
    }
\\ \hline

    \textbf{Seq} \ $A_1$\ $A_2$
   & 
    $\evtseq{(compile\ A_1)}{(compile\ A_2)}$
\\ \hline

	\textbf{If}\ c\ $A_1$\ $A_2$
	&
	{\slshape
	{\isacharparenleft}{\isacharparenleft}\textbf{EVENT\isactrlsub A}\ {\isacharparenleft}{\isacharprime}{\isacharprime}If{\isacharprime}{\isacharprime}{\isacharparenright}\ \textbf{WHEN}\ {\isacharparenleft}{\isasymacute}{\isacharparenleft}{\isasymlambda}s{\isachardot}\ s{\isasymin}c{\isacharparenright}{\isacharparenright}\ \textbf{THEN}\ \textit{SKIP}\ \textbf{END}{\isacharparenright}
	
	\quad $\rhd$\ {\isacharparenleft}compile\ A{\isadigit{1}}{\isacharparenright}{\isacharparenright}

\ $\oplus$

{\isacharparenleft}{\isacharparenleft}\textbf{EVENT\isactrlsub A}\ {\isacharparenleft}{\isacharprime}{\isacharprime}Else{\isacharprime}{\isacharprime}{\isacharparenright}\ \textbf{WHEN}\ {\isacharparenleft}{\isasymacute}{\isacharparenleft}{\isasymlambda}s{\isachardot}\ s{\isasymnotin}c{\isacharparenright}{\isacharparenright}\ \textbf{THEN}\ \textit{SKIP}\ \textbf{END}{\isacharparenright}

	\quad $\rhd$\ {\isacharparenleft}compile\ A{\isadigit{2}}{\isacharparenright}{\isacharparenright}	
	}
\\ \hline
	\textbf{While} \ c\ A
   & 
    $\evtrec{c}{(compile\ A)}$
\\ \hline

	\textbf{Pick} \ $A_1$\ $A_2$
   & 
    $\evtchoice{compile\ A_1}{compile\ A_2}$
\\ \hline 

	\textbf{Flow} \ $A_1$\ $A_2$
   & 
    $\evtjoin{compile\ A_1}{compile\ A_2}$
\\ \hline 

\textbf{ActTerminator}
   & 
    $\fin$
\\ \hline 

	\textbf{OnMessage}\ ptl\ ptt\ opn\ spc\ at
   & 
    
    {\slshape
    {\isacharparenleft}\textbf{EVENT\isactrlsub A}\ {\isacharparenleft}{\isacharprime}{\isacharprime}OnMessage{\isacharprime}{\isacharprime}{\isacharat}ptl{\isacharat}ptt{\isacharat}opn{\isacharparenright}
    
    \textbf{WHEN}\ True\ \textbf{THEN}\ {\isacharparenleft}Basic\ spc{\isacharparenright}\ \textbf{END}{\isacharparenright}
    
    $\rhd$ \ {\isacharparenleft}compile\ at {\isacharparenright}
    }
\\ \hline 

	\textbf{OnAlarm}\ t\ at
   & 
   
   {\slshape
    {\isacharparenleft}\textbf{EVENT\isactrlsub A}\ {\isacharparenleft}{\isacharprime}{\isacharprime}OnAlarm{\isacharprime}{\isacharprime}{\isacharparenright}\ \textbf{WHEN}\ t\ >\ {\isasymacute}tick\ \textbf{THEN}\ SKIP\ \textbf{END}{\isacharparenright}
    
    $\rhd$ \ {\isacharparenleft}compile\ at{\isacharparenright}  
   }
\\ \hline

  \end{tabularx}
  \label{tbl:translate_basic}
\end{table}

The translated {\slang} specification of \textbf{receive}, \textbf{reply}, \textbf{assign} and \textbf{empty} activities are similar to the \textbf{invoke}'s. 
A \textbf{wait} activity has an additional guard condition $t < tick$, i.e. the time point $t$ has passed. 
The termination \textbf{ActFin} is translated to $\fin$. 

An \textbf{if} activity is translated to the choice composition of two event systems, each of which has an atomic event with disjoint guard conditions. The translation of \textbf{seq}, \textbf{while}, \textbf{pick} and \textbf{flow} activities is straightforward, as well as event handlers \textbf{OnMessage} and \textbf{OnAlarm}.  

Finally, due to different event names and bodies of translated events for BPEL activities, the translation we design is an injection shown as the following lemma. 

\begin{lemma}[\textbf{Translation} is an Injection]
\label{lm:inj_compile}
if $\textbf{compile}(b_1) = \textbf{compile}(b_2)$ then $b_1 = b_2$.
\end{lemma}

\subsection{Correctness Proof of the Translation}

\newcommand{\traceeq}[2]{{#1} \approxeq {#2}}
\newcommand{\picorebpeleq}[2]{\symbenv \vdash {#1} \approxeq {#2}}

To show the correctness of the translation, we first give the intuition behind the semantic equivalence between a BPEL program and a {\slang} Specification. 


A computation $c_{\symbevtsys}$ of an event system $\symbevtsys$ is equivalent to a computation $c_B$ of a BPEL program $B$, denoted as $\traceeq{c_B}{c_{\symbevtsys}}$, that is defined as follows. 
\[
\left\{
\begin{aligned}
& \traceeq{Nil}{Nil} = True \\
& \traceeq{(B,s)\# tr1}{(\symbevtsys,t)\# tr2} = (compile(B) = \symbevtsys \wedge s = t \wedge \traceeq{tr1}{tr2}) \\
& \mathbf{otherwise} ... = False
\end{aligned}
\right.
\]

Computations of BPEL programs are defined similarly to the linear definition in Subsection \ref{subsect:comp}, thus we use $\compfun(B,s)$ to denote the set of computations of a BPEL program $B$ from an initial state $s$. Then the equivalence of a {\slang} specification and a BPEL program is defined by the equivalence of their state traces as follows.

\begin{definition}[Equivalence of BPEL and {\slang}]
\label{def:equiv_bpel_picore}
A {\slang} specification $\symbevtsys$ and a BPEL program $B$ is equivalent from an initial state $\symbstate$ under a configuration $\symbenv$, denoted as $\picorebpeleq{(B,s)}{(\symbevtsys,t)}$, if
\begin{enumerate}
\item $\forall tr \in \compfun(B,s).\ \exists tr' \in \tilde{\compfun}(\symbenv,\symbevtsys,t).\ \traceeq{tr}{tr'}$.
\item $\forall tr' \in \tilde{\compfun}(\symbenv,\symbevtsys,t).\ \exists tr \in \compfun(B,s).\ \traceeq{tr}{tr'}$. 
\end{enumerate}
\end{definition}


The equivalence above is defined directly over state traces of programs, which is intuitive but not compositional. 
To decompose the correct proof of the translation into each execution step of BPEL activities, we use a bisimulation relation considering each step of in a BPEL and {\slang} trace to represent the equivalence of a BPEL program and its translated {\slang} specification. Moreover, the bisimulation we defined is sound and complete w.r.t. the equivalence of state traces. 

The bisimulation relation is coinductively defined as follows. Since BPEL and {\slang} have the same state definition, the state relation in the bisimulation is the equivalent relation ``=''. 

\begin{definition}[Bisimulation of BPEL and {\slang}]
\label{def:bisim}
A BPEL program $B$ and a {\slang} specification $\symbevtsys$ are bisimilar from the same starting state $s$ under a configuration $\symbenv$, denoted as $\bisim{B}{s}{\symbevtsys}{s}$, when the following are true.
\begin{enumerate}
\item for arbitrary $B'$ and $t$, if $\bpeltran{B}{s}{B'}{t}$, then there exists $\symbevtsys'$ such that $\exists \etranlabel.\ \estrannx{\symbevtsys}{s}{\etranlabel}{\symbevtsys'}{t}$ and $\bisim{B'}{t}{\symbevtsys'}{t}$. 
\item for arbitrary $\symbevtsys'$ and $t$, if $\exists \etranlabel.\ \estrannx{\symbevtsys}{s}{\etranlabel}{\symbevtsys'}{t}$, then there exists $B'$ such at $\bpeltran{B}{s}{B'}{t}$ and $\bisim{B'}{t}{\symbevtsys'}{t}$. 
\item $\symbevtsys = \textbf{compile}(B)$.
\end{enumerate}
\end{definition}

The bisimulation is sound and complete w.r.t. the state trace equivalence between BPEL and {\slang}. The soundness and completeness of the bisimulation are shown as the following theorem, where {\lemmaprefix} \ref{lm:inj_compile} is a necessary condition for the completeness proof. 
\begin{theorem}[Soundness and Completeness of Bisimulation]
$\picorebpeleq{(B,s)}{(\symbevtsys,s)}$ iff $\bisim{B}{s}{\symbevtsys}{s}$. 
\end{theorem}
\begin{proof}
The following two cases prove the theorem: 
\begin{enumerate}
\item \emph{soundness}: it is sufficient to show that if $\bisim{B}{s}{\symbevtsys}{s}$, then 

\begin{enumerate}
\item for an arbitrary computation $tr \in \compfun(B,s)$, there exist $\symbcomp \in \tilde{\compfun}(\symbenv,\symbevtsys,s)$ such that $\traceeq{tr}{\symbcomp}$. 

\item for an arbitrary computation $\symbcomp \in \tilde{\compfun}(\symbenv,\symbevtsys,s)$, there exist $tr \in \compfun(B,s)$ such that $\traceeq{tr}{\symbcomp}$.
\end{enumerate}
The two propositions could be proved by induction of the list $tr$ and $\symbcomp$ respectively. 

\item \emph{completeness}: by coinduction of the bisimulation, we only need to prove the following three propositions under the assumption $\picorebpeleq{(B,s)}{(\symbevtsys,s)}$:

\begin{enumerate}
\item if $\bpeltran{B}{s}{B'}{s'}$, there exists $\symbevtsys'$ such that $\exists \etranlabel.\ \estrannx{\symbevtsys}{s}{\etranlabel}{\symbevtsys'}{s'}$ and $\picorebpeleq{(B',s')}{(\symbevtsys',s')}$. We prove it by $\symbevtsys' = compile(B')$. 
\begin{enumerate}
\item from $\bpeltran{B}{s}{B'}{s'}$ we have that $[(B,s),(B',s')] \in \compfun(B,s)$. With the assumption we have $[(\symbevtsys,s),(\symbevtsys',s')] \in \tilde{\compfun}(\symbenv,\symbevtsys,s)$. Thus, $\exists \etranlabel.\ \estrannx{\symbevtsys}{s}{\etranlabel}{\symbevtsys'}{s'}$. 

\item to prove $\picorebpeleq{(B',s')}{(\symbevtsys',s')}$, we need to show that
\begin{enumerate}
\item \label{item:A}
for arbitrary $tr \in \compfun(B',s')$, there exists $\symbcomp \in \tilde{\compfun}(\symbenv,\symbevtsys',s')$ such that $\traceeq{tr}{\symbcomp}$. With $\bpeltran{B}{s}{B'}{s'}$, we have that $(B,s)\#tr \in \compfun(B,s)$. With the assumption $\picorebpeleq{(B,s)}{(\symbevtsys,s)}$, we obtain a computation $\symbcomp' \in \tilde{\compfun}(\symbenv,\symbevtsys,s)$ and $\traceeq{(B,s)\#tr}{\symbcomp'}$. Then it is proved by $\symbcomp = tl\ \symbcomp'$, where $\symbcomp$ is the tail of the list $\symbcomp'$. 

\item for arbitrary $\symbcomp \in \tilde{\compfun}(\symbenv,\symbevtsys',s')$, there exists $tr \in \compfun(B',s')$ such that $\traceeq{tr}{\symbcomp}$. It could be proved by a reversed way using the injection of the translation ({\lemmaprefix} \ref{lm:inj_compile}). 
\end{enumerate}

\end{enumerate}

\item if $\exists \etranlabel.\ \estrannx{\symbevtsys}{s}{\etranlabel}{\symbevtsys'}{s'}$, there exist $B'$ such that $\bpeltran{B}{s}{B'}{s'}$ and $\picorebpeleq{(B',s')}{(\symbevtsys',s')}$. It could be similarly proved as above by $\symbevtsys' = compile(B')$. 

\item $\symbevtsys = compile(B)$. Its straightforward from the assumption and {\defprefix} \ref{def:equiv_bpel_picore}. 
\end{enumerate}
\end{enumerate}
\end{proof}

Finally, we have proved the theorem of correct translation from BPEL to {\slang} as follows. 
\begin{theorem}[Correctness of Translation]
$\forall B\ s.\ \bisim{B}{s}{compile(B)}{s}$
\end{theorem}
\begin{proof}
By coinduction of the bisimulation, it is sufficient to prove the following propositions: 
\begin{enumerate}
\item if $\bpeltran{B}{s}{B'}{s'}$ then $\exists \etranlabel.\ \estrannx{compile(B)}{s}{\etranlabel}{compile(B')}{s'}$. It could be proved by induction of $B$ and then the preservation on each type of BPEL activity and its compiled {\slang} specification, among which the proof of the \textbf{invoke} activity is the most complex. It is proved by the cases of transitions rules of \textbf{invoke}: 
\begin{enumerate}
\item \bpelsemrulename{InvokeSuc} in {\figprefix} \ref{fig:bpel_semantics}: the \textbf{Event\isactrlsub {succ}} in $\textbf{compile}(\textbf{Invoke} \ fls\ ptl\ ptt\ opn\ spc\ cts\ cta)$ in {\figprefix} \ref{fig:compile_invoke} simulates the successful invocation. Furthermore, we have $$\exists \etranlabel.\ \estrannx{\textbf{Event\isactrlsub {succ}}}{s}{\etranlabel}{\fin}{s'}$$ and $$\textbf{compile}(\bpelfin) = \fin$$

\item \bpelsemrulename{InvokeFault} in {\figprefix} \ref{fig:bpel_semantics}: we further split the case to two cases, i.e. the fault is caught by the default handler and by special handlers. The first case is simulated by $\evtseq{ \textbf{Event\isactrlsub {fail}} }{ \textbf{compile}(cta) }$ in $\textbf{compile}(\textbf{Invoke} \ fls\ ptl\ ptt\ opn\ spc\ cts\ cta)$ in {\figprefix} \ref{fig:compile_invoke}. We have $$\bpeltran{\textbf{Invoke}\ fls\ ptl\ ptt\ opn\ spc\ cts\ cta}{s}{cta}{s'}$$ and $$\exists \etranlabel.\ \estrannx{\evtseq{ \textbf{Event\isactrlsub {fail}} }{ \textbf{compile}(cta) }}{s}{\etranlabel}{\textbf{compile}(cta)}{s'}$$
A similar way proves the second case. 

\end{enumerate}

\item if $\estrannx{compile(B)}{s}{\etranlabel}{Q}{s'}$, there exist $B'$ such that $\bpeltran{B}{s}{B'}{s'}$ and $Q = compile(B')$. It could be proved in a similar approach as the first case. 
\end{enumerate}

\end{proof}

\section{Evaluation and Discussion}
\label{sect:eval}

\subsection{Evaluation}

We use Isabelle/HOL as the specification and verification system. All derivations of our proofs have passed through the Isabelle proof kernel.
{\tableprefix} \ref{tbl:stat} shows the statistics for the effort and size of the proofs in the Isabelle/HOL theorem prover. In total, the models and mechanized verification consists of $\approx$ 28,400 lines of specification and proofs (\emph{LOSP}), and the total effort is $\approx$ 24 person-months (\emph{PM}). The specification and proof of {\slang} are reusable for the verification of other systems. 

We use $\approx$ 7,400 lines of specification and proof to develop the {\slang} framework, which takes 8 \emph{PM}s. 
The {\implang} language and its rely-guarantee proof system consist of $\approx$ 2,400 \emph{LOSP}, and {\csimpllang} $\approx$ 15,000 \emph{LOSP}. The two parts of specification and proof are completely reused in {\pcimp} and {\pccsimpl} respectively. 
The adapter of {\implang} is $\approx$ 660 \emph{LOSP} including new proof rules and their soundness as well as a concrete syntax. The adapter of {\csimpllang} is $\approx$ 470 \emph{LOSP}. The two {\slang} instances takes about 2 \emph{PM}s. 

We develop $\approx$ 17,600 \emph{LOSP} for the Zephyr case study, 40 times more than the lines of the C code due to the in-kernel concurrency, where invariant proofs represent the largest part. The case study takes 12 \emph{PM}s. The BPEL case study takes 2,200 LOSP and 2 \emph{PM}s.

\begin{table}[t]
\centering
\footnotesize
\caption{Specification and Proof Statistics} 
\begin{tabular} {|c|r|c||c|r|c|}
\hline
\multicolumn{3}{|c||}{\textbf{{\slang} Framework}} & \multicolumn{3}{c|}{\textbf{{\csimpllang} and {\implang} Adapters}} \\
\hline
\textbf{Item} 					  & \textbf{LOSP}  & \textbf{PM} & \textbf{Item} & \textbf{LOSP} & \textbf{PM} \\
\hline
\textit{Language and Semantics}   & 300 	& \multirow{6}{*}{8} & \textit{{\implang} extension} & 530 & \multirow{6}{*}{2} \\
\cline{1-2} \cline{4-5} 
\textit{Computation} 			  & 2,270 	& & \textit{{\implang} Adapter} & 130 & \\
\cline{1-2} \cline{4-5} 
\textit{Validity and Proof Rules} & 280 	& & \textit{{\csimpllang} Adapter} & 470 & \\
\cline{1-2} 
\textit{Soundness}                & 4,200 	& &							    &  & \\
\cline{1-2} 
\textit{Invariant/other lemmas}   & 380 	& & 							&  & \\
\cline{1-2} \cline{4-5} 
\textbf{Total}                    & 7,430 	& & \textbf{Total} 				& 1,130 & \\
\hline \hline

\multicolumn{3}{|c||}{\textbf{Zephyr Memory Case}} & \multicolumn{3}{c|}{\textbf{BPEL Case}} \\
\hline
\textbf{Item} 					  	& \textbf{LOSP}  & \textbf{PM} & \textbf{Item} & \textbf{LOSP} & \textbf{PM} \\
\hline
\textit{Specification} 				& 400 	& \multirow{5}{*}{12} & Syntax 				& 120 & \multirow{5}{*}{2}\\
\cline{1-2} \cline{4-5} 
\textit{Auxiliary Lemmas/Invariant} & 1,700 	& 					  & Semantics			& 160 & \\
\cline{1-2} \cline{4-5} 
\textit{Proof of Allocation} 		& 10,600 & 					  & Translation			& 120 & \\
\cline{1-2} \cline{4-5} 
\textit{Proof of Free} 				& 4,950 	&					  & Correctness Proof	& {1,800} & \\
\cline{1-2} \cline{4-5} 
\textbf{Total} 						& 17,650 &					  & \textbf{Total} 		& {2,200} & \\
\hline
\end{tabular}
\label{tbl:stat}
\end{table} 

\subsection{Further Related Work}

\paragraph{Rely-guarantee Approach and Mechanization}
The rely-guarantee approach has been mechanized in Isabelle/HOL (e.g. \cite{Nieto03,Stephan15,Hayes16FM,Hayes16FAC,Sanan17}) and Coq (e.g. \cite{LiangFF12,Moreira13}). 
In \cite{Hayes16FM,Hayes16FAC}, an abstract algebra of atomic steps is developed, and rely/guarantee concurrency is an interpretation of the algebra. 
To allow a meaningful comparison of rely-guarantee semantic models, two abstract models for rely-guarantee are developed and mechanized in \cite{Stephan15}. The two works do not consider the concrete imperative languages for rely-guarantee. The works \cite{Nieto03,Moreira13} mechanize the rely-guarantee approach for simple imperative languages. Later, a rely-guarantee proof system for {\csimpllang} \cite{Sanan17}, a generic and realistic imperative language by extending \emph{Simpl}, is developed in Isabelle/HOL. These mechanizations focus on imperative languages for pure programs. Two of them \cite{Nieto03,Sanan17} are mechanized in Isabelle/HOL and they have been integrated into {\slang}. 


\paragraph{Rely-guarantee Reasoning about Event-based Systems}
Refinement of reactive systems \cite{Back96} and the subsequent Event-B approach \cite{Abrial07} propose a refinement-based formal method for system-level modeling and analysis. 
Different from Event-B where the execution of events are atomic, we consider the interleaved semantics of events and relax the atomicity of them in {\slang}. 
In \cite{Hoang10}, an Event-B model is created to mimic rely-guarantee style reasoning for concurrent programs, but not to provide a rely-guarantee framework for Event-B. 
The rely-guarantee reasoning for event-based applications has been studied in \cite{Dingel98fac,Garlan98fse,Fenkam03fme,Fenkam03fase}. 
The definition of events is similar to {\slang}. They extend a simple, sequential, imperative language by primitives for announcing and consuming events, \emph{announce(e)} and \emph{consume(e(x))} where \emph{e} is an event. Therefore, events are triggered by imperative programs in another event. This is very different from the reactive semantics in {\slang}, which supports complex reaction structures to simulate real reactive systems. Moreover, the language to specify events in these works is a simple imperative language, whilst {\slang} has an open interface for the integration and reusability of different languages and frameworks. 


\paragraph{Formal Verification of Memory Management}
Memory models \cite{Saraswat07} provide the necessary abstraction to separate the behaviour of a program from the behaviour of the memory it reads and writes. There are several formalizations of memory models in the literature \cite{Leroy2008,Tews2009,Gallardo2009,Sevcik13,Mansky15}, where some of them only create an abstract specification of the services for memory allocation and release \cite{Gallardo2009,Sevcik13,Mansky15}. 
Formal verification of OS memory management has been studied in CertiKOS\cite{Vaynberg12,Gu16}, seL4 \cite{Klein04,Klein09}, Verisoft \cite{Alkassar08}, and in the hypervisors from~\cite{Blan15,Bolig16}, where only the works in~\cite{Gu16,Blan15} consider concurrency. Comparing to buddy memory allocation, the data structures and algorithms verified in \cite{Gu16} are relatively simpler, without block split/coalescence and multiple levels of free lists and bitmaps. \cite{Blan15} only considers virtual mapping but not allocation or deallocation of memory areas. 
Algorithms and implementations of dynamic memory allocation have been formally specified and verified in an extensive number of works \cite{Yu03,Fang17a,Marti06,Su16,Fang17b,Fang18}. However, the buddy memory allocation is only studied in \cite{Fang18}, which does not consider concrete data structures (e.g. bitmaps) and concurrency. The Zephyr case study in this article presents the first formal specification and mechanized proof for a concurrent buddy memory allocation of a realistic operating system. 

\paragraph{Formal Verification of BPEL}
In recent years, how to increase the dependability of complex Web services composition has become a key problem. Thus numerous approaches have been proposed to formally specify and verify service compositions \cite{Campos14}. 
Many works attempt to formally verify BPEL programs by manually modeling BPEL using verification tools (e.g. \cite{Fu04,Babin17,Stach18}) and by model transformation (e.g.\cite{Zhu17}), however without proof of correctness of the modeling and transformation. 
The rely-guarantee approach has been applied to BPEL in \cite{Zhuhb12}. The authors proposed the rely-guarantee proof rules for each BPEL activity. 
Compare to these works, this article presents the first verified translation from BPEL to {\slang}, and thus conduct the formal verification of BPEL by reusing the {\slang} proof systems.



\section{Conclusion and Future Work}
\label{sect:conclusion}

In this article, we propose an event-based rely-guarantee framework for concurrent reactive systems. This framework is open to the specification of event behaviours. It provides a rely-guarantee interface to integrate systems for specification and reasoning at that level that eases formal methods reusability. We have mechanized the integration of the {\implang} and {\csimpllang} languages and their proof systems into {\slang} in the Isabelle/HOL theorem prover. We show the simplicity of events to represent concurrent reactive systems and the usefullness of {\slang} for realistic systems in the verification of the concurrent buddy memory allocation of Zephyr RTOS. We also show the expressiveness of reaction structures in {\slang} to model complex business processes in the BPEL case study. 

Currently we are working on using {\slang} on the formal verification of noninterference properties on concurrent reactive systems. As future work, we plan to extend {\slang} to support step-wise refinement. More domain-specific reasoning frameworks, such as interrupts of multicore platforms and messaging systems for autonomous vehicles, are under development using {\slang} and its instances.

%


%
\bibliographystyle{ACM-Reference-Format}
\bibliography{paperref}

%
\appendix

\section{C Code of \emph{k\_mem\_pool\_free}}
\label{appx:mempoolfree_c}

\lstdefinestyle{customc}{
  belowcaptionskip=1\baselineskip,
  breaklines=true,
  frame=single, 
  xleftmargin=0pt, 
  language=C,
  numbers=left,
  stepnumber=1,
  showstringspaces=false,
  basicstyle=\scriptsize\ttfamily, 
  keywordstyle=\bfseries\color{green!40!black},
  commentstyle=\itshape\color{purple!40!black},
  identifierstyle=\color{blue},
  stringstyle=\color{orange},
}
\lstset{escapechar=@,style=customc}
\begin{lstlisting}
static void free_block(struct k_mem_pool *p, int level, size_t *lsizes, int bn)
{
  int i, key, lsz = lsizes[level];
  void *block = block_ptr(p, lsz, bn);

  key = irq_lock();

  set_free_bit(p, level, bn);

  if (level && partner_bits(p, level, bn) == 0xf) {
    for (i = 0; i < 4; i++) {
      int b = (bn & ~3) + i;

      clear_free_bit(p, level, b);
      if (b != bn && block_fits(p, block_ptr(p, lsz, b), lsz)) {
        sys_dlist_remove(block_ptr(p, lsz, b));
      }
    }

    irq_unlock(key);
    free_block(p, level-1, lsizes, bn / 4); /* tail recursion! */
    return;
  }

  if (block_fits(p, block, lsz)) {
    sys_dlist_append(&p->levels[level].free_list, block);
  }

  irq_unlock(key);
}

void k_mem_pool_free(struct k_mem_block *block)
{
  int i, key, need_sched = 0;
  struct k_mem_pool *p = get_pool(block->id.pool);
  size_t lsizes[p->n_levels];

  /* As in k_mem_pool_alloc(), we build a table of level sizes
  * to avoid having to store it in precious RAM bytes.
  * Overhead here is somewhat higher because free_block()
  * doesn't inherently need to traverse all the larger
  * sublevels.
  */
  lsizes[0] = _ALIGN4(p->max_sz);
  for (i = 1; i <= block->id.level; i++) {
    lsizes[i] = _ALIGN4(lsizes[i-1] / 4);
  }

  free_block(get_pool(block->id.pool), block->id.level, lsizes, block->id.block);

  /* Wake up anyone blocked on this pool and let them repeat
   * their allocation attempts
   */
  key = irq_lock();

  while (!sys_dlist_is_empty(&p->wait_q)) {
    struct k_thread *th = (void *)sys_dlist_peek_head(&p->wait_q);

    _unpend_thread(th);
    _abort_thread_timeout(th);
    _ready_thread(th);
    need_sched = 1;
  }

  if (need_sched && !_is_in_isr()) {
    _reschedule_threads(key);
  } else {
    irq_unlock(key);
  }
}
\end{lstlisting}

\section{Specification and Proof Sketch of \emph{k\_mem\_pool\_free}}
\label{appx:mempoolfree}

The formal specification of \emph{k\_mem\_pool\_free} (in \emph{black} color) and its rely-guarantee proof sketch (in \emph{blue} color) are shown as follows. 

\vspace{5mm}
\isabellestyle{sl} 
\begin{isabellec} \fontsize{8pt}{0cm} 
\specrg{\isacommand{Mem{\isacharunderscore}pool{\isacharunderscore}free{\isacharunderscore}pre}\ t\ {\isasymequiv}\ {\isasymlbrace} {\isasymacute}inv\ {\isasymand}\ {\isasymacute}allocating{\isacharunderscore}node\ t\ {\isacharequal}\ None\ {\isasymand}\ {\isasymacute}freeing{\isacharunderscore}node\ t\ {\isacharequal}\ None{\isasymrbrace}}

\stmtevthead{mem\_pool\_free}{Block\ b}{($\mathcal{T}$\ t)}

\isacommand{WHEN}\isanewline
\quad pool\ b\ {\isasymin}\ {\isasymacute}mem{\isacharunderscore}pools\ \isanewline
\quad {\isasymand}\ level\ b\ {\isacharless}\ length\ {\isacharparenleft}levels\ {\isacharparenleft}{\isasymacute}mem{\isacharunderscore}pool{\isacharunderscore}info\ {\isacharparenleft}pool\ b{\isacharparenright}{\isacharparenright}{\isacharparenright}\isanewline
\quad {\isasymand}\ block\ b\ {\isacharless}\ length\ {\isacharparenleft}bits\ {\isacharparenleft}levels\ {\isacharparenleft}{\isasymacute}mem{\isacharunderscore}pool{\isacharunderscore}info\ {\isacharparenleft}pool\ b{\isacharparenright}{\isacharparenright}{\isacharbang}{\isacharparenleft}level\ b{\isacharparenright}{\isacharparenright}{\isacharparenright}\isanewline
\quad {\isasymand}\ data\ b\ {\isacharequal}\ block{\isacharunderscore}ptr\ {\isacharparenleft}{\isasymacute}mem{\isacharunderscore}pool{\isacharunderscore}info\ {\isacharparenleft}pool\ b{\isacharparenright}{\isacharparenright}\isanewline
\quad \quad \quad \quad \quad  {\isacharparenleft}{\isacharparenleft}ALIGN{\isadigit{4}}\ {\isacharparenleft}max{\isacharunderscore}sz\ {\isacharparenleft}{\isasymacute}mem{\isacharunderscore}pool{\isacharunderscore}info\ {\isacharparenleft}pool\ b{\isacharparenright}{\isacharparenright}{\isacharparenright}{\isacharparenright}\ div\ {\isacharparenleft}{\isadigit{4}}\ {\isacharcircum}\ {\isacharparenleft}level\ b{\isacharparenright}{\isacharparenright}{\isacharparenright}\ {\isacharparenleft}block\ b{\isacharparenright}\isanewline
\isacommand{THEN}

\quad \specrg{Mem{\isacharunderscore}pool{\isacharunderscore}free{\isacharunderscore}pre\ t\ {\isasyminter}\ {\isasymlbrace} g {\isasymrbrace}}
\quad \speccomment{(* g is the guard condition of the event *)}

\quad \speccomment{{\isacharparenleft}{\isacharasterisk}\ here\ we\ set\ the\ bit\ to\ FREEING{\isacharcomma}\ so\ that\ other\ thread\ cannot\ mem{\isacharunderscore}pool{\isacharunderscore}free\ the\ same\ block\ \isanewline
\ \ \ \ \ \ \ it\ also\ requires\ that\ it\ can\ only\ free\ ALLOCATED\ block\ {\isacharasterisk}{\isacharparenright}}

\quad t\ \isactrlenum \ \isacommand{AWAIT}\ {\isacharparenleft}bits\ {\isacharparenleft}{\isacharparenleft}levels\ {\isacharparenleft}{\isasymacute}mem{\isacharunderscore}pool{\isacharunderscore}info\ {\isacharparenleft}pool\ b{\isacharparenright}{\isacharparenright}{\isacharparenright}\ {\isacharbang}\ {\isacharparenleft}level\ b{\isacharparenright}{\isacharparenright}{\isacharparenright}\ {\isacharbang}\ {\isacharparenleft}block\ b

\quad \quad \quad \quad \quad \quad \quad \quad \quad {\isacharequal}\ ALLOCATED\ \isacommand{THEN}

\quad \quad \quad \quad {\isasymacute}mem{\isacharunderscore}pool{\isacharunderscore}info\ {\isacharcolon}{\isacharequal}\ set{\isacharunderscore}bit{\isacharunderscore}freeing\ {\isasymacute}mem{\isacharunderscore}pool{\isacharunderscore}info\ {\isacharparenleft}pool\ b{\isacharparenright}\ {\isacharparenleft}level\ b{\isacharparenright}\ {\isacharparenleft}block\ b{\isacharparenright}{\isacharsemicolon}{\isacharsemicolon}

\quad \quad \quad \quad {\isasymacute}freeing{\isacharunderscore}node\ {\isacharcolon}{\isacharequal}\ {\isasymacute}freeing{\isacharunderscore}node\ {\isacharparenleft}t\ {\isacharcolon}{\isacharequal}\ Some\ b{\isacharparenright}

\quad \quad \quad \isacommand{END}{\isacharparenright}{\isacharsemicolon}{\isacharsemicolon}

\quad \specrg{\isacommand{mp{\isacharunderscore}free{\isacharunderscore}precond{\isadigit{2}}}\ t\ b\ {\isasymequiv}\ {\isasymlbrace} {\isasymacute}inv\ {\isasymand}\ {\isasymacute}allocating{\isacharunderscore}node\ t\ {\isacharequal}\ None\ {\isasymand} \ g\ {\isasymand}\ {\isasymacute}freeing{\isacharunderscore}node\ t\ {\isacharequal}\ Some\ b{\isasymrbrace}}

\quad t\ \isactrlenum \ {\isasymacute}need{\isacharunderscore}resched\ {\isacharcolon}{\isacharequal}\ {\isasymacute}need{\isacharunderscore}resched{\isacharparenleft}t\ {\isacharcolon}{\isacharequal}\ False{\isacharparenright}{\isacharsemicolon}{\isacharsemicolon}

\quad \specrg{\isacommand{mp{\isacharunderscore}free{\isacharunderscore}precond{\isadigit{3}}}\ t\ b\ {\isasymequiv}\ {\isacharparenleft}mp{\isacharunderscore}free{\isacharunderscore}precond{\isadigit{2}}\ t\ b{\isacharparenright}\ {\isasyminter}\ {\isasymlbrace}{\isasymacute}need{\isacharunderscore}resched\ t\ {\isacharequal}\ False{\isasymrbrace}}

\quad t\ \isactrlenum \ {\isasymacute}lsizes\ {\isacharcolon}{\isacharequal}\ {\isasymacute}lsizes{\isacharparenleft}t\ {\isacharcolon}{\isacharequal}\ {\isacharbrackleft}ALIGN{\isadigit{4}}\ {\isacharparenleft}max{\isacharunderscore}sz\ {\isacharparenleft}{\isasymacute}mem{\isacharunderscore}pool{\isacharunderscore}info\ {\isacharparenleft}pool\ b{\isacharparenright}{\isacharparenright}{\isacharparenright}{\isacharbrackright}{\isacharparenright}{\isacharsemicolon}{\isacharsemicolon}

\quad \specrg{\isacommand{mp{\isacharunderscore}free{\isacharunderscore}precond{\isadigit{4}}}\ t\ b\ {\isasymequiv}\ \isanewline
\quad \quad mp{\isacharunderscore}free{\isacharunderscore}precond{\isadigit{3}}\ t\ b\ {\isasyminter}\ {\isasymlbrace}{\isasymacute}lsizes\ t\ {\isacharequal}\ {\isacharbrackleft}ALIGN{\isadigit{4}}\ {\isacharparenleft}max{\isacharunderscore}sz\ {\isacharparenleft}{\isasymacute}mem{\isacharunderscore}pool{\isacharunderscore}info\ {\isacharparenleft}pool\ b{\isacharparenright}{\isacharparenright}{\isacharparenright}{\isacharbrackright}{\isasymrbrace}}

\quad \isacommand{FOR}\ {\isacharparenleft}t\ \isactrlenum \ {\isasymacute}i\ {\isacharcolon}{\isacharequal}\ {\isasymacute}i{\isacharparenleft}t\ {\isacharcolon}{\isacharequal}\ {\isadigit{1}}{\isacharparenright}{\isacharparenright}{\isacharsemicolon}\  {\isasymacute}i\ t\ {\isasymle}\ level\ b{\isacharsemicolon}\ \ {\isacharparenleft}t\ \isactrlenum \ {\isasymacute}i\ {\isacharcolon}{\isacharequal}\ {\isasymacute}i{\isacharparenleft}t\ {\isacharcolon}{\isacharequal}\ {\isasymacute}i\ t\ {\isacharplus}\ {\isadigit{1}}{\isacharparenright}{\isacharparenright}\ \isacommand{DO}\isanewline
\quad \quad t\ \isactrlenum \ {\isasymacute}lsizes\ {\isacharcolon}{\isacharequal}\ {\isasymacute}lsizes{\isacharparenleft}t\ {\isacharcolon}{\isacharequal}\ {\isasymacute}lsizes\ t\ {\isacharat}\ {\isacharbrackleft}ALIGN{\isadigit{4}}\ {\isacharparenleft}{\isasymacute}lsizes\ t\ {\isacharbang}\ {\isacharparenleft}{\isasymacute}i\ t\ {\isacharminus}\ {\isadigit{1}}{\isacharparenright}\ div\ {\isadigit{4}}{\isacharparenright}{\isacharbrackright}{\isacharparenright}

\quad \isacommand{ROF}{\isacharsemicolon}{\isacharsemicolon}

\quad \specrg{\isacommand{mp{\isacharunderscore}free{\isacharunderscore}precond{\isadigit{5}}}\ t\ b\ {\isasymequiv}\ mp{\isacharunderscore}free{\isacharunderscore}precond{\isadigit{3}}\ t\ b \ {\isasyminter} \isanewline
\quad \quad {\isasymlbrace}{\isacharparenleft}{\isasymforall}ii{\isacharless}length\ {\isacharparenleft}{\isasymacute}lsizes\ t{\isacharparenright}{\isachardot}\ {\isasymacute}lsizes\ t\ {\isacharbang}\ ii\ {\isacharequal}\ {\isacharparenleft}ALIGN{\isadigit{4}}\ {\isacharparenleft}max{\isacharunderscore}sz\ {\isacharparenleft}{\isasymacute}mem{\isacharunderscore}pool{\isacharunderscore}info\ {\isacharparenleft}pool\ b{\isacharparenright}{\isacharparenright}{\isacharparenright}{\isacharparenright}\isanewline
\quad \quad \quad \quad \quad \quad \quad \quad \quad \quad \quad \quad \quad \quad \quad \quad \quad \quad \quad \quad 
div\ {\isacharparenleft}{\isadigit{4}}\ {\isacharcircum}\ ii{\isacharparenright}{\isacharparenright}\ {\isasymand}\ length\ {\isacharparenleft}{\isasymacute}lsizes\ t{\isacharparenright}\ {\isachargreater}\ level\ b{\isasymrbrace}}

\quad \speccomment{{\isacharparenleft}{\isacharasterisk}\ {\isacharequal} {\isacharequal} {\isacharequal} start{\isacharcolon}\ free{\isacharunderscore}block{\isacharparenleft}pool{\isacharcomma}\ level{\isacharcomma}\ lsizes{\isacharcomma}\ block{\isacharparenright}{\isacharsemicolon}\ {\isacharequal} {\isacharequal} {\isacharequal}{\isacharasterisk}{\isacharparenright}}

\quad t\ \isactrlenum \ {\isasymacute}free{\isacharunderscore}block{\isacharunderscore}r\ {\isacharcolon}{\isacharequal}\ {\isasymacute}free{\isacharunderscore}block{\isacharunderscore}r\ {\isacharparenleft}t\ {\isacharcolon}{\isacharequal}\ True{\isacharparenright}{\isacharsemicolon}{\isacharsemicolon}

\quad \specrg{\isacommand{mp{\isacharunderscore}free{\isacharunderscore}precond{\isadigit{6}}}\ t\ b\ {\isasymequiv}\ mp{\isacharunderscore}free{\isacharunderscore}precond{\isadigit{5}}\ t\ b\ {\isasyminter}\ {\isasymlbrace}{\isasymacute}free{\isacharunderscore}block{\isacharunderscore}r\ t\ {\isacharequal}\ True{\isasymrbrace}}

\quad t\ \isactrlenum \ {\isasymacute}bn\ {\isacharcolon}{\isacharequal}\ {\isasymacute}bn\ {\isacharparenleft}t\ {\isacharcolon}{\isacharequal}\ block\ b{\isacharparenright}{\isacharsemicolon}{\isacharsemicolon}

\quad \specrg{\isacommand{mp{\isacharunderscore}free{\isacharunderscore}precond{\isadigit{7}}}\ t\ b\ {\isasymequiv}\ mp{\isacharunderscore}free{\isacharunderscore}precond{\isadigit{6}}\ t\ b\ {\isasyminter}\ {\isasymlbrace}{\isasymacute}bn\ t\ {\isacharequal}\ block\ b{\isasymrbrace}}

\quad t\ \isactrlenum \ {\isasymacute}lvl\ {\isacharcolon}{\isacharequal}\ {\isasymacute}lvl\ {\isacharparenleft}t\ {\isacharcolon}{\isacharequal}\ level\ b{\isacharparenright}{\isacharsemicolon}{\isacharsemicolon}

\quad \specrg{
\isacommand{mp{\isacharunderscore}free{\isacharunderscore}loopinv}\ t\ b\ {\isasymalpha}
}\isanewline
\quad \isacommand{WHILE}\ {\isasymacute}free{\isacharunderscore}block{\isacharunderscore}r\ t\ \isacommand{DO}\isanewline
\quad \specrg{
\isacommand{mp{\isacharunderscore}free{\isacharunderscore}cnd1}\ t\ b\ {\isasymalpha}\ {\isasymequiv} mp{\isacharunderscore}free{\isacharunderscore}loopinv\ t\ b\ {\isasymalpha} {\isasyminter} {\isasymlbrace} {\isasymalpha} {\isachargreater} 0 {\isasymrbrace}
}\isanewline
\quad \quad t\ {\isactrlenum} \ {\isasymacute}lsz\ {\isacharcolon}{\isacharequal}\ {\isasymacute}lsz\ {\isacharparenleft}t\ {\isacharcolon}{\isacharequal}\ {\isasymacute}lsizes\ t\ {\isacharbang}\ {\isacharparenleft}{\isasymacute}lvl\ t{\isacharparenright}{\isacharparenright}{\isacharsemicolon}{\isacharsemicolon}\isanewline
\quad \specrg{
\isacommand{mp{\isacharunderscore}free{\isacharunderscore}cnd2}\ t\ b\ {\isasymalpha}\ {\isasymequiv} mp{\isacharunderscore}free{\isacharunderscore}cnd1 \ t\ b\ {\isasymalpha} {\isasyminter} {\isasymlbrace} {\isasymacute}lsz\ t\ {\isacharequal}\ {\isasymacute}lsizes\ t\ {\isacharbang}\ {\isacharparenleft}{\isasymacute}lvl\ t{\isacharparenright} {\isasymrbrace}
}\isanewline
\quad \quad t\ {\isactrlenum} \ {\isasymacute}blk\ {\isacharcolon}{\isacharequal}\ {\isasymacute}blk\ {\isacharparenleft}t\ {\isacharcolon}{\isacharequal}\ block{\isacharunderscore}ptr\ {\isacharparenleft}{\isasymacute}mem{\isacharunderscore}pool{\isacharunderscore}info\ {\isacharparenleft}pool\ b{\isacharparenright}{\isacharparenright}\ {\isacharparenleft}{\isasymacute}lsz\ t{\isacharparenright}\ {\isacharparenleft}{\isasymacute}bn\ t{\isacharparenright}{\isacharparenright}{\isacharsemicolon}{\isacharsemicolon}\isanewline
\quad \specrg{
\isacommand{mp{\isacharunderscore}free{\isacharunderscore}cnd3}\ t\ b\ {\isasymalpha}\ {\isasymequiv} mp{\isacharunderscore}free{\isacharunderscore}cnd2 \ t\ b\ {\isasymalpha} {\isasyminter} \isanewline
\quad \quad \quad \quad \quad \quad \quad \quad \quad \quad {\isasymlbrace} {\isasymacute}blk\ t\ {\isacharequal}\ block{\isacharunderscore}ptr\ {\isacharparenleft}{\isasymacute}mem{\isacharunderscore}pool{\isacharunderscore}info\ {\isacharparenleft}pool\ b{\isacharparenright}{\isacharparenright}\ {\isacharparenleft}{\isasymacute}lsz\ t{\isacharparenright}\ {\isacharparenleft}{\isasymacute}bn\ t{\isacharparenright} {\isasymrbrace}
}\isanewline
\quad \quad t\ {\isactrlenum} \ \isacommand{ATOM}\isanewline
\quad \quad \specrg{\{V1\} \speccomment{{\isacharparenleft}V1 {\isasymin} mp{\isacharunderscore}free{\isacharunderscore}cnd3 \ t\ b\ {\isasymalpha} {\isasyminter} {\isasymlbrace}{\isasymacute}cur\ {\isacharequal}\ Some\ t{\isasymrbrace}{\isacharparenright}}}
\isanewline
\quad \quad \quad {\isasymacute}mem{\isacharunderscore}pool{\isacharunderscore}info\ {\isacharcolon}{\isacharequal}\ set{\isacharunderscore}bit{\isacharunderscore}free\ {\isasymacute}mem{\isacharunderscore}pool{\isacharunderscore}info\ {\isacharparenleft}pool\ b{\isacharparenright}\ {\isacharparenleft}{\isasymacute}lvl\ t{\isacharparenright}\ {\isacharparenleft}{\isasymacute}bn\ t{\isacharparenright}{\isacharsemicolon}{\isacharsemicolon}\isanewline
\quad \quad \specrg{\{V2\} \speccomment{{\isacharparenleft}V2 = V1{\isasymlparr}mem{\isacharunderscore}pool{\isacharunderscore}info\ {\isacharcolon}{\isacharequal}\isanewline
\quad \quad \quad \quad \quad \quad \quad \quad \quad set{\isacharunderscore}bit{\isacharunderscore}free\ {\isacharparenleft}mem{\isacharunderscore}pool{\isacharunderscore}info\ V1{\isacharparenright}\ {\isacharparenleft}pool\ b{\isacharparenright}\ {\isacharparenleft}lvl\ V1\ t{\isacharparenright}\ {\isacharparenleft}bn\ V1\ t{\isacharparenright}{\isasymrparr}{\isacharparenright}}}
\isanewline
\quad \quad \quad {\isasymacute}freeing{\isacharunderscore}node\ {\isacharcolon}{\isacharequal}\ {\isasymacute}freeing{\isacharunderscore}node\ {\isacharparenleft}t\ {\isacharcolon}{\isacharequal}\ None{\isacharparenright}{\isacharsemicolon}{\isacharsemicolon}\isanewline
\quad \quad \specrg{\{V3\} \speccomment{{\isacharparenleft}V3 = V2{\isasymlparr}freeing{\isacharunderscore}node\ {\isacharcolon}{\isacharequal}\ {\isacharparenleft}freeing{\isacharunderscore}node\ V2{\isacharparenright}{\isacharparenleft}t\ {\isacharcolon}{\isacharequal}\ None{\isacharparenright}{\isasymrparr}{\isacharparenright}}}
\isanewline
\quad \quad \quad \isacommand{IF}\ {\isasymacute}lvl\ t\ {\isachargreater}\ {\isadigit{0}}\ {\isasymand}\ partner{\isacharunderscore}bits\ {\isacharparenleft}{\isasymacute}mem{\isacharunderscore}pool{\isacharunderscore}info\ {\isacharparenleft}pool\ b{\isacharparenright}{\isacharparenright}\ {\isacharparenleft}{\isasymacute}lvl\ t{\isacharparenright}\ {\isacharparenleft}{\isasymacute}bn\ t{\isacharparenright}\ \isacommand{THEN}

\quad \quad \quad
\speccomment{(V3 {\isasymin} {\isasymlbrace}NULL\ {\isacharless}\ {\isasymacute}lvl\ t\ {\isasymand}\ partner{\isacharunderscore}bits\ {\isacharparenleft}{\isasymacute}mem{\isacharunderscore}pool{\isacharunderscore}info\ {\isacharparenleft}pool\ b{\isacharparenright}{\isacharparenright}\ {\isacharparenleft}{\isasymacute}lvl\ t{\isacharparenright}\ {\isacharparenleft}{\isasymacute}bn\ t{\isacharparenright}{\isasymrbrace})}

\quad \quad \specrg{
\isacommand{mergeblock{\isacharunderscore}loopinv}\ V3 \ t\ b\ {\isasymalpha}\ {\isasymequiv} \isanewline
\quad \quad {\isacharbraceleft}V{\isachardot}\ let\ minf{\isadigit{0}}\ {\isacharequal}\ {\isacharparenleft}mem{\isacharunderscore}pool{\isacharunderscore}info\ V3{\isacharparenright}{\isacharparenleft}pool\ b{\isacharparenright}{\isacharsemicolon}
\ lvl{\isadigit{0}}\ {\isacharequal}\ {\isacharparenleft}levels\ minf{\isadigit{0}}{\isacharparenright}\ {\isacharbang}\ {\isacharparenleft}lvl\ V3\ t{\isacharparenright}{\isacharsemicolon}\isanewline
\quad \quad \quad \quad \quad minf{\isadigit{1}}\ {\isacharequal}\ {\isacharparenleft}mem{\isacharunderscore}pool{\isacharunderscore}info\ V{\isacharparenright}{\isacharparenleft}pool\ b{\isacharparenright}{\isacharsemicolon}
\ lvl{\isadigit{1}}\ {\isacharequal}\ {\isacharparenleft}levels\ minf{\isadigit{1}}{\isacharparenright}\ {\isacharbang}\ {\isacharparenleft}lvl\ V3\ t{\isacharparenright}\ in\ \isanewline
\quad \quad \quad \quad {\isacharparenleft}bits\ lvl{\isadigit{1}}\ {\isacharequal}\ list{\isacharunderscore}updates{\isacharunderscore}n\ {\isacharparenleft}bits\ lvl{\isadigit{0}}{\isacharparenright}\ {\isacharparenleft}{\isacharparenleft}bn\ V3\ t\ div\ {\isadigit{4}}{\isacharparenright}\ {\isacharasterisk}\ {\isadigit{4}}{\isacharparenright}\ {\isacharparenleft}i\ V\ t{\isacharparenright}\ NOEXIST{\isacharparenright}\isanewline
\quad \quad \quad \quad {\isasymand}\ {\isacharparenleft}free{\isacharunderscore}list\ lvl{\isadigit{1}}\ {\isacharequal}\ removes\ {\isacharparenleft}map\ {\isacharparenleft}{\isasymlambda}ii{\isachardot}\ block{\isacharunderscore}ptr\ minf{\isadigit{0}}\ {\isacharparenleft}lsz\ V3\ t{\isacharparenright}\isanewline
\quad \quad \quad \quad \quad \quad \quad \quad \quad \quad \quad \quad \quad {\isacharparenleft}{\isacharparenleft}bn\ V3\ t\ div\ {\isadigit{4}}{\isacharparenright}\ {\isacharasterisk}\ {\isadigit{4}}\ {\isacharplus}\ ii{\isacharparenright}{\isacharparenright}\ {\isacharbrackleft}{\isadigit{0}}{\isachardot}{\isachardot}{\isacharless}{\isacharparenleft}i\ V\ t{\isacharparenright}{\isacharbrackright}{\isacharparenright}\ {\isacharparenleft}free{\isacharunderscore}list\ lvl{\isadigit{0}}{\isacharparenright}{\isacharparenright}

\quad \quad \quad \quad {\isasymand}\ {\isacharparenleft}wait{\isacharunderscore}q\ minf{\isadigit{0}}\ {\isacharequal}\ wait{\isacharunderscore}q\ minf{\isadigit{1}}{\isacharparenright}\ {\isasymand}\ {\isacharparenleft}{\isasymforall}t{\isacharprime}{\isachardot}\ t{\isacharprime}\ {\isasymnoteq}\ t\ {\isasymlongrightarrow}\ lvars{\isacharunderscore}nochange\ t{\isacharprime}\ V\ V3{\isacharparenright}

\quad \quad \quad \quad {\isasymand}\ {\isacharparenleft}{\isasymforall}p{\isachardot}\ p\ {\isasymnoteq}\ pool\ b\ {\isasymlongrightarrow}\ mem{\isacharunderscore}pool{\isacharunderscore}info\ V\ p\ {\isacharequal}\ mem{\isacharunderscore}pool{\isacharunderscore}info\ V3\ p{\isacharparenright}\isanewline
\quad \quad \quad \quad {\isasymand}\ {\isacharparenleft}{\isasymforall}j{\isachardot}\ j\ {\isasymnoteq}\ lvl\ V3\ t\ {\isasymlongrightarrow}\ {\isacharparenleft}levels\ minf{\isadigit{0}}{\isacharparenright}{\isacharbang}j\ {\isacharequal}\ {\isacharparenleft}levels\ minf{\isadigit{1}}{\isacharparenright}{\isacharbang}j{\isacharparenright}

\quad \quad \quad \quad {\isasymand}\ {\isacharparenleft}V{\isacharcomma}V3{\isacharparenright}{\isasymin}gvars{\isacharunderscore}conf{\isacharunderscore}stable\ {\isasymand} \ i\ V\ t\ {\isasymle}\ \isadigit{4} \ {\isasymand} \ {\isasymand} \ {\isasymalpha} = 4 - i\ V\ t\ ...... {\isacharbraceright}
}

\quad \quad \quad \quad \isacommand{FOR}\ {\isasymacute}i\ {\isacharcolon}{\isacharequal}\ {\isasymacute}i{\isacharparenleft}t\ {\isacharcolon}{\isacharequal}\ {\isadigit{0}}{\isacharparenright}{\isacharsemicolon}\ {\isasymacute}i\ t\ {\isacharless}\ {\isadigit{4}}{\isacharsemicolon}\ {\isasymacute}i\ {\isacharcolon}{\isacharequal}\ {\isasymacute}i{\isacharparenleft}t\ {\isacharcolon}{\isacharequal}\ {\isasymacute}i\ t\ {\isacharplus}\ {\isadigit{1}}{\isacharparenright}\ \isacommand{DO}

\quad \quad \quad \quad \specrg{
\isacommand{mergeblock{\isacharunderscore}loopinv}\ V3 \ t\ b\ {\isasymalpha}\ {\isasyminter}\ {\isasymlbrace} {\isasymalpha} {\isachargreater} 0 {\isasymrbrace}
}

\quad \quad \quad \quad \quad 
\specrg{\{V4\} \speccomment{{\isacharparenleft}V4 {\isasymin} mergeblock{\isacharunderscore}loopinv\ V3 \ t\ b\ {\isasymalpha}\ {\isasyminter}\ {\isasymlbrace} {\isasymalpha} {\isachargreater} 0 {\isasymrbrace} {\isacharparenright}}}

\quad \quad \quad \quad \quad {\isasymacute}bb\ {\isacharcolon}{\isacharequal}\ {\isasymacute}bb\ {\isacharparenleft}t\ {\isacharcolon}{\isacharequal}\ {\isacharparenleft}{\isasymacute}bn\ t\ div\ {\isadigit{4}}{\isacharparenright}\ {\isacharasterisk}\ {\isadigit{4}}\ {\isacharplus}\ {\isasymacute}i\ t{\isacharparenright}{\isacharsemicolon}{\isacharsemicolon}

\quad \quad \quad \quad \quad 
\specrg{\{V5\} \speccomment{{\isacharparenleft}V5 {\isasymequiv}\ V4{\isasymlparr}bb\ {\isacharcolon}{\isacharequal}\ {\isacharparenleft}bb\ V{\isacharparenright}\ {\isacharparenleft}t{\isacharcolon}{\isacharequal}{\isacharparenleft}bn\ V4\ t\ div\ {\isadigit{4}}{\isacharparenright}\ {\isacharasterisk}\ {\isadigit{4}}\ {\isacharplus}\ i\ V4\ t{\isacharparenright}{\isasymrparr} {\isacharparenright}}}

\quad \quad \quad \quad \quad {\isasymacute}mem{\isacharunderscore}pool{\isacharunderscore}info\ {\isacharcolon}{\isacharequal}\ set{\isacharunderscore}bit{\isacharunderscore}noexist\ {\isasymacute}mem{\isacharunderscore}pool{\isacharunderscore}info\ {\isacharparenleft}pool\ b{\isacharparenright}\ {\isacharparenleft}{\isasymacute}lvl\ t{\isacharparenright}\ {\isacharparenleft}{\isasymacute}bb\ t{\isacharparenright}{\isacharsemicolon}{\isacharsemicolon}

\quad \quad \quad \quad \quad 
\specrg{\{V6\} \speccomment{{\isacharparenleft}V6 {\isasymequiv}\ V5{\isasymlparr} mem{\isacharunderscore}pool{\isacharunderscore}info\ {\isacharcolon}{\isacharequal}\isanewline
\quad \quad \quad \quad \quad \quad \quad \quad \quad \quad \quad 
set{\isacharunderscore}bit{\isacharunderscore}noexist\ (mem{\isacharunderscore}pool{\isacharunderscore}info\ V5)\ (pool\ b)\ (lvl\ V5\ t)\ (bb\ V5\ t) {\isasymrparr} {\isacharparenright}}}

\quad \quad \quad \quad \quad {\isasymacute}block{\isacharunderscore}pt\ {\isacharcolon}{\isacharequal}\ {\isasymacute}block{\isacharunderscore}pt\ {\isacharparenleft}t\ {\isacharcolon}{\isacharequal}\ block{\isacharunderscore}ptr\ {\isacharparenleft}{\isasymacute}mem{\isacharunderscore}pool{\isacharunderscore}info\ {\isacharparenleft}pool\ b{\isacharparenright}{\isacharparenright}\ {\isacharparenleft}{\isasymacute}lsz\ t{\isacharparenright}\ {\isacharparenleft}{\isasymacute}bb\ t{\isacharparenright}{\isacharparenright}{\isacharsemicolon}{\isacharsemicolon}

\quad \quad \quad \quad \quad 
\specrg{\{V7\} \speccomment{{\isacharparenleft}V7 {\isasymequiv}\ V6{\isasymlparr}block{\isacharunderscore}pt\ {\isacharcolon}{\isacharequal}\ {\isacharparenleft}block{\isacharunderscore}pt\ V6{\isacharparenright}\isanewline
\quad \quad \quad \quad \quad \quad \quad \quad \quad \quad  {\isacharparenleft}t{\isacharcolon}{\isacharequal}block{\isacharunderscore}ptr\ {\isacharparenleft}mem{\isacharunderscore}pool{\isacharunderscore}info\ V6\ {\isacharparenleft}pool\ b{\isacharparenright}{\isacharparenright}\ {\isacharparenleft}lsz\ V6\ t{\isacharparenright}\ {\isacharparenleft}bb\ V6\ t{\isacharparenright}{\isacharparenright}{\isasymrparr} {\isacharparenright}}}

\quad \quad \quad \quad \quad \isacommand{IF}\ {\isasymacute}bn\ t\ {\isasymnoteq}\ {\isasymacute}bb\ t\ {\isasymand}\ block{\isacharunderscore}fits\ {\isacharparenleft}{\isasymacute}mem{\isacharunderscore}pool{\isacharunderscore}info\ {\isacharparenleft}pool\ b{\isacharparenright}{\isacharparenright}\ {\isacharparenleft}{\isasymacute}block{\isacharunderscore}pt\ t{\isacharparenright}\ {\isacharparenleft}{\isasymacute}lsz\ t{\isacharparenright}\ \isacommand{THEN}\isanewline
\quad \quad \quad \quad \quad \quad {\isasymacute}mem{\isacharunderscore}pool{\isacharunderscore}info\ {\isacharcolon}{\isacharequal}\ {\isasymacute}mem{\isacharunderscore}pool{\isacharunderscore}info\ {\isacharparenleft}{\isacharparenleft}pool\ b{\isacharparenright}\ {\isacharcolon}{\isacharequal}\ \isanewline
\quad \quad \quad \quad \quad \quad \quad \quad remove{\isacharunderscore}free{\isacharunderscore}list\ {\isacharparenleft}{\isasymacute}mem{\isacharunderscore}pool{\isacharunderscore}info\ {\isacharparenleft}pool\ b{\isacharparenright}{\isacharparenright}\ {\isacharparenleft}{\isasymacute}lvl\ t{\isacharparenright}\ {\isacharparenleft}{\isasymacute}block{\isacharunderscore}pt\ t{\isacharparenright}{\isacharparenright}\isanewline
\quad \quad \quad \quad \quad \isacommand{FI}\isanewline
\quad \quad \quad \quad \isacommand{ROF}{\isacharsemicolon}{\isacharsemicolon}

\quad \quad \quad \quad \specrg{
\isacommand{mergeblock{\isacharunderscore}loopinv}\ V3 \ t\ b\ {\isasymalpha}\ {\isasyminter}\ {\isasymlbrace} {\isasymalpha} = 0 {\isasymrbrace}
}

\quad \quad \quad \quad {\isasymacute}lvl\ {\isacharcolon}{\isacharequal}\ {\isasymacute}lvl\ {\isacharparenleft}t\ {\isacharcolon}{\isacharequal}\ {\isasymacute}lvl\ t\ {\isacharminus}\ {\isadigit{1}}{\isacharparenright}{\isacharsemicolon}{\isacharsemicolon}\isanewline
\quad \quad \quad \quad {\isasymacute}bn\ {\isacharcolon}{\isacharequal}\ {\isasymacute}bn\ {\isacharparenleft}t\ {\isacharcolon}{\isacharequal}\ {\isasymacute}bn\ t\ div\ {\isadigit{4}}{\isacharparenright}{\isacharsemicolon}{\isacharsemicolon}\isanewline
\quad \quad \quad \quad {\isasymacute}mem{\isacharunderscore}pool{\isacharunderscore}info\ {\isacharcolon}{\isacharequal}\ set{\isacharunderscore}bit{\isacharunderscore}freeing\ {\isasymacute}mem{\isacharunderscore}pool{\isacharunderscore}info\ {\isacharparenleft}pool\ b{\isacharparenright}\ {\isacharparenleft}{\isasymacute}lvl\ t{\isacharparenright}\ {\isacharparenleft}{\isasymacute}bn\ t{\isacharparenright}{\isacharsemicolon}{\isacharsemicolon}\isanewline
\quad \quad \quad \quad {\isasymacute}freeing{\isacharunderscore}node\ {\isacharcolon}{\isacharequal}\ {\isasymacute}freeing{\isacharunderscore}node\ {\isacharparenleft}t\ {\isacharcolon}{\isacharequal}\ Some\ {\isasymlparr}pool\ {\isacharequal}\ {\isacharparenleft}pool\ b{\isacharparenright}{\isacharcomma}\ level\ {\isacharequal}\ {\isacharparenleft}{\isasymacute}lvl\ t{\isacharparenright}{\isacharcomma}\ \isanewline
\quad \ \ \ \ \ \ \ \ \ \ \ \ \ \ \ \ \ \ \ \ block\ {\isacharequal}\ {\isacharparenleft}{\isasymacute}bn\ t{\isacharparenright}{\isacharcomma}\ 
data\ {\isacharequal}\ block{\isacharunderscore}ptr\ {\isacharparenleft}{\isasymacute}mem{\isacharunderscore}pool{\isacharunderscore}info\ {\isacharparenleft}pool\ b{\isacharparenright}{\isacharparenright}\ \isanewline
\quad \ \ \ \ \ \ \ \ \ \ \ \ \ \ \ \ \ \ \ \ \ \  {\isacharparenleft}{\isacharparenleft}{\isacharparenleft}ALIGN{\isadigit{4}}\ {\isacharparenleft}max{\isacharunderscore}sz\ {\isacharparenleft}{\isasymacute}mem{\isacharunderscore}pool{\isacharunderscore}info\ {\isacharparenleft}pool\ b{\isacharparenright}{\isacharparenright}{\isacharparenright}{\isacharparenright}\ div\ {\isacharparenleft}{\isadigit{4}}\ {\isacharcircum}\ {\isacharparenleft}{\isasymacute}lvl\ t{\isacharparenright}{\isacharparenright}{\isacharparenright}{\isacharparenright}\ 
{\isacharparenleft}{\isasymacute}bn\ t{\isacharparenright}\ {\isasymrparr}{\isacharparenright}\isanewline
\quad \quad \quad \isacommand{ELSE}

\quad \quad \quad \quad \specrg{\{V3\}\ {\isasyminter}\ {\isacharminus}\ {\isasymlbrace}NULL\ {\isacharless}\ {\isasymacute}lvl\ t\ {\isasymand}\ partner{\isacharunderscore}bits\ {\isacharparenleft}{\isasymacute}mem{\isacharunderscore}pool{\isacharunderscore}info\ {\isacharparenleft}pool\ b{\isacharparenright}{\isacharparenright}\ {\isacharparenleft}{\isasymacute}lvl\ t{\isacharparenright}\ {\isacharparenleft}{\isasymacute}bn\ t{\isacharparenright}{\isasymrbrace}}

\quad \quad \quad \quad \isacommand{IF}\ block{\isacharunderscore}fits\ {\isacharparenleft}{\isasymacute}mem{\isacharunderscore}pool{\isacharunderscore}info\ {\isacharparenleft}pool\ b{\isacharparenright}{\isacharparenright}\ {\isacharparenleft}{\isasymacute}blk\ t{\isacharparenright}\ {\isacharparenleft}{\isasymacute}lsz\ t{\isacharparenright}\ \isacommand{THEN}\isanewline
\quad \quad \quad \quad \quad {\isasymacute}mem{\isacharunderscore}pool{\isacharunderscore}info\ {\isacharcolon}{\isacharequal}\ {\isasymacute}mem{\isacharunderscore}pool{\isacharunderscore}info\ {\isacharparenleft}{\isacharparenleft}pool\ b{\isacharparenright}\ {\isacharcolon}{\isacharequal}\ \isanewline
\quad \quad \quad \quad \quad \quad append{\isacharunderscore}free{\isacharunderscore}list\ {\isacharparenleft}{\isasymacute}mem{\isacharunderscore}pool{\isacharunderscore}info\ {\isacharparenleft}pool\ b{\isacharparenright}{\isacharparenright}\ {\isacharparenleft}{\isasymacute}lvl\ t{\isacharparenright}\ {\isacharparenleft}{\isasymacute}blk\ t{\isacharparenright}\ {\isacharparenright}\isanewline
\quad \quad \quad \quad \isacommand{FI}{\isacharsemicolon}{\isacharsemicolon}\isanewline
\quad \quad \quad \quad {\isasymacute}free{\isacharunderscore}block{\isacharunderscore}r\ {\isacharcolon}{\isacharequal}\ {\isasymacute}free{\isacharunderscore}block{\isacharunderscore}r\ {\isacharparenleft}t\ {\isacharcolon}{\isacharequal}\ False{\isacharparenright}\isanewline
\quad \quad \quad \isacommand{FI}\isanewline
\quad \quad \isacommand{END} \ \speccomment{(* END\ of\ ATOM *)} \isanewline
\quad \isacommand{OD} \ \speccomment{(* END\ of\ WHILE \ free\_block\_r \ DO  *)} 

\quad \specrg{\isacommand{mp{\isacharunderscore}free{\isacharunderscore}precond{\isadigit{9}}}\ t\ b\ {\isasymequiv}\ Mem{\isacharunderscore}pool{\isacharunderscore}free{\isacharunderscore}pre\ t\ {\isasyminter}\ {\isasymlbrace} g {\isasymrbrace}}

\quad \speccomment{{\isacharparenleft}{\isacharasterisk}\ {\isacharequal} {\isacharequal} {\isacharequal} end of {\isacharcolon}\ free{\isacharunderscore}block{\isacharparenleft}pool{\isacharcomma}\ level{\isacharcomma}\ lsizes{\isacharcomma}\ block{\isacharparenright}{\isacharsemicolon}\ {\isacharequal} {\isacharequal} {\isacharequal}{\isacharasterisk}{\isacharparenright}}

\quad t\ \isactrlenum \ \isacommand{ATOMIC}

\quad \specrg{\{Va\} {\isacharparenleft}Va {\isasymin} mp{\isacharunderscore}free{\isacharunderscore}precond{\isadigit{9}} \ t\ b\ {\isasyminter} {\isasymlbrace}{\isasymacute}cur\ {\isacharequal}\ Some\ t{\isasymrbrace}{\isacharparenright}}

\quad  \specrg{
stm{\isadigit{9}}{\isacharunderscore}loopinv\ Va\ t\ b\ {\isasymalpha}\ {\isasymequiv}\isanewline
\quad \quad {\isacharbraceleft}V{\isachardot}\ inv\ V\ {\isasymand}\ cur\ V\ {\isacharequal}\ cur\ Va\ {\isasymand}\ tick\ V\ {\isacharequal}\ tick\ Va\ {\isasymand}\ {\isacharparenleft}V{\isacharcomma}Va{\isacharparenright}{\isasymin}gvars{\isacharunderscore}conf{\isacharunderscore}stable\ \isanewline
\quad \quad \quad {\isasymand}\ freeing{\isacharunderscore}node\ V\ t\ {\isacharequal}\ freeing{\isacharunderscore}node\ Va\ t\ {\isasymand}\ allocating{\isacharunderscore}node\ V\ t\ {\isacharequal}\ allocating{\isacharunderscore}node\ Va\ t\isanewline
\quad \quad \quad {\isasymand}\ {\isacharparenleft}{\isasymforall}p{\isachardot}\ levels\ {\isacharparenleft}mem{\isacharunderscore}pool{\isacharunderscore}info\ V\ p{\isacharparenright}\ {\isacharequal}\ levels\ {\isacharparenleft}mem{\isacharunderscore}pool{\isacharunderscore}info\ Va\ p{\isacharparenright}{\isacharparenright}\isanewline
\quad \quad \quad {\isasymand}\ {\isacharparenleft}{\isasymforall}p{\isachardot}\ p\ {\isasymnoteq}\ pool\ b\ {\isasymlongrightarrow}\ mem{\isacharunderscore}pool{\isacharunderscore}info\ V\ p\ {\isacharequal}\ mem{\isacharunderscore}pool{\isacharunderscore}info\ Va\ p{\isacharparenright}\isanewline
\quad \quad \quad {\isasymand}\ {\isacharparenleft}{\isasymforall}t{\isacharprime}{\isachardot}\ t{\isacharprime}\ {\isasymnoteq}\ t\ {\isasymlongrightarrow}\ lvars{\isacharunderscore}nochange\ t{\isacharprime}\ V\ Va{\isacharparenright}\isanewline
\quad \quad \quad {\isasymand} {\isasymalpha}\ {\isacharequal} length\ {\isacharparenleft}wait{\isacharunderscore}q\ {\isacharparenleft}{\isasymacute}mem{\isacharunderscore}pool{\isacharunderscore}info\ {\isacharparenleft}pool\ b{\isacharparenright}{\isacharparenright}{\isacharparenright}
{\isacharbraceright}
}

\quad \quad \isacommand{WHILE}\ wait{\isacharunderscore}q\ {\isacharparenleft}{\isasymacute}mem{\isacharunderscore}pool{\isacharunderscore}info\ {\isacharparenleft}pool\ b{\isacharparenright}{\isacharparenright}\ {\isasymnoteq}\ {\isacharbrackleft}{\isacharbrackright}\ \isacommand{DO}\

\quad \quad \quad \specrg{
stm{\isadigit{9}}{\isacharunderscore}loopinv\ Va\ t\ b\ {\isasymalpha}\ {\isasyminter}\ {\isasymlbrace} {\isasymalpha} {\isachargreater} 0 {\isasymrbrace}
}

\quad \quad \quad {\isasymacute}th\ {\isacharcolon}{\isacharequal}\ {\isasymacute}th\ {\isacharparenleft}t\ {\isacharcolon}{\isacharequal}\ hd\ {\isacharparenleft}wait{\isacharunderscore}q\ {\isacharparenleft}{\isasymacute}mem{\isacharunderscore}pool{\isacharunderscore}info\ {\isacharparenleft}pool\ b{\isacharparenright}{\isacharparenright}{\isacharparenright}{\isacharparenright}{\isacharsemicolon}{\isacharsemicolon}

\quad \quad \quad \speccomment{{\isacharparenleft}{\isacharasterisk}\ {\isacharunderscore}unpend{\isacharunderscore}thread{\isacharparenleft}th{\isacharparenright}{\isacharsemicolon}\ {\isacharasterisk}{\isacharparenright}}

\quad \quad \quad {\isasymacute}mem{\isacharunderscore}pool{\isacharunderscore}info\ {\isacharcolon}{\isacharequal}\ {\isasymacute}mem{\isacharunderscore}pool{\isacharunderscore}info\ {\isacharparenleft}pool\ b\ {\isacharcolon}{\isacharequal}\ {\isasymacute}mem{\isacharunderscore}pool{\isacharunderscore}info\ {\isacharparenleft}pool\ b{\isacharparenright}

\quad \quad \quad \quad \quad {\isasymlparr}wait{\isacharunderscore}q\ {\isacharcolon}{\isacharequal}\ tl\ {\isacharparenleft}wait{\isacharunderscore}q\ {\isacharparenleft}{\isasymacute}mem{\isacharunderscore}pool{\isacharunderscore}info\ {\isacharparenleft}pool\ b{\isacharparenright}{\isacharparenright}{\isacharparenright}{\isasymrparr}{\isacharparenright}{\isacharsemicolon}{\isacharsemicolon}

\quad \quad \quad \speccomment{{\isacharparenleft}{\isacharasterisk}\ {\isacharunderscore}ready{\isacharunderscore}thread{\isacharparenleft}th{\isacharparenright}{\isacharsemicolon}\ {\isacharasterisk}{\isacharparenright}}

\quad \quad \quad {\isasymacute}thd{\isacharunderscore}state\ {\isacharcolon}{\isacharequal}\ {\isasymacute}thd{\isacharunderscore}state\ {\isacharparenleft}{\isasymacute}th\ t\ {\isacharcolon}{\isacharequal}\ READY{\isacharparenright}{\isacharsemicolon}{\isacharsemicolon}

\quad \quad \quad {\isasymacute}need{\isacharunderscore}resched\ {\isacharcolon}{\isacharequal}\ {\isasymacute}need{\isacharunderscore}resched{\isacharparenleft}t\ {\isacharcolon}{\isacharequal}\ True{\isacharparenright}

\quad \quad  \isacommand{OD}{\isacharsemicolon}{\isacharsemicolon}

\quad \quad \specrg{
stm{\isadigit{9}}{\isacharunderscore}loopinv\ Va\ t\ b\ {\isasymalpha}\ {\isasyminter}\ {\isasymlbrace} {\isasymalpha} = 0 {\isasymrbrace}
}

\quad \quad \isacommand{IF}\ {\isasymacute}need{\isacharunderscore}resched\ t\ \isacommand{THEN}\isanewline
\quad \quad \quad reschedule \quad \speccomment{(* \_reschedule\_threads(key) *)}\isanewline
\quad \quad \isacommand{FI}\isanewline
\quad \isacommand{END} \ \speccomment{(* END\ of\ ATOM *)}\isanewline
\isacommand{END}{\isachardoublequoteclose}

\specrg{\isacommand{Mem{\isacharunderscore}pool{\isacharunderscore}free{\isacharunderscore}post}\ t\ {\isasymequiv}\ {\isasymlbrace} {\isasymacute}inv\ {\isasymand}\ {\isasymacute}allocating{\isacharunderscore}node\ t\ {\isacharequal}\ None\ {\isasymand}\ {\isasymacute}freeing{\isacharunderscore}node\ t\ {\isacharequal}\ None{\isasymrbrace}}

\end{isabellec}

\end{document}